\documentclass[10pt, oneside]{article}      
\usepackage[hmargin=3cm, vmargin=3cm]{geometry}   
\usepackage{graphicx}			  % Use pdf, png, jpg, or epsÃÂÃÂÃÂÃÂ§ with pdflatex; use eps in DVI mode
\usepackage{setspace}
\usepackage{amssymb}
\usepackage{amsfonts}
\usepackage{amsmath}
\usepackage{amsthm} 
\usepackage{blindtext}
\usepackage{nicefrac}
\usepackage{float}
\usepackage{indentfirst}
\usepackage{algorithm}
\usepackage{color}
\usepackage{multirow}
\usepackage{tabularx}
\usepackage{lscape}
\usepackage{subcaption}
\usepackage{placeins}
\usepackage{diagbox}
\usepackage[title]{appendix}
\usepackage[graphicx]{realboxes}
\usepackage{rotating}
\usepackage{lmodern}
\usepackage{slantsc}
\usepackage{soul,xcolor}

\usepackage[round]{natbib}
\allowdisplaybreaks    %  allow math mode equations to span over multiple pages
\usepackage{longtable} %  allow table to span over multiple pages

\numberwithin{equation}{section}
\numberwithin{figure}{section}
\numberwithin{table}{section}

\newtheorem{lemma}{Lemma}[section]
\newtheorem{proposition}{Proposition}[section]

\newtheorem{remark}{Remark}[section]

\newtheorem{definition}{Definition}[section]

\usepackage[hidelinks, bookmarks=true, bookmarksopen=true]{hyperref}
\hypersetup{bookmarksnumbered}

\usepackage[multiple]{footmisc}

\title{Statistical Arbitrage for Multiple Co-Integrated Stocks}

\author{T. N. L\textsc{\v{i}}\thanks{Department of Mathematics, New York University, 251 Mercer Street, New York, NY 10012. {\em nl747@cims.nyu.edu}}~~and
A. Papanicolaou\thanks{Department of Mathematics, North Carolina State University, 2311 Stinson Drive, Raleigh, NC 27695. {\em apapani@ncsu.edu}. This work was partially supported by NSF grant DMS-1907518.}}

\begin{document}
\maketitle
\setstcolor{red}

%%%%%%%%%%%%%%%%%%%%%%%%%%%%%
\begin{abstract}
In this article, we analyse optimal statistical arbitrage strategies from stochastic optimal control problems with multiple co-integrated stocks with eigenportfolios being factors. Optimal portfolio weights are found by solving Hamilton-Jacobi-Bellman (HJB) partial differential equations (PDEs), to which we can introduce a constraint for market neutrality. Our analyses demonstrate sufficient conditions on the model parameters to ensure long-term stability of the HJB solutions and stable growth rates for the optimal portfolios. To gauge how these optimal portfolios perform in practice, we perform backtests on historical prices of the S\&P 500 constituents from year 2000 through year 2021. These backtests suggest two key conclusions: that the optimal portfolios are sensitive to parameter estimation, and that the statistical arbitrage strategies are more profitable in periods when overall market volatility is high.
\newline

\newline

\textbf{AMS Subject Codes:} 62P05, 91B28, 93E20. 
%62P05 Applications to actuarial sciences and financial mathematics
%91B28 Finance, portfolios, investment
%93E20 Optimal stochastic control
\end{abstract}

%%%%%%%%%%%%%%%%%%%%%%%%%%%%%%%%%%%%%%%%%%%%%%%%%%
\section{Introduction}

Statistical arbitrage strategies involve trading among pairs of assets having co-integration. The essential idea is that a pair of co-integrated asset prices have a difference that is mean reverting. This mean-reverting difference is referred to as a spread. For a trader with the ability to sell short and utilise leverage, a possible strategy is to long the cheaper asset, short the expensive asset, and then wait for the spread to converge, at which point the positions can be closed for a profit. This is an example of statistical arbitrage because, while it may seem like a sure profit, there is not any finite time by when a spread would be almost surely converged. Instead there is only a high probability of the spread converging before reaching a fixed and finite investment horizon. 

The model we consider is proposed in \citet{AvellanedaLee2010} and is implemented in \citet{YeoPapanicolaou2017}. The model considers stocks whose total returns have co-integration with the total returns of a set of factors. The factors can be any selection of explanatory variables, traded or untraded, which have explanatory power in the cross sections of stock returns. The spreads are defined as the residuals that are obtained after regressing a stock's total returns onto the factor returns. A stock has co-integration if its spread is a stationary process. We determine whether a spread has stationarity or not by checking that its time series rejects a unit-root hypothesis, which is the same idea as that utilised in the test of \citet{EngleGranger1987}. 

The factors that we utilise are eigenportfolios, which are the orthogonal portfolios constructed from the correlation matrix of stocks returns. Eigenportfolios are an effective factor construction because they are orthogonal, and so the addition of each factor adds another independent variable to the regression. An alternative factor choice is to regress onto exchange traded funds (ETFs), but ETFs were not as prevalent in the early 2000s, which means that long-term backtesting requires synthesising of ETFs. In \citet{AvellanedaLee2010} they find better performance utilising eigenportfolios instead of historical sector ETFs data or synthetic sector ETFs. Trading of eigenportfolios can incur heavy transaction costs because they contain hundreds of stocks but only a few dozen that have co-integration. Therefore, trading signals with lower transaction costs treat the eigenportfolios as untradeable factors and only take positions in the co-integrated stocks. This is a feature of the model that adds some generality, as there are other factors such as the illiquidity returns of \cite{PastorStambaugh2003} that lack tradeability. 

Untradeable factors cause the market to be incomplete, which is in contrast to the complete markets considered in \citet{chiu2013optimal} and \citet{ma2019optimal}. An advantage to using the factor model of \citet{AvellanedaLee2010} is the simplification of parameter estimation, as the factor model allows for drift parameters to be estimated marginally for each stock. Comparatively, it is more challenging to estimate the matrices of drift parameters for generalised models that search for co-integrated linear combinations among the space of traded stocks.

The contribution of this article is the analyses and implementation of solutions to the Hamilton-Jacobi-Bellman equations arising from the model of \citet{AvellanedaLee2010} with non-tradeability of factors. We formulate stochastic optimal control problems where the state variables are the values of the trading portfolios and the spreads
\begin{align*}
Z_{t}^{i}=\int_0^t\frac{dS_u^i}{S_u^i}-\sum_{j=1}^m\beta^{ij}\int_0^t\frac{dF_u^j}{F_u^j}-\alpha^{i}t\,,
\end{align*}
where $S_{t}^{i}$ is the price of the $i^{\mathrm{th}}$ stock for $i=1,\:2,\:\cdots,\:d\in\mathbb{N}$, $F_{t}^{j}$ is the $j^{\mathrm{th}}$ eigenportfolio for $j=1,\:2,\:\cdots,\:m\in\mathbb{N}$, and the pair $\left(\alpha^{i},\:\beta^{ij}\right)$ are the regression coefficients returned by the statistical test for making $Z_{t}^{i}$ stationary. As is the case in \citet{MudchanatongsukPW2008} and \citet{LiTourin2016}, these spreads form a stationary vector Ornstein-Uhlenbeck (OU) process.

The optimal portfolios are the solutions to HJB equations, which in the case of power utility we are able to reduce to systems of ordinary differential equations (ODEs) that include a matrix Riccati equation and a pair of linear equations. We perform long-term stability analyses of these ODEs, which give us indications of soundness of the models for financial applications. In particular, finiteness of the ODEs for any finite-time investment horizon indicates that the models have an absence of arbitrage; if there was an arbitrage then the ODEs would have singularities. In addition, if the solutions of the HJB equations converge to steady states as the investment horizon tends toward infinity, then there are long-term statistical arbitrage portfolios that earn positive profits with probability close to one. 

We analyse HJB equations for an unconstrained portfolio and for another portfolio constrained for market neutrality. Market neutrality is important when doing statistical arbitrage with a factor model, as it immunises the portfolio against market fluctuations. Basic pairs trading strategies where the spreads are directly tradeable have this immunisation built in, see \citet{angoshtari2014}. However, in our case, because the factors are not tradeable, the optimisation needs to be constrained in order to have a market-neutral portfolio.

We implement the optimal portfolios in backtests on historical stock-price data. These studies give us a practical sense for profitability of the optimal strategies. We take the S\&P 500 constituents from year 2000 through year 2021, and then look at the profits, expected returns, volatilities, Sharpe ratios, and maximum drawdowns for out-of-sample portfolios computed with varying estimation windows, both with and without a market neutrality constraint. From these studies, we arrive at two main conclusions about statistical arbitrage strategies: first, these strategies are sensitive to parameter estimation, and second, these strategies have greater potential to out-perform the benchmark during periods of higher overall market volatility. Sensitivity to parameter estimation is in line with the backtesting studies in \citet{YeoPapanicolaou2017} where they show the variation in Sharpe ratios relative to estimation windows and stock selections.

%%%%%%%%%%%%%%%%%%%%%%%%%%%%%%%%%%%%%%%%%%%%%%%%%%%%%%%%%%
\subsection{Overview of Related Work\label{subsec:1.1}}

A formal definition for statistical arbitrage is given in \citet{HoganJTW2004}. A test for co-integration of financial time series is designed in \citet{EngleGranger1987}, namely, the Engle-Granger test. An application of \citet{EngleGranger1987} for co-integration based trading strategies are shown in \citet{Vidyamurthy2004}, trading of co-integrated pairs alongside methods for filtering and parameter estimation to handle latency is studied in \citet{ElliottHoekMalcolm2005}, and an in-depth statistical analysis of the performance of pairs trading strategies is done in \citet{GatevGR2006}. Principal component analysis of large number of stocks co-integrated through common factors is the topic in \citet{AvellanedaLee2010}, and is the basis for the models in this article; empirical testing of pairs trading, including out-of-sample experiments with changing parameters, is done in \citet{GalenkoPopova2012} and in \citet{YeoPapanicolaou2017}. Analysis showing significance of short-term reversal and momentum factors on returns of pairs trading portfolios is presented in \citet{ChenCCL2017}.

Stochastic optimal control for pairs trading with OU spreads is studied in \citet{MudchanatongsukPW2008}. A stochastic control for optimal trading of co-integrated pairs is proposed and solved in \citet{tourinYan2013} and \citet{angoshtari2014}, and stochastic optimal control for pairs trading with a local-volatility model is analysed by \citet{LiTourin2016}. A multiple variate version of the stochastic control problem with power utility is the topic in \citet{chiu2013optimal} and \citet{ma2019optimal}, with analysis of the matrix Riccati equation being presented. Additionally, there is the long-term stability analysis for matrix Riccati equations of multiple-asset models done in \citet{DavisLleo2014}, and the matrix Riccati equations analysis for a single co-integrated pair with partial information in \citet{LeePapanicolaou2016}. Related work includes the optimal trading of spreads with transaction costs and stop-loss criterion, which are analysed in \citet{LeiXu2015} and \citet{LeungLi2015}. There are also the machine learning approaches to statistical arbitrage, such as reinforcement learning and boosting applied to co-integrated constituents in the S\&P 500 done by \citet{FallahpourHTR2016}. An HJB equation for an optimal portfolio constrained to be 100\% long is presented in \citet{jaimungal2018} with application toward comparing active and passive fund management, and is similar to the market neutral constraint we presented in this article.

%%%%%%%%%%%%%%%%%%%%%%%%%%%%%%%%%%%%%%%%%%%%%%%%%%%%%%%%%
\subsection{Structure of This Article\label{subsec:1.2}}

In this article, we propose and solve stochastic optimal control problems, and then analyse the solutions for the unconstrained portfolio and the market-neutral constrained portfolio. Later on in the article we perform some empirical studies on historical data. The organisation of the article is as follows: Section \ref{sec:model_construction} contains the definitions for the models along with analyses of the HJB equations, with Section \ref{subsec:unconstrained_hjb} presenting the solution of the HJB for the unconstrained portfolio via an exponential ansatz, with Section \ref{subsec:unconstrained_stability} presenting the stability analysis, and then with Section \ref{subsec:constrained_model} presenting the HJB and stability analysis for optimisation with market-neutral constraint; the empirical analyses of historical data come in Section \ref{sec:numerical_empirical} with construction of factors through eigenportfolios presented in Section \ref{subsec:eigenportfolio}, preliminary data analyses and parameter estimation presented in Section \ref{subsec:parameter_estimation}, and analyses of portfolio performance, for example Sharpe ratio and maximum drawndown, presented in Section \ref{subsec:portfolio_performance}; Section \ref{sec:conclusion} is the conclusion.

%%%%%%%%%%%%%%%%%%%%%%%%%%%%%%%%%%%%%%%%%%%%%%%%%%%%%%%%%%%%%%%%%%%%%%%%%%%%%%%%%%%%%%%%%%%%%
\section{Model Constructions and Optimisations\label{sec:model_construction}}

This section introduces the stochastic optimal control problems for multiple co-integrated stocks. We first build the stochastic system for stock prices, factors, and spreads in Section \ref{subsec:unconstrained_sde}. Then in Section \ref{subsec:unconstrained_hjb}, we propose a stochastic optimal control problem for unconstrained portfolio, and analyse the stability for its solution in Section \ref{subsec:unconstrained_stability}. Lastly, in Section \ref{subsec:constrained_model} we formulate a constrained portfolio with respect to market factor neutrality and complete the stability analysis for its solution.

%%%%%%%%%%%%%%%%%%%%%%%%%%%%%%%%%%%%%%%%%%%%%%%%%%%%%%%%%%%%%%%%%%%%%%%%%%
\subsection{Model with Co-Integration\label{subsec:unconstrained_sde}}

Suppose $S_{t}^{i}$, $i=1,\:2,\:\cdots,\:d\in\mathbb{N}$, is the stock price of the $i^{\mathrm{th}}$ individual firm in the financial market or the industrial sector. Let $F_{t}^{j}$ denote the value of the $j^{\mathrm{th}}$ factor for $j=1,\:2,\:\cdots,\:m\in\mathbb{N}$. The stocks are co-integrated with the factors if there is a stationary process $Z_{t}^{i}$ given by
\begin{align}
\frac{dS_t^i}{S_t^i}=\alpha^{i}dt+\sum_{j=1}^{m}\beta^{ij}\frac{dF_t^j}{F_t^j}+dZ_{t}^{i}\,,\label{eq:factor}
\end{align}
where $\alpha^{j}$ is a systemic return after controlling for the factor returns, and $\beta^{ij}$ is the $i^{\mathrm{th}}$ loading of stock on the $j^{\mathrm{th}}$ factor and is recorded in the matrix $\boldsymbol{\beta}\in\mathbb{R}^{d\times m}$. The total returns $\int_0^t\frac{dS_u^i}{S_u^i}$ and $\int_0^t\frac{dF_u^j}{F_u^j}$ are non-stationary, but there will be a statistical arbitrage strategy if $S^i$ is co-integrated with the factors. A stock $S^i$ is determined to be co-integrated with the factors if the process $Z_t^i$ rejects a unit-root hypothesis. If we know the stock is co-integrated, then we can further specify $Z_t^i$ to be a stationary Ornstein-Uhlenbeck process if we model the dynamics of $S_t^i$ in the same way as \citet{tourinYan2013}.

Let $\boldsymbol{B}_{t}=\left[B_{t}^{1},\:B_{t}^{2},\:\cdots,\:B_{t}^{d+m}\right]^{\top}\in\mathbb{R}^{d+m}$ denote a vector of independent Standard Brownian Motions (SBMs). The stochastic differential equation for the dynamics of a factor is
\begin{align}
\frac{dF_{t}^{j}}{F_{t}^{j}}=\eta^{j}dt+\sum_{k=1}^{m+d}\psi_{0}^{jk}dB_{t}^{k}\,,\label{eq:dF}
\end{align}
where $B_{t}^{k}$ is a component of the SBM vector $\boldsymbol{B}_{t}$, $\eta^{j}$ is a component of the factor drift coefficient vector $\boldsymbol{\eta}=\left[\eta^{1},\:\eta^{2},\:\cdots,\:\eta^{m}\right]^{\top}\in\mathbb{R}^{m}$, and $\psi_{0}^{jk}$ is a component of matrix $\mathbf{\Psi}_{0}\in\mathbb{R}^{m\times\left(d+m\right)}$, which forms a symmetric positive definite (SPD) diffusion matrix $\mathbf{\Sigma}_{0}=\mathbf{\Psi}_{0}\mathbf{\Psi}_{0}^{\top}$. The dynamics of an individual stock is
\begin{align}
\frac{dS_{t}^{i}}{S_{t}^{i}}=\left(\mu^{i}-\delta^{i}Z_{t}^{i}\right)dt+\sum_{k=1}^{m+d}\psi_{1}^{ik}dB_{t}^{k}\,,\label{eq:dSi}
\end{align}
where $\mu^{i}$ is the stock price drift coefficient, $\delta^{i}$ is the mean reversion speed of the process $Z_{t}^{i}$, and $\psi_{1}^{ik}$ is a component of matrix $\mathbf{\Psi}_{1}\in\mathbb{R}^{d\times\left(d+m\right)}$, which forms a SPD diffusion matrix $\mathbf{\Sigma}_{1}=\mathbf{\Psi}_{1}\mathbf{\Psi}_{1}^{\top}$. Utilising equation $\eqref{eq:dF}$ and equation $\eqref{eq:dSi}$, the co-integrated process $Z_{t}^{i}$ has the following stochastic differential equation (SDE):
\begin{align}
dZ_{t}^{i}=\delta^{i}\left(\theta^{i}-Z_{t}\right)dt-\sum_{j=1}^{m}\beta^{ij}\left(\mathbf{\Psi}_{0}d\boldsymbol{B}_{t}\right)^{j}+\left(\mathbf{\Psi}_{1}d\boldsymbol{B}_{t}\right)^i\,,\label{eq:dZ}
\end{align}
where $\theta^{i}=\left(-\alpha^{i}-\sum_{j=1}^{m}\beta^{ij}\eta^{j}+\mu^{i}\right)/\delta^{i}$ is the stationary mean, $\left(\mathbf{\Psi}_{0}d\boldsymbol{B}_{t}\right)^{j}$ is the $j^{\mathrm{th}}$ element of vector $\mathbf{\Psi}_{0}d\boldsymbol{B}_{t}$, and $\left(\mathbf{\Psi}_{1}d\boldsymbol{B}_{t}\right)^{i}$ is the $i^{\mathrm{th}}$ element of vector $\mathbf{\Psi}_{1}d\boldsymbol{B}_{t}$. Each and every $Z_{t}^{i}$ is a one dimensional stationary OU process if $\delta^{i}>0$. We also assume that there is a risk-free asset, such as a money market account, with interest rate $r\geq0$. The factor loadings in equation \eqref{eq:factor} are
\begin{align}
\label{eq:beta}
\boldsymbol{\beta}=%\left(\frac{dS_t}{S_t}\right)\left(\frac{dF_t}{F_t}\right)^\top \left(\left(\frac{dF_t}{F_t}\right)\left(\frac{dF_t}{F_t}\right)^\top \right)^{-1} = 
\mathbf{\Psi}_1\mathbf{\Psi}_0^\top \mathbf{\Sigma}_0^{-1}\,,
\end{align}
which can be confirmed by looking at the quadratic cross variation between $\frac{dS_t^i}{S_t^i}$ and $\frac{dF_t^j}{F_t^j}$ for all $i\leq d$ and $j\leq m$. 

\begin{remark}
Please note that from equation \eqref{eq:beta}, it follows that the cross variation $\frac{dF_t^j}{F_t^j}dZ_t^i$ equals zero for any $i$ and $j$. This will be relevant in Section \ref{subsec:constrained_model} when we equate market neutrality, or factor neutrality, to portfolios being adapted to the filtration generated by the $Z_t^i$'s.  
\end{remark}

We consider a self-financing portfolio process $W_{t}$ with portfolio weight in the $i^{\mathrm{th}}$ risky asset at time $t$ denoted by $\pi_{t}^{i}$ for $i=1,\:2,\:\cdots,\:d$, and $\left(1-\sum_{i=1}^{d}\pi_{t}^{i}\right)$ being the proportion of wealth invested in the risk-free asset with interest rate $r$. The dynamics of $W_t$ are given by the following SDE:
\begin{align}
dW_{t}=\sum_{i=1}^{d}\pi_{t}^{i}W_{t}\frac{dS_{t}^{i}}{S_{t}^{i}}+r\left(1-\sum_{i=1}^{d}\pi_{t}^{i}\right)W_{t}dt \,.\label{eq:dW}
\end{align}
Please note that none of the factors $F_{t}^{j}$ are contained in the portfolio, as you can observe from equation \eqref{eq:dW}. 

The value function for this stochastic optimal control problem is denoted by $u\left(t,\:w,\:\boldsymbol{z}\right)$. A control process $\left(\boldsymbol{\pi}_t\right)_{t\leq T}$ is at each time the portfolio allocation $\boldsymbol{\pi}_{t}=\left[\pi_{t}^{1},\:\pi_{t}^{2},\:\cdots,\:\pi_{t}^{d}\right]^{\top}\in\mathbb{R}^{d}$, which is sought to maximise the expectation of utility function $U\left(\cdot\right)$ with respect to the terminal wealth $W_T$,
\begin{align}
u\left(t,\:w,\:\boldsymbol{z}\right)=\sup_{\boldsymbol{\pi}\in\cal A}\mathbb{E}\left[U\left(W_{T}\right)\bigl|W_{t}=w,\:\boldsymbol{Z}_{t}=\boldsymbol{z}\right]\,,\label{eq:valueFunction}
\end{align}
where $t\in\left[0,\:T\right]$, $\boldsymbol{Z}_{t}=\left[Z_{t}^{1},\:Z_{t}^{2},\:\cdots,\:Z_{t}^{d}\right]^{\top}\in\mathbb{R}^{d}$, $w\in\mathbb R^+$ and $\boldsymbol{z}=\left[z^{1},\:z^{2},\:\cdots,\:z^{d}\right]^{\top}\in\mathbb{R}^{d}$ are the state variables, and $\left(\boldsymbol{\pi}_t\right)_{t\leq T}$ is selected from a set of admissible controls ${\cal A}$ defined as
\begin{align}
\mathcal{A}=\left\{\left(\boldsymbol{\pi}_{t}\right)_{t\leq T}\;\mathrm{with }\;\boldsymbol{\pi}_{t}\in\mathbb{R}^{d}\:\Bigl|\:\boldsymbol{\pi}_{t}\;\text{is non-anticipative and}\:\int_{0}^{T}\left\Vert \boldsymbol{\pi}_{t}W_{t}\right\Vert ^{2}dt<\infty\;\text{ a.s.}\right\}\label{eq:admissible_control}\,,
\end{align}
see chapter four of \citet{FlemingSoner2006} for more mathematical details. 

In this paper we assume a concave utility function $U\left(w\right)$ of power type:
\begin{align}
\label{eq:utility_function}
U\left(w\right)=\frac{1}{\gamma}w^{\gamma}\,,
\end{align}
where $\gamma<1$ $\left(\gamma\neq0\right)$ is the risk aversion coefficient. Risk aversion coefficient is utilised to measure risk preference of the trader. If $\gamma$ approaches to one, it indicates that the trader is more risk loving, and if $\gamma$ approaches to $-\infty$, it represents that the trader is more risk averse; $\gamma$ tending toward zero is the case of logarithmic utility.

%%%%%%%%%%%%%%%%%%%%%%%%%%%%%%%%%%%%%%%%%%%%%%%%%%%%%%%%%%%%%%%%%%%%%%%%%%%%%%%
\subsection{Hamilton-Jacobi-Bellman Equation\label{subsec:unconstrained_hjb}}

In addition to $\boldsymbol{\beta}$, $\mathbf{\Sigma}_1$, and $\mathbf{\Sigma}_2$ that are defined in Section \ref{subsec:unconstrained_sde}, we denote the following vectors and matrices for mathematical convenience for the rest sections of this article,
\begin{align}
\boldsymbol{\mu}
&=\left[\mu^{i}-r\right]\in\mathbb{R}^{d}\,,\label{eq:notations}\\
\boldsymbol{\theta}
&=\left[\theta^{1},\:\theta^{2},\cdots,\:\theta^{d}\right]^{\top}\in\mathbb{R}^{d}\,, \nonumber\\
\boldsymbol{\delta}
&=\mathrm{diag}\left(\left[\delta^{1},\:\delta^{2},\:\cdots,\:\delta^{d}\right]\right)\in\mathbb{R}^{d\times d}\,, \nonumber\\
\mathbf{\Sigma}_{2} 
&=\mathbf{\Sigma}_{1}-\mathbf{\Psi}_1\mathbf{\Psi}_0^\top\boldsymbol{\beta}^\top\in\mathbb{R}^{d\times d}\,, \nonumber\\
\mathbf{\Sigma}_{3} 
&=\mathbf{\Sigma}_1-\mathbf{\Psi}_1\mathbf{\Psi}_0^\top \boldsymbol{\beta}^\top
- \boldsymbol{\beta}\mathbf{\Psi}_0\mathbf{\Psi}_1^\top+\boldsymbol{\beta}\mathbf{\Sigma}_0\boldsymbol{\beta}^\top\in\mathbb{R}^{d\times d}\,,\nonumber
\end{align}
Please note that we assume that $\mathbf{\Sigma}_{1}$ is positive definite and invertible in this article. We can also observe that $\mathbf{\Sigma}_{1}$ and $\mathbf{\Sigma}_{3}$ are symmetric matrices, however, $\mathbf{\Sigma}_{2}$ is not. Subsequently, we apply the standard stochastic control and optimisation techniques and expect the value function $u\left(t,\:w,\:\boldsymbol{z}\right)$ defined in equation \eqref{eq:valueFunction} to satisfy the following HJB equation:
\begin{align}
-u_{t} & -\left(\boldsymbol{\theta}-\boldsymbol{z}\right)^{\top}\boldsymbol{\delta}\nabla_{\boldsymbol{z}}u-\frac{1}{2}\mathrm{tr}\left(\mathbf{\Sigma}_{3}\nabla_{\boldsymbol{z}}^{2}u\right)-rwu_{w}\label{eq:unconstrainedHJB}\\
& -\sup_{\boldsymbol{\pi}\in\mathbb R^d}\left(\boldsymbol{\pi}^{\top}\left(\boldsymbol{\mu}-\boldsymbol{\delta}\boldsymbol{z}\right)wu_{w}+\boldsymbol{\pi}^{\top}\mathbf{\Sigma}_{2}\nabla_{\boldsymbol{z}}\left(\nabla_{w}u\right)w+\frac{1}{2}\boldsymbol{\pi}^{\top}\mathbf{\Sigma}_{1}\boldsymbol{\pi}w^{2}u_{ww}\right)=0\,,\nonumber\\
u\big|_T & =\frac{1}{\gamma}w^{\gamma}\,.\nonumber 
\end{align}
In equation \eqref{eq:unconstrainedHJB}, we can observe that the control variable $\boldsymbol{\pi}$ is not subject to any condition, so we call it the unconstrained stochastic optimal control problem, and the optimal portfolio $W_{t}$ that is given by its solution is called the unconstrained portfolio. The wealth variable $w$ can be factored out of the solution by utilising the following ansatz:
\begin{align}
u\left(t,\:w,\:\boldsymbol{z}\right)=\frac{1}{\gamma}w^{\gamma}g\left(t,\:\boldsymbol{z}\right),\label{eq:ansatz_u}
\end{align}
and the derivative with respect to this ansatz function are 
\begin{align*}
u_{t} & =\frac{1}{\gamma}w^{\gamma}g_{t}\,,\quad\quad u_{ww}=\left(\gamma-1\right)w^{\gamma-2}g\,,\quad\quad\nabla_{\boldsymbol{z}}u=\frac{1}{\gamma}w^{\gamma}\nabla_{\boldsymbol{z}}g\,,\\
u_{w} & =w^{\gamma-1}g\,,\quad\quad\nabla_{\boldsymbol{z}}^{2}u=\frac{1}{\gamma}w^{\gamma}\nabla_{\boldsymbol{z}}^{2}g\,,\quad\quad\nabla_{\boldsymbol{z}}\left(\nabla_{w}u\right)=w^{\gamma-1}\nabla_{\boldsymbol{z}}g\,.
\end{align*}
Therefore, the HJB equation \eqref{eq:unconstrainedHJB} can be transformed into
\begin{align}
-g_{t} & -\left(\boldsymbol{\theta}-\boldsymbol{z}\right)^{\top}\boldsymbol{\delta}\nabla_{\boldsymbol{z}}g-\frac{1}{2}\mathrm{tr}\left(\mathbf{\Sigma}_{3}\nabla_{\boldsymbol{z}}^{2}g\right)-r\gamma g\label{eq:hjb_g_u}\\
& -\inf_{\boldsymbol{\pi}\in\mathbb R^d}\left(\boldsymbol{\pi}^{\top}\left(\boldsymbol{\mu}-\boldsymbol{\delta}\boldsymbol{z}\right)\gamma g+\boldsymbol{\pi}^{\top}\gamma\mathbf{\Sigma}_{2}\nabla_{\boldsymbol{z}}g+\frac{1}{2}\boldsymbol{\pi}^{\top}\mathbf{\Sigma}_{1}\boldsymbol{\pi}g\gamma\left(\gamma-1\right)\right)=0\,,\nonumber \\
g\big|_T & =1\,.\nonumber 
\end{align}
We then can compute the optimal control variable $\boldsymbol{\pi}^{*}$ by solving the stochastic optimal control problem described by equation \eqref{eq:hjb_g_u} in terms of function $g\left(t,\:\mathbf{z}\right)$ and its partial derivatives
\begin{align}
\boldsymbol{\pi}^{*}
& =\frac{1}{1-\gamma}\mathbf{\Sigma}_{1}^{-1}\left(\boldsymbol{\mu}-\boldsymbol{\delta}\boldsymbol{z}\right)+\frac{1}{1-\gamma}\mathbf{\Sigma}_{1}^{-1}\mathbf{\Sigma}_{2}\nabla_{\boldsymbol{z}}\left(\ln g\right)\,.\label{eq:control_pi_u}
\end{align}
Inserting the optimal $\boldsymbol{\pi}^{*}$ of equation \eqref{eq:control_pi_u} back into equation \eqref{eq:hjb_g_u} results in the following non-linear partial differential equation:
\begin{align}
g_{t} & +\left(\boldsymbol{\theta}-\boldsymbol{z}\right)^{\top}\boldsymbol{\delta}\nabla_{\boldsymbol{z}}g+\frac{1}{2}\mathrm{tr}\left(\mathbf{\Sigma}_{3}\nabla_{\boldsymbol{z}}^{2}g\right)+r\gamma g\label{eq:pde_g_u}\\
& +\frac{\gamma g}{2\left(1-\gamma\right)}\left(\boldsymbol{\mu}-\boldsymbol{\delta}\boldsymbol{z}+\mathbf{\Sigma}_{2}\nabla_{\boldsymbol{z}}\left(\ln g\right)\right)^{\top}\mathbf{\Sigma}_{1}^{-1}\left(\boldsymbol{\mu}-\boldsymbol{\delta}\boldsymbol{z}+\mathbf{\Sigma}_{2}\nabla_{\boldsymbol{z}}\left(\ln g\right)\right)=0\nonumber\,, \\
g\big|_T & =1\,.\nonumber 
\end{align}

We approach solving PDE \eqref{eq:pde_g_u} with an exponential ansatz for $g\left(t,\:\boldsymbol{z}\right)$,
\begin{align}
g\left(t,\:\boldsymbol{z}\right)=\exp\left(a\left(t\right)+\boldsymbol{b}^{\top}\left(t\right)\boldsymbol{z}+\boldsymbol{z}^{\top}\boldsymbol{C}\left(t\right)\boldsymbol{z}\right),\label{eq:ansztz_g}
\end{align}
where $a\left(t\right)\in\mathbb{R}$ is a scalar, $\boldsymbol{b}\left(t\right)=\left[b^{i}\left(t\right)\right]\in\mathbb{R}^{d}$ is a column vector, and $\boldsymbol{C}\left(t\right)=\left[c^{ij}\left(t\right)\right]\in\mathbb{R}^{d\times d}$ is a symmetric matrix, for $i,\:j=1,\:2,\:\cdots,\:d$. By utilising the ansatz \eqref{eq:ansztz_g}, PDE \eqref{eq:pde_g_u} can be transformed into a system of ODEs.

\begin{proposition} 
\label{prop:pde_solution_u}
The PDE \eqref{eq:pde_g_u} for the unconstrained stochastic optimal control problem \eqref{eq:unconstrainedHJB} is solved by utilising the exponential ansatz \eqref{eq:ansztz_g}, where for any $t\leq T$ the functions $a\left(t\right)\in\mathbb{R}$, $\boldsymbol{b}\left(t\right)\in\mathbb{R}^{d}$, and $\boldsymbol{C}\left(t\right)\in\mathbb{R}^{d\times d}$ satisfy the following system of ODEs:
\begin{align}
\frac{d a\left(t\right)}{d t} & =-\boldsymbol{b}^{\top}\left(t\right)\left(\frac{\gamma}{2\left(1-\gamma\right)}\mathbf{\Sigma}_{2}^{\top}\mathbf{\Sigma}_{1}^{-1}\mathbf{\Sigma}_{2}+\frac{1}{2}\mathbf{\Sigma}_{3}\right)\boldsymbol{b}\left(t\right)\label{eq:ode_a_u}\\
& \quad\;-\frac{\gamma}{2\left(1-\gamma\right)}\boldsymbol{b}^{\top}\left(t\right)\mathbf{\Sigma}_{2}^{\top}\mathbf{\Sigma}_{1}^{-1}\boldsymbol{\mu}-\frac{\gamma}{2\left(1-\gamma\right)}\boldsymbol{\mu}^{\top}\mathbf{\Sigma}_{1}^{-1}\mathbf{\Sigma}_{2}\boldsymbol{b}\left(t\right)-\boldsymbol{\theta}^{\top}\boldsymbol{\delta}\boldsymbol{b}\left(t\right)\nonumber \\
& \quad\;-\frac{\gamma}{2\left(1-\gamma\right)}\boldsymbol{\mu}^{\top}\mathbf{\Sigma}_{1}^{-1}\boldsymbol{\mu}-\mathrm{tr}\left[\mathbf{\Sigma}_{3}\boldsymbol{C}\left(t\right)\right]-r\gamma\,,\nonumber \\
a\left(T\right) & =0\,;\nonumber 
\end{align}
\begin{align}
\frac{d\boldsymbol{b}\left(t\right)}{d t} & =-\boldsymbol{C}\left(t\right)\left(\frac{2\gamma}{1-\gamma}\mathbf{\Sigma}_{2}^{\top}\mathbf{\Sigma}_{1}^{-1}\mathbf{\Sigma}_{2}+\mathbf{2\Sigma}_{3}\right)\boldsymbol{b}\left(t\right)\label{eq:ode_b_u}\\
& \quad\;+\left(\frac{\gamma}{1-\gamma}\boldsymbol{\delta}\mathbf{\Sigma}_{1}^{-1}\mathbf{\Sigma}_{2}+\boldsymbol{\delta}\right)\boldsymbol{b}\left(t\right)-\boldsymbol{C}\left(t\right)\left(\frac{2\gamma}{1-\gamma}\mathbf{\Sigma}_{2}^{\top}\mathbf{\Sigma}_{1}^{-1}\boldsymbol{\mu}+2\boldsymbol{\delta}\boldsymbol{\theta}\right)\nonumber \\
& \quad\;+\frac{\gamma}{1-\gamma}\boldsymbol{\delta}\mathbf{\Sigma}_{1}^{-1}\boldsymbol{\mu}\,,\nonumber \\
\boldsymbol{b}\left(T\right) & =\boldsymbol{0}\,;\nonumber 
\end{align}
\begin{align}
\frac{d\boldsymbol{C}\left(t\right)}{d t} & =-\boldsymbol{C}\left(t\right)\left(\frac{2\gamma}{1-\gamma}\mathbf{\Sigma}_{2}^{\top}\mathbf{\Sigma}_{1}^{-1}\mathbf{\Sigma}_{2}+2\mathbf{\Sigma}_{3}\right)\boldsymbol{C}\left(t\right)\label{eq:ode_c_u}\\
& \quad\;+\boldsymbol{C}\left(t\right)\left(\frac{\gamma}{1-\gamma}\mathbf{\Sigma}_{2}^{\top}\mathbf{\Sigma}_{1}^{-1}\boldsymbol{\delta}+\boldsymbol{\delta}\right)+\left(\frac{\gamma}{1-\gamma}\mathbf{\Sigma}_{2}^{\top}\mathbf{\Sigma}_{1}^{-1}\boldsymbol{\delta}+\boldsymbol{\delta}\right)^{\top}\boldsymbol{C}\left(t\right)\nonumber \\
& \quad\;-\frac{\gamma}{2\left(1-\gamma\right)}\boldsymbol{\delta}\mathbf{\Sigma}_{1}^{-1}\boldsymbol{\delta}\,,\nonumber \\
\boldsymbol{C}\left(T\right) & =\boldsymbol{0}\,.\nonumber 
\end{align}
\end{proposition}
\begin{proof}
\textit{By inserting the exponential ansatz \eqref{eq:ansztz_g} into PDE \eqref{eq:pde_g_u}, and grouping terms as either quadratic in $\boldsymbol{z}$ and $\boldsymbol{z}^{\top}$, linear in $\boldsymbol{z}$, or constant in $\boldsymbol{z}$, then ODE \eqref{eq:ode_a_u}, ODE \eqref{eq:ode_b_u}, and ODE \eqref{eq:ode_c_u} are respectively obtained.}
\end{proof}

%\textcolor{red}{We don't need this:\\We can observe that the ODE (\ref{eq:ode_a_u}) with respect to $a\left(t\right)$, ODE (\ref{eq:ode_b_u}) with respect to $\boldsymbol{b}\left(t\right)$, and ODE (\ref{eq:ode_c_u}) with respect to $\boldsymbol{C}\left(t\right)$ are coupled. The coupling is recursive in the sense that the equation with respect to $\boldsymbol{C}\left(t\right)$ is autonomous, while we can solve for $\boldsymbol{b}\left(t\right)$ given $\left(\boldsymbol{C}\left(t\right)\right)_{s\geq t}$, and solve for $a\left(t\right)$ given $\left(\boldsymbol{b}\left(t\right)\right)_{s\geq t}$ and $\left(\boldsymbol{C}\left(t\right)\right)_{s\geq t}$. The type of the ODE that $\boldsymbol{C}\left(t\right)$ solves is a matrix Riccati equation.} 

The gradient of the exponential ansatz is $\nabla_{\boldsymbol{z}}\left(\ln g\right) = 2\boldsymbol{C}\left(t\right)\boldsymbol{z} + \boldsymbol{b}\left(t\right)$ so that the optimal $\boldsymbol{\pi^{*}}$ from equation \eqref{eq:control_pi_u} can be expressed as
\begin{align}
\boldsymbol{\pi}^{*}=\frac{1}{1-\gamma}\mathbf{\Sigma}_{1}^{-1}\left(\boldsymbol{\mu}-\boldsymbol{\delta}\boldsymbol{z}\right)+\frac{1}{1-\gamma}\mathbf{\Sigma}_{1}^{-1}\mathbf{\Sigma}_{2}\left(2\boldsymbol{C}\left(t\right)\boldsymbol{z}+\boldsymbol{b}\left(t\right)\right)\,.\label{eq:control_pi_u2}
\end{align}
Verification of optimality for strategy $\boldsymbol{\pi}^{*}$ given by \eqref{eq:control_pi_u} is shown by using the argument of \citet{DavisLleo2014}:

\begin{remark}[Verification]
\label{prop:verification_theorem_u}
Utilising the solution of the HJB equation \eqref{eq:hjb_g_u} that is obtained from the exponential ansatz \eqref{eq:ansztz_g}, the optimal control $\left(\boldsymbol{\pi}^{*}\right)_{t\leq T}$ given by formula \eqref{eq:control_pi_u} belongs to the set of admissible controls $\mathcal{A}$ defined in equation \eqref{eq:admissible_control}, and maximises the expected utility function defined in equation \eqref{eq:valueFunction}. The form of the model for this unconstrained stochastic optimal control problem fits into the framework set forth in \citet{DavisLleo2008} and \citet{DavisLleo2014}, and hence the same argument for verification applies here.
\end{remark}

From equation \eqref{eq:control_pi_u2}, we can observe that the optimal control $\boldsymbol{\pi}^{*}$ contains two parts: a time-constant part and dynamic part that contains the solution of the HJB equation \eqref{eq:hjb_g_u}. We call the time-constant part of equation \eqref{eq:control_pi_u2} the myopic control of the unconstrained portfolio,
\begin{align}
\boldsymbol{\pi}_{m}^{*}=\frac{1}{1-\gamma}\mathbf{\Sigma}_{1}^{-1}\left(\boldsymbol{\mu}-\boldsymbol{\delta}\boldsymbol{z}\right)\,.\label{eq:myopic_control_u}
\end{align}

%%%%%%%%%%%%%%%%%%%%%%%%%%%%%%%%%%%%%%%%%%%%%%%%%%%%%%%%%
\subsection{Stability Analysis\label{subsec:unconstrained_stability}}

ODE \eqref{eq:ode_c_u} is a matrix Riccati equation. Stability analyses of the matrix Riccati equation, the linear ODE \eqref{eq:ode_b_u}, and the ODE \eqref{eq:ode_a_u} inform us whether our solution to PDE \eqref{eq:pde_g_u} blows up or not. We extend the time domain for the ODEs \eqref{eq:ode_c_u}, \eqref{eq:ode_b_u}, and \eqref{eq:ode_a_u} to $\left(-\infty,T\right]$ for any finite $T$, and if the solution remains finite for all time then we have a stable system from which we can draw intuition about long-term portfolio performance. 

Our analysis for the matrix Riccati equation \eqref{eq:ode_c_u} with respect to $\boldsymbol{C}\left(t\right)$ utilises Theorem 2.1 from \citet{Wonham1968} to show that the solution of the equation exists, is bounded, and is unique for all $t\leq T$. Let us rewrite the matrix Riccati ODE as
\begin{align}
\frac{d\boldsymbol{C}\left(t\right)}{d t} & =-\mathbf{A}_{u}^{\top}\boldsymbol{C}\left(t\right)-\boldsymbol{C}\left(t\right)\mathbf{A}_{u}-\boldsymbol{C}\left(t\right)\mathbf{Q}_{u}\boldsymbol{C}\left(t\right)-\mathbf{P}_{u}\,,\label{eq:riccati_u}\\
\boldsymbol{C}\left(T\right) & =\boldsymbol{0}\,,\nonumber 
\end{align}
where
\begin{align*}
\mathbf{Q}_{u} & =\frac{2\gamma}{1-\gamma}\mathbf{\Sigma}_{2}^{\top}\mathbf{\Sigma}_{1}^{-1}\mathbf{\Sigma}_{2}+2\mathbf{\Sigma}_{3}\,,\\
\mathbf{A}_{u} & =-\frac{\gamma}{1-\gamma}\mathbf{\Sigma}_{2}^{\top}\mathbf{\Sigma}_{1}^{-1}\boldsymbol{\delta}-\boldsymbol{\delta}\,,\\
\mathbf{P}_{u} & =\frac{\gamma}{2\left(1-\gamma\right)}\boldsymbol{\delta}\mathbf{\Sigma}_{1}^{-1}\boldsymbol{\delta}\,.
\end{align*}

\begin{proposition}
\label{prop:riccati_quadractic_term_u}For $\gamma<0$, the coefficient matrices of the quadratic term and the constant term of the matrix Riccati equation with respect to $\boldsymbol{C}\left(t\right)$, namely $\mathbf{Q}_{u}$ and $\mathbf{P}_{u}$ in equation \eqref{eq:riccati_u}, are symmetric positive definite and symmetric negative definite, respectively.
\end{proposition}
\begin{proof}
\textit{From formula \eqref{eq:notations}, the coefficient matrix $\mathbf{Q}_{u}$ of the quadratic term of the matrix Riccati equation \eqref{eq:riccati_u} for the unconstrained stochastic optimal control problem has the following decomposition,
\begin{align}
\mathbf{Q}_{u} 
& =\frac{2\gamma}{1-\gamma}\left(\mathbf{\Sigma}_{1}-\mathbf{\Psi}_1  \mathbf{\Psi}_0^\top \boldsymbol{\beta}^\top\right)^{\top}\mathbf{\Sigma}_{1}^{-1}\left(\mathbf{\Sigma}_{1}-\mathbf{\Psi}_1\mathbf{\Psi}_0^\top \boldsymbol{\beta}^\top\right)+2\mathbf{\Sigma}_3\nonumber\\%\label{eq:riccati_q_u}\\
& =\frac{2\gamma}{1-\gamma}\mathbf{\Sigma}_{1}-\frac{2\gamma}{1-\gamma}\mathbf{\Psi}_1 \mathbf{\Psi}_0^\top\boldsymbol{\beta}^\top-\frac{2\gamma}{1-\gamma}\boldsymbol{\beta}\mathbf{\Psi}_0 \mathbf{\Psi}_1^\top+\frac{2\gamma}{1-\gamma}\boldsymbol{\beta}\mathbf{\Psi}_0\mathbf{\Psi}_1^\top \mathbf{\Sigma}_1^{-1}\mathbf{\Psi}_1\mathbf{\Psi}_0^\top \boldsymbol{\beta}^\top +2\mathbf{\Sigma}_3 \nonumber\\
& =\frac{2}{1-\gamma}\mathbf{\Sigma}_{3}-\frac{2\gamma}{1-\gamma}\boldsymbol{\beta}\mathbf{\Sigma}_0\boldsymbol{\beta}^\top+\frac{2\gamma}{1-\gamma}\boldsymbol{\beta}\mathbf{\Psi}_0\mathbf{\Psi}_1^\top \mathbf{\Sigma}_1^{-1}\mathbf{\Psi}_1\mathbf{\Psi}_0^\top \boldsymbol{\beta}^\top \nonumber \\
& =\frac{2}{1-\gamma}\mathbf{\Sigma}_{3}-\frac{2\gamma}{1-\gamma}\boldsymbol{\beta}\mathbf{\Psi}_0\left(\mathbf{I}-\mathbf{\Psi}_1^\top \mathbf{\Sigma}_1^{-1}\mathbf{\Psi}_1\right)\mathbf{\Psi}_0^\top \boldsymbol{\beta}^\top \,,\nonumber 
\end{align}
which is symmetric positive definite for $\gamma<0$ if we can show that $\mathbf{I}-\mathbf{\Psi}_1^\top \mathbf{\Sigma}_1^{-1}\mathbf{\Psi}_1$ is symmetric positive semi-definite. Indeed, for any vector $\boldsymbol{x}\in\mathbb R^{d+m}$ we have $\boldsymbol{x} = \mathbf{\Psi}_1^\top \boldsymbol{y} + \tilde{\boldsymbol{y}} $ where $\mathbf{\Psi}_1\tilde{\boldsymbol{y}} = 0$. Then, $\boldsymbol{x}^\top\left(\mathbf{I}-\mathbf{\Psi}_1^\top \mathbf{\Sigma}_1^{-1}\mathbf{\Psi}_1\right)\boldsymbol{x}=\tilde{\boldsymbol{y}}^\top \boldsymbol{y} \geq 0 $, thereby confirming positive semi-definiteness.}

\textit{Proving that $\mathbf{P}_{u}$ is symmetric negative definite is uncomplicated. We can observe that matrix $\boldsymbol{\delta}\mathbf{\Sigma}_{1}^{-1}\boldsymbol{\delta}$ is symmetric positive definite. Consequently, for $\gamma<0$, $-\mathbf{P}_{u}$ is symmetric positive definite, in other words $\mathbf{P}_{u}$ is symmetric negative definite.}
\end{proof}

Given Proposition \ref{prop:riccati_quadractic_term_u}, the stability analysis from \citet{Wonham1968} applies directly, and to do so it will be useful to define the following properties: 

\begin{definition}[Controllability]
\label{def:controllability}
\textit{Let $\mathbf{A}\in\mathbb{R}^{n\times n}$ and $\mathbf{B}\in\mathbb{R}^{n\times m}$ be constant matrices. The controllability matrix of $\left(\mathbf{A},\:\mathbf{B}\right)$ is the $n\times mn$ matrix $\boldsymbol{\Gamma}\left(\mathbf{A},\:\mathbf{B}\right)=\left[\mathbf{B},\:\mathbf{A}\mathbf{B},\:\cdots,\:\mathbf{A}^{n-1}\mathbf{B}\right]$. The pair $\left(\mathbf{A},\:\mathbf{B}\right)$ is controllable if the rank of $\boldsymbol{\Gamma}$ is $n$. If $\left(\mathbf{A},\:\mathbf{B}\right)$ is controllable, so is $\left(\mathbf{A}-\mathbf{B}\mathbf{M},\:\mathbf{B}\right)$ for every matrix $\mathbf{M}\in\mathbb{R}^{m\times n}$.}
\end{definition}
\begin{definition}[Observability]
\label{def:observability}
\textit{Let $\mathbf{A}\in\mathbb{R}^{n\times n}$ and $\mathbf{E}\in\mathbb{R}^{p\times n}$ be constant matrices. The pair $\left(\mathbf{E},\:\mathbf{A}\right)$ is observable if the pair $\left(\mathbf{A}^{\top},\:\mathbf{E}^{\top}\right)$ is controllable.}
\end{definition}
\begin{definition}[Stabilisability]
\label{def:Stabilisability}
\textit{Let $\mathbf{A}\in\mathbb{R}^{n\times n}$ and $\mathbf{B}\in\mathbb{R}^{n\times m}$ be constant matrices. The pair $\left(\mathbf{A},\:\mathbf{B}\right)$ is stabilisable if there exists a constant matrix $\mathbf{M}$ such that all the eigenvalues of $\mathbf{A}-\mathbf{B}\mathbf{M}$ have negative real parts.}
\end{definition}

With these definitions that are given above, we have the following proposition for the matrix Riccati equation \eqref{eq:ode_c_u}.

\begin{proposition}
\label{prop:riccati_solution_u}
For $\gamma<0$, the coefficient matrix $\mathbf{Q}_{u}$ of the quadratic term in the matrix Riccati equation \eqref{eq:riccati_u} is symmetric positive definite. Consequently, there are matrices $\mathbf{B}_{u}$, $\mathbf{E}_{u}$, and $\mathbf{N}_{u}$ such that $\mathbf{Q}_{u}=\mathbf{B}_{u}\mathbf{N}_{u}^{-1}\mathbf{B}_{u}^{\top}$ and $-\mathbf{P}_{u}=\mathbf{E}_{u}^{\top}\mathbf{E}_{u}$, with the pair $\left(\mathbf{A}_{u},\:\mathbf{B}_{u}\right)$ being stabilisable, and the pair $\left(\mathbf{E}_{u},\:\mathbf{A}_{u}\right)$ being observable. Hence, there is an unique solution $\boldsymbol{C}\left(t\right)$ to the matrix Riccati ODE \eqref{eq:riccati_u} that is negative semi-definite and bounded on $\left(-\infty,\:T\right]$, and there exists a unique limit $\bar{\mathbf{C}}=\lim_{t\rightarrow-\infty}\boldsymbol{C}\left(t\right)$.
\end{proposition}
\begin{proof}
\textit{We first perform a change of variable. Define $\tilde{\boldsymbol{C}}\left(t\right)=-\boldsymbol{C}\left(t\right)$, so the matrix Riccati equation \eqref{eq:riccati_u} becomes}
\begin{align}
\frac{d\tilde{\boldsymbol{C}}\left(t\right)}{d t} & +\mathbf{A}_{u}^{\top}\tilde{\boldsymbol{C}}\left(t\right)+\tilde{\boldsymbol{C}}\left(t\right)\mathbf{A}_{u}-\tilde{\boldsymbol{C}}\left(t\right)\mathbf{Q}_{u}\tilde{\boldsymbol{C}}\left(t\right)+\left(-\mathbf{P}_{u}\right)=0\,,\label{eq:riccati_q2_u}\\
\tilde{\boldsymbol{C}}\left(T\right) & =0\,.\nonumber 
\end{align}
\textit{Proposition \ref{prop:riccati_quadractic_term_u} has shown that matrix $-\mathbf{P}_{u}\in\mathbb{R}^{d\times d}$ is symmetric positive definite. So, all eigenvalues $\boldsymbol{\lambda}_{P_{u}}$ of $-\mathbf{P}_{u}$ are positive and there exists an orthonormal basis for $\mathbb{R}^{d}$ of their associated eigenvectors, in other words, there is an orthonormal matrix $\mathbf{O}_{u}$ such that $-\mathbf{P}_{u}=\mathbf{O}_{u}\mathbf{D}_{u}\mathbf{O}_{u}^{\top}$, where $\mathbf{D}_{u}=\text{diag}\left(\boldsymbol{\lambda}_{P_{u}}\right)\in\mathbb{R}^{d\times d}$ is a diagonal matrix with positive entries on the diagonal. Hence, we can write $-\mathbf{P}_{u}=\mathbf{E}_{u}^{\top}\mathbf{E}_{u}$, where $\mathbf{E}_{u}=\left(\mathbf{O}_{u}\sqrt{\mathbf{D}_{u}}\right)^{\top}$ is a real square matrix. The matrix $\mathbf{E}_u$ is invertible, and so the controllability matrix $\boldsymbol{\Gamma}\left(\mathbf{A}_{u}^{\top},\:\mathbf{E}_{u}^{\top}\right)\in\mathbb{R}^{d\times d^{2}}$, as defined by Definition \ref{def:controllability}, has rank $d$. Consequently, the pair $\left(\mathbf{A}_{u}^{\top},\:\mathbf{E}_{u}^{\top}\right)$ is controllable and the pair $\left(\mathbf{E}_{u},\:\mathbf{A}_{u}\right)$ is observable as per Definition \ref{def:observability}.}

\textit{The symmetric positive definiteness of matrix $\mathbf{Q}_{u}\in\mathbb{R}^{d\times d}$ is proven in Proposition \ref{prop:riccati_quadractic_term_u} as well. Thus, the matrix also has a diagonal decomposition, $\mathbf{Q}_{u}=\mathbf{B}_{u}\mathbf{N}_{u}^{-1}\mathbf{B}_{u}^{\top}$, where $\mathbf{B}_{u}$ is an orthogonal matrix, and $\mathbf{N}_{u}^{-1}=\text{diag}\left(\boldsymbol{\lambda}_{Q_{u}}\right)\in\mathbb{R}^{d\times d}$ where $\boldsymbol{\lambda}_{Q_{u}}$ are the positive eigenvalues of $\mathbf{Q}_{u}$. The matrix $\mathbf{B}_{u}$ is invertible, hence we can find a constant matrix $\mathbf{M}_{u}\in\mathbb{R}^{d\times d}$ such that all eigenvalues of $\mathbf{A}_{u}-\mathbf{B}_{u}\mathbf{M}_{u}$ have negative real parts, therefore the pair $\left(\mathbf{A}_{u},\:\mathbf{B}_{u}\right)$ is stabilisable.}

\textit{The above analyses of matrix $\mathbf{Q}_{u}$ and matrix $-\mathbf{P}_{u}$ confirm that we can apply Theorem 2.1 from \citet{Wonham1968} to conclude that solution $\tilde{\boldsymbol{C}}\left(t\right)$ to the matrix Riccati equation \eqref{eq:riccati_q2_u} is unique, positive semi-definite, bounded on $\left(-\infty,\:T\right]$, and has unique limit as $t$ tends toward $-\infty$.}
\end{proof}

\begin{remark}
The stability analysis of Proposition \ref{prop:riccati_solution_u} is sufficient for there to be no arbitrage in the model proposed by equations \eqref{eq:dF} and equation \eqref{eq:dSi}. If there were an arbitrage then it would always be optimal to take additional positions in the arbitrage portfolio, hence causing the value function to have a singularity in finite time, thus reaching a Nirvana, see \citet{LeePapanicolaou2016}. Stability of the matrix Riccati equation with respect to $\boldsymbol{C}\left(t\right)$ ensures no such singularity for $\gamma<0$.
\end{remark}

After solving the matrix Riccati equation \eqref{eq:ode_c_u}, we then study the stability of the solution to the linear ODE \eqref{eq:ode_b_u} with respect to $\boldsymbol{b}\left(t\right)$. We start the analysis by introducing the following lemma, which is a theorem from \citet{Wielandt1973} in regard to the eigenvalues of matrices. General mathematical knowledge for this lemma can be found in chapter one of \citet{HornJohnson1994}.

\begin{lemma}
\label{lem:eigenvalue_lemma}Let $\mathbf{M}$ be a $d\times d$ matrix and define the field of values $\boldsymbol{S}\left(\mathbf{M}\right):=\big\{ \boldsymbol{v}^{\top}\mathbf{M}\boldsymbol{v}\mid\boldsymbol{v}\;is\;a\;vector \allowbreak such\;that\;\boldsymbol{v}^{\top}\boldsymbol{v}=1\big\} $, which contains the eigenvalues of $\mathbf{M}$.

(a) Suppose $\mathbf{M}_{1}$ and $\mathbf{M}_{2}$ are two $d\times d$ matrices. If $\lambda$ is an eigenvalue of $\mathbf{M}_{1}+\mathbf{M}_{2}$, then $\lambda\in \boldsymbol{S}\left(\mathbf{M}_{1}\right)+\boldsymbol{S}\left(\mathbf{M}_{2}\right)=\left\{ \lambda_{1}+\lambda_{2}\mid\lambda_{1}\in \boldsymbol{S}\left(\mathbf{M}_{1}\right),\;\lambda_{2}\in \boldsymbol{S}\left(\mathbf{M}_{2}\right)\right\}$;

(b) Suppose $\mathbf{M}_{1}$ and $\mathbf{M}_{2}$ are two $d\times d$ matrices, $0\notin \boldsymbol{S}\left(\mathbf{M}_{2}\right)$ and $\mathbf{M}_{2}^{-1}$
exits. If $\lambda$ is an eigenvalue of $\mathbf{M}_{2}^{-1}\mathbf{M}_{1}$, then $\lambda\in\boldsymbol{S}\left(\mathbf{M}_{1}\right)/\boldsymbol{S}\left(\mathbf{M}_{2}\right)=\left\{\lambda_{1}/\lambda_{2}\mid\lambda_{1}\in \boldsymbol{S}\left(\mathbf{M}_{1}\right),\;\lambda_{2}\in \boldsymbol{S}\left(\mathbf{M}_{2}\right)\right\}$;

(c) Suppose $\mathbf{M}_{1}$ is an arbitrary $d\times d$ matrix, $\mathbf{M}_{2}$ is symmetric positive semi-definite matrix. If $\lambda$ is an eigenvalue of $\mathbf{M}_{1}\mathbf{M}_{2}$, then $\lambda\in \boldsymbol{S}\left(\mathbf{M}_{1}\right)\boldsymbol{S}\left(\mathbf{M}_{2}\right)=\left\{ \lambda_{1}\lambda_{2}\mid\lambda_{1}\in \boldsymbol{S}\left(\mathbf{M}_{1}\right),\;\lambda_{2}\in \boldsymbol{S}\left(\mathbf{M}_{2}\right)\right\}$.
\end{lemma}

\begin{proof}
\textit{The detailed and comprehensive proofs are given by the theorems of \citet{Wielandt1973}}
\end{proof}

\begin{proposition}
\label{prop:ode_b_solution_u}
Let $\boldsymbol{R}_{u}\left(t\right)$ be the coefficient matrix of the homogeneous part of equation \eqref{eq:ode_b_u}. For $\gamma<0$ there exists $t^{*}>-\infty$ such that $\boldsymbol{R}_{u}\left(t\right)$ has all positive eigenvalues for $t<t^{*}$, and therefore the solution of ODE \eqref{eq:ode_b_u} has a finite steady state.
\end{proposition}

\begin{proof}
\textit{By observing ODE \eqref{eq:ode_b_u} and utilising the formulae given by equation \eqref{eq:notations}, we have the coefficient matrix for the ODE,
\begin{align}
\boldsymbol{R}_{u}\left(t\right) 
& =-2\boldsymbol{C}\left(t\right)\left(\frac{\gamma}{1-\gamma}\mathbf{\Sigma}_{2}^{\top}\mathbf{\Sigma}_{1}^{-1}\mathbf{\Sigma}_{2}+\mathbf{\Sigma}_{3}\right)+\left(\frac{\gamma}{1-\gamma}\boldsymbol{\delta}\mathbf{\Sigma}_{1}^{-1}\mathbf{\Sigma}_{2}+\boldsymbol{\delta}\right)\nonumber \\
& =-\boldsymbol{C}\left(t\right)\mathbf{Q}_{u}+\frac{1}{1-\gamma}\boldsymbol{\delta}-\frac{\gamma}{1-\gamma}\boldsymbol{\delta}\mathbf{\Sigma}_{1}^{-1}\mathbf{\Psi}_1  \mathbf{\Psi}_0^\top \boldsymbol{\beta}^\top \nonumber \,.
\end{align}
Utilising equation \eqref{eq:beta}, this expression can be simplified to
\begin{align}
\label{eq:ode_b_r_u}
\boldsymbol{R}_{u}\left(t\right) =-\boldsymbol{C}\left(t\right)\mathbf{Q}_{u}+\frac{1}{1-\gamma} \boldsymbol{\delta}-\frac{\gamma}{1-\gamma}\boldsymbol{\delta}\mathbf{\Sigma}_{1}^{-1}\left(\mathbf{\Psi}_1\mathbf{\Psi}_0^\top\mathbf{\Sigma}_0^{-1}\mathbf{\Psi}_0\mathbf{\Psi}_1^\top\right)\,.
\end{align}
Let $\bar{\mathbf{C}}=\lim_{t\rightarrow-\infty}\boldsymbol{C}\left(t\right)$, then the limit of equation \eqref{eq:ode_b_r_u} is
\begin{align}
\mathbf{R}_{u} =-\bar{\mathbf{C}}\mathbf{Q}_{u}+\frac{1}{1-\gamma} \boldsymbol{\delta}-\frac{\gamma}{1-\gamma}\boldsymbol{\delta}\mathbf{\Sigma}_{1}^{-1}\left(\mathbf{\Psi}_1\mathbf{\Psi}_0^\top\mathbf{\Sigma}_0^{-1}\mathbf{\Psi}_0\mathbf{\Psi}_1^\top\right)\,,\label{eq:ode_b_r2_u}
\end{align}
where matrix $\boldsymbol{\delta}$ is assumed to be symmetric positive semi-definite. }

\textit{Proposition \ref{prop:riccati_quadractic_term_u} proves that matrix $\mathbf{Q}_{u}$ is symmetric positive semi-definite, and proposition \ref{prop:riccati_solution_u} implies that matrix $-\bar{\mathbf{C}}$ is symmetric positive semi-definite. Matrix $\mathbf{\Psi}_1  \mathbf{\Psi}_0^\top\mathbf{\Sigma}_0^{-1}\mathbf{\Psi}_0\mathbf{\Psi}_1^\top$ is also symmetric positive semi-definite. Therefore, if $\gamma<0$ it follows from Lemma \ref{lem:eigenvalue_lemma} that matrix $\mathbf{R}_{u}$ given by equation \eqref{eq:ode_b_r2_u} has all positive eigenvalues. This allows us to say that there exists a $t^{*}>-\infty$ such that $\boldsymbol{R}_{u}\left(t\right)$ has positive eigenvalues for all $t<t^{*}$, and the solution $\boldsymbol{b}\left(t\right)$ of the equation \eqref{eq:ode_b_u} has a finite steady state.}
\end{proof}
\begin{proposition}
\label{prop:long_term_ode_a_u}For $\gamma<0$, the solution of equation \eqref{eq:ode_a_u} has a finite steady state as $t\rightarrow-\infty$. Consequently, the long term certainty equivalent rate of the unconstrained stochastic optimal control problem that is described by equation \eqref{eq:hjb_g_u} is asymptotically proportional to the solution of ODE \eqref{eq:ode_a_u}.
\end{proposition}
\begin{proof}
\textit{Denote by $L\left(t\right)$ the right hand side of equation \eqref{eq:ode_a_u}. From the analyses of Proposition \ref{prop:riccati_solution_u} and Proposition \ref{prop:ode_b_solution_u}, both solution $\boldsymbol{C}\left(t\right)$ and $\boldsymbol{b}\left(t\right)$ with respect to ODE \eqref{eq:ode_c_u}
and ODE \eqref{eq:ode_b_u} have finite limits as the time variable $t$ tends to $-\infty$, therefore when $t$ approaches to $-\infty$, we have
\begin{align*}
\lim_{t\rightarrow-\infty}L\left(t\right) 
& =-\bar{\boldsymbol{b}}^{\top}\left(\frac{\gamma}{2\left(1-\gamma\right)}\mathbf{\Sigma}_{2}^{\top}\mathbf{\Sigma}_{1}^{-1}\mathbf{\Sigma}_{2}+\frac{1}{2}\mathbf{\Sigma}_{3}\right)\bar{\boldsymbol{b}}\\
& \quad\;-\frac{\gamma}{2\left(1-\gamma\right)}\bar{\boldsymbol{b}}^{\top}\mathbf{\Sigma}_{2}^{\top}\mathbf{\Sigma}_{1}^{-1}\boldsymbol{\mu}-\frac{\gamma}{2\left(1-\gamma\right)}\boldsymbol{\mu}^{\top}\mathbf{\Sigma}_{1}^{-1}\mathbf{\Sigma}_{2}\bar{\boldsymbol{b}}-\boldsymbol{\theta}^{\top}\boldsymbol{\delta}\bar{\boldsymbol{b}}\\
& \quad\;-\frac{\gamma}{2\left(1-\gamma\right)}\boldsymbol{\mu}^{\top}\mathbf{\Sigma}_{1}^{-1}\boldsymbol{\mu}-\mathrm{tr}\left[\mathbf{\Sigma}_{3}\bar{\boldsymbol{C}}\right]-r\gamma\\
& =\bar{L}\,,
\end{align*}
where $\bar{\boldsymbol{b}}=\underset{t\rightarrow-\infty}{\lim}\boldsymbol{b}\left(t\right)$ and $\bar{\mathbf{C}}=\underset{t\rightarrow-\infty}{\lim}\boldsymbol{C}\left(t\right)$. Hence, as $t$ tends toward negative infinity, equation \eqref{eq:ode_a_u} relaxes and we have
\begin{align}
\lim_{t\rightarrow-\infty}\frac{1}{T-t}\int_{t}^{T}\frac{d a\left(s\right)}{d s} & ds=\bar{L}\,,\label{eq:l_bar_u}
\end{align}
which shows that the solution of ODE \eqref{eq:ode_a_u} has a finite steady state,
\begin{align}
\lim_{t\rightarrow-\infty}\frac{a\left(t\right)}{T-t}=-\bar{L}\,.\label{eq:a_asymptotic_u}
\end{align}}

\textit{Furthermore, given the utility function $\eqref{eq:utility_function}$, the value function \eqref{eq:ansatz_u}, and the exponential ansatz for $g\left(t,\:\boldsymbol{z}\right)$ that is defined by formula \eqref{eq:ansztz_g}, the certainty equivalent is defined by
\begin{align*}
U^{-1}\left(u\left(t,\:w,\:\boldsymbol{z}\right)\right) 
& =U^{-1}\left(\frac{1}{\gamma}w^{\gamma}g\left(t,\:\boldsymbol{z}\right)\right)\\
& =w\exp\left(\frac{1}{\gamma}\left(a\left(t\right)+\boldsymbol{b}^{\top}\left(t\right)\boldsymbol{z}+\boldsymbol{z}^{\top}\boldsymbol{C}\left(t\right)\boldsymbol{z}\right)\right)\,.
\end{align*}
Hence, under the optimal control variable $\boldsymbol{\pi}^{*}$, the long-term growth rate is 
\begin{align*}
\frac{\ln\left(U^{-1}\left(u\left(t,\:w,\:\boldsymbol{z}\right)\right)\right)}{T-t} 
& =\frac{1}{T-t}\ln\left(w\exp\left(\frac{1}{\gamma}\left(a\left(t\right)+\boldsymbol{b}^{\top}\left(t\right)\boldsymbol{z}+\boldsymbol{z}^{\top}\boldsymbol{C}\left(t\right)\boldsymbol{z}\right)\right)\right)\\
& =\frac{1}{\gamma\left(T-t\right)}\left(\gamma\ln w+\boldsymbol{b}^{\top}\left(t\right)\boldsymbol{z}+\boldsymbol{z}^{\top}\boldsymbol{C}\boldsymbol{z}+a\left(t\right)\right)\,.
\end{align*}
In the previous paragraphs, utilising the conclusions of Proposition \ref{prop:riccati_solution_u} and Proposition \ref{prop:ode_b_solution_u}, we have already shown by equation \eqref{eq:a_asymptotic_u} that $a\left(t\right)$ is asymptotically linear as $t$ tends to $-\infty$, therefore,
\begin{align}
\lim_{T\rightarrow\infty}\frac{\ln\left(U^{-1}\left(u\left(t,\:w,\:\boldsymbol{z}\right)\right)\right)}{T-t} 
& =\lim_{t\rightarrow-\infty}\frac{1}{\gamma\left(T-t\right)}\left(\gamma\ln w+\boldsymbol{b}^{\top}\left(t\right)\boldsymbol{z}\;+\boldsymbol{z}^{\top}\boldsymbol{C}\left(t\right)\boldsymbol{z}+a\left(t\right)\right)\label{eq:long_term_growth_u}\\
& =\lim_{t\rightarrow-\infty}\frac{1}{\gamma\left(T-t\right)}a\left(t\right)\,.\nonumber 
\end{align}
Consequently, utilising the result that is demonstrated by equation \eqref{eq:l_bar_u}, the limit in equation \eqref{eq:long_term_growth_u} is
\begin{align*}
\lim_{T\rightarrow\infty}\frac{1}{T-t}\ln\left(U^{-1}\left(u\left(t,\:w,\:\boldsymbol{z}\right)\right)\right)=-\frac{1}{\gamma}\bar{L}\,,
\end{align*}
which demonstrates that the long-term growth rate is a constant.}
\end{proof}

%%%%%%%%%%%%%%%%%%%%%%%%%%%%%%%%%%%%%%%%%%%%%%%%%%%%%%%%%%%%%%%%%%%%%%%%%%%%%%%%%%%%%%%%%%%%%%%%%%%%
\subsection{Market Neutral Constraint\label{subsec:constrained_model}}

It is common to seek statistical arbitrage strategies that are market neutrality. Market neutrality generally means that the returns of a portfolio are impacted only by the idiosyncratic returns of the stocks contained in the portfolio, and are uncorrelated with the returns of a benchmark or market factors, see \cite{angoshtari2014} and \cite{AvellanedaLee2010}. Hence, under the condition of market neutrality, if we can diversify with a large number of co-integrated stocks, then there is a very high probability that the portfolio can maintain steady growth and low volatility. We will consider a portfolio to be market neutral if $W_t$ is adapted to the filtration generated by $\boldsymbol{Z}_t$. This is the case if $\boldsymbol{\pi}^{\top}\boldsymbol{\beta}=0$, because
\begin{align}
\frac{dW_{t}}{W_{t}} 
& =\sum_{i=1}^{d}\pi_{t}^{i}\frac{dS_{t}^{i}}{S_{t}^{i}}-r\left(1-\sum_{i=1}^d\pi_t^i\right)W_tdt\label{eq:market_neutral}\\
& =\sum_{i=1}^{d}\pi_{t}^{i}\left(dZ_{t}^{i}+\alpha^{i}dt+\sum_{j=1}^m\beta^{ij}\frac{d F_{t}^{j}}{F_{t}^{j}}\right)-r\left(1-\sum_{i=1}^d\pi_t^i\right)W_tdt \nonumber\\
& =\boldsymbol{\pi}_t^\top d\boldsymbol{Z}_t +\boldsymbol{\pi}_t^\top\boldsymbol{\alpha} dt+\underbrace{\boldsymbol{\pi}_t^\top\boldsymbol{\beta}}_{=0}\frac{d \boldsymbol{F}_t}{\boldsymbol{F}_t}+r\left(1-\boldsymbol{\pi}_t^\top\mathbf 1\right)W_tdt \,, \nonumber
\end{align}
where $\boldsymbol{\alpha}=\left[\alpha^{1},\:\alpha^{2},\:\cdots,\:\alpha^{d}\right]^{\top}\in\mathbb R^d$, and $\mathbf 1\in\mathbb R^d$ is a vector of ones. In matrix/vector notation, market neutrality is $\boldsymbol{\pi}^{\top}\boldsymbol{\beta}=0$. 

We now reformulate the optimal portfolio that is studied in Section \ref{subsec:unconstrained_sde} and Section \ref{subsec:unconstrained_hjb}. We shall include an equality constraint $\boldsymbol{\pi}^{\top}\boldsymbol{\beta}=0$ with respect to the market neutrality for the control variable $\boldsymbol{\pi}$. Consequently, we now have a constrained stochastic optimal control problem, the portfolio $W_{t}$ that is given by its solution is called the market-neutral constrained portfolio, and the HJB equation is
\begin{align}
-u_{t} & -\left(\boldsymbol{\theta}-\boldsymbol{z}\right)^{\top}\boldsymbol{\delta}\nabla_{\boldsymbol{z}}u-\frac{1}{2}\mathrm{tr}\left(\mathbf{\Sigma}_{3}\nabla_{\boldsymbol{z}}^{2}u\right)-rwu_{w}\label{eq:constrainedHJB}\\
& -\sup_{\substack{\boldsymbol{\pi}\in\mathbb R^d\\
\boldsymbol{\pi}^{\top}\boldsymbol{\beta}=0}}\left( \boldsymbol{\pi}^{\top}\left(\boldsymbol{\mu}-\boldsymbol{\delta}\boldsymbol{z}\right)wu_{w}+\boldsymbol{\pi}^{\top}\mathbf{\Sigma}_{2}\nabla_{\boldsymbol{z}}\left(\nabla_{w}u\right)w+\frac{1}{2}\boldsymbol{\pi}^{\top}\mathbf{\Sigma}_{1}\boldsymbol{\pi}w^{2}u_{ww}\right) =0\,,\nonumber \\
u\big|_T &=\frac{1}{\gamma}w^{\gamma}\,.\nonumber 
\end{align}
Utilising equation \eqref{eq:ansatz_u}, the transformed HJB equation for the market-neutral constrained problem is,
\begin{align}
-g_{t} & -\left(\boldsymbol{\theta}-\boldsymbol{z}\right)^{\top}\boldsymbol{\delta}\nabla_{\boldsymbol{z}}g-\frac{1}{2}\mathrm{tr}\left(\mathbf{\Sigma}_{3}\nabla_{\boldsymbol{z}}^{2}g\right)-r\gamma g\label{eq:hjb_g_c}\\
& -\inf_{\substack{\boldsymbol{\pi}\in\mathbb R^d\\
\boldsymbol{\pi}^{\top}\boldsymbol{\beta}=0}}\left( \boldsymbol{\pi}^{\top}\left(\boldsymbol{\mu}-\boldsymbol{\delta}\boldsymbol{z}\right)\gamma g+\boldsymbol{\pi}^{\top}\gamma\mathbf{\Sigma}_{2}\nabla_{\boldsymbol{z}}g+\frac{1}{2}\boldsymbol{\pi}^{\top}\mathbf{\Sigma}_{1}\boldsymbol{\pi}g\gamma\left(\gamma-1\right)\right) =0\,,\nonumber \\
g\big|_T & =1\,.\nonumber 
\end{align}
Letting $\boldsymbol{\lambda}\in\mathbb R^{m}$ denote a Lagrange multiplier, we define the Lagrangian function for the market-neutral constrained control problem
\begin{align*}
L\left(\boldsymbol{\pi},\:\boldsymbol{\lambda}\right) =\boldsymbol{\pi}^{\top}\left(\boldsymbol{\mu}-\boldsymbol{\delta}\boldsymbol{z}\right)\gamma g+\boldsymbol{\pi}^{\top}\gamma\mathbf{\Sigma}_{2}\nabla_{\boldsymbol{z}}g+\frac{1}{2}\boldsymbol{\pi}^{\top}\mathbf{\Sigma}_{1}\boldsymbol{\pi}g\gamma\left(\gamma-1\right)-\boldsymbol{\pi}^{\top}\boldsymbol{\beta}\boldsymbol{\lambda}\,.
\end{align*}
By utilising the first order condition $\nabla_{\boldsymbol{\pi}}L=\boldsymbol{0}$, we can get the optimal control,
\begin{align}
\boldsymbol{\pi}^{*} 
& = \frac{1}{g\gamma\left(1-\gamma\right)}\mathbf{\Sigma}_{1}^{-1}\left(g\gamma\boldsymbol{\mu}-g\gamma\boldsymbol{\delta}\boldsymbol{z}+\gamma\mathbf{\Sigma}_{2}\nabla_{\boldsymbol{z}}g-\boldsymbol{\beta}\boldsymbol{\lambda}\right)\,,\label{eq:control_pi_c}\\
& = \frac{1}{1-\gamma}\mathbf{\Sigma}_{1}^{-1}\left(\boldsymbol{\mu}-\boldsymbol{\delta}\boldsymbol{z}+\mathbf{\Sigma}_{2}\nabla_{\boldsymbol{z}}\left(\ln g\right)-\frac{1}{g\gamma}\boldsymbol{\beta}\boldsymbol{\lambda}\right)\,.\nonumber 
\end{align}
We then solve for the Lagrange multiplier $\boldsymbol{\lambda}$ to get its optimal value with respect to the condition of market-neutral constraint $\boldsymbol{\beta}^\top \boldsymbol{\pi}^{*}=0$,
\begin{align}
\boldsymbol{\lambda}^{*} =\gamma g\left(\boldsymbol{\beta}^{\top}\mathbf{\Sigma}_{1}^{-1}\boldsymbol{\beta}\right)^{-1}\left(\boldsymbol{\beta}^{\top}\mathbf{\Sigma}_{1}^{-1}\left(\boldsymbol{\mu}-\boldsymbol{\delta}\boldsymbol{z}+\mathbf{\Sigma}_{2}\nabla_{\boldsymbol{z}}\left(\ln g\right)\right)\right)\,.\label{eq:lagrange_multiplier}
\end{align}
Inserting the optimal $\boldsymbol{\lambda}^{*}$ that is given by formula \eqref{eq:lagrange_multiplier} back into equation \eqref{eq:control_pi_c}, we can get the optimal control variables $\boldsymbol{\pi}^{*}$ with respect to the market-neutral constraint,
\begin{align}
\boldsymbol{\pi}^{*} =\frac{1}{1-\gamma}\left(\mathbf{\Sigma}_{1}^{-1}-\mathbf{\Sigma}_{c}\right)\left(\boldsymbol{\mu}-\boldsymbol{\delta}\boldsymbol{z}+\mathbf{\Sigma}_{2}\nabla_{\boldsymbol{z}}\left(\ln g\right)\right)\,,\label{eq:control_pi2_c}
\end{align}
where $\mathbf{\Sigma}_{c}=\mathbf{\Sigma}_{1}^{-1}\boldsymbol{\beta}\left(\boldsymbol{\beta}^{\top}\mathbf{\Sigma}_{1}^{-1}\boldsymbol{\beta}\right)^{-1}\boldsymbol{\beta}^\top \mathbf{\Sigma}_1^{-1}\in\mathbb{R}^{d\times d}$ is a symmetric matrix. We then insert the optimal control variable $\boldsymbol{\pi}^{*}$ that is given by formula \eqref{eq:control_pi2_c} back into the HJB equation \eqref{eq:hjb_g_c} and get the following non-linear PDE for the market-neutral constrained problem,
\begin{align}
g_{t} & +\left(\boldsymbol{\theta-\boldsymbol{z}}\right)^{\top}\boldsymbol{\delta}\nabla_{\boldsymbol{z}}g+\frac{1}{2}\mathrm{tr}\left(\mathbf{\Sigma}_{3}\nabla_{\boldsymbol{z}}^{2}g\right)+r\gamma g\label{eq:pde_g_c}\\
& +\frac{\gamma g}{2\left(1-\gamma\right)}\left(\boldsymbol{\mu}-\boldsymbol{\delta}\boldsymbol{z}+\mathbf{\Sigma}_{2}\nabla_{\boldsymbol{z}}\left(\ln g\right)\right)^{\top}\left(\mathbf{\Sigma}_{1}^{-1}-\mathbf{\Sigma}_{c}\right)\left(\boldsymbol{\mu}-\boldsymbol{\delta}\boldsymbol{z}+\mathbf{\Sigma}_{2}\nabla_{\boldsymbol{z}}\left(\ln g\right)\right)=0\,,\nonumber \\
g\big|_T & =1\,.\nonumber 
\end{align}
Corresponding to the unconstrained stochastic optimal control problem of Section \ref{subsec:unconstrained_hjb}, we approach solving PDE \eqref{eq:pde_g_c} by utilising the exponential ansatz \eqref{eq:ansztz_g}, from which the optimal portfolio can be expressed as
\begin{align}
\label{eq:control_pi2_c2}
\boldsymbol{\pi}^{*} =\frac{1}{1-\gamma}\left(\mathbf{\Sigma}_{1}^{-1}-\mathbf{\Sigma}_{c}\right)\left(\boldsymbol{\mu}-\boldsymbol{\delta}\boldsymbol{z}\right)+\frac{1}{1-\gamma}\left(\mathbf{\Sigma}_{1}^{-1}-\mathbf{\Sigma}_{c}\right)\left(\mathbf{\Sigma}_{2}\left(2\boldsymbol{C}\left(t\right)\boldsymbol{z}+\boldsymbol{b}\left(t\right)\right)\right)\,,
\end{align}
where $a\left(t\right)$, $\boldsymbol{b}\left(t\right)$, and $\boldsymbol{C}\left(t\right)$ are the solutions of the system of ODEs that are given by Proposition \ref{prop:ode_system_c}. Correspondingly, the myopic control for the market-neutral constrained problem is
\begin{align}
\boldsymbol{\pi}_{m}^{*}=\frac{1}{1-\gamma}\left(\mathbf{\Sigma}_{1}^{-1}-\mathbf{\Sigma}_{c}\right)\left(\boldsymbol{\mu}+\boldsymbol{\delta}\boldsymbol{z}\right)\,.\label{eq:myopic_control_c}
\end{align}

\begin{remark}
The myopic control in equation \eqref{eq:myopic_control_c} satisfies the market neutrality condition that is given by equation \eqref{eq:market_neutral}, which is $\boldsymbol{\beta}^\top \boldsymbol{\pi}_m^*=0$. This will be shown in the proof of Proposition \ref{prop:riccati_q_c} as a consequence of proving $\boldsymbol{\beta}^\top\left(\mathbf{\Sigma}_{1}^{-1}-\mathbf{\Sigma}_{c}\right)=0$.
\end{remark}

\begin{proposition}
\label{prop:ode_system_c}
The PDE \eqref{eq:pde_g_c} for the constrained stochastic optimal control problem \eqref{eq:constrainedHJB} is solved by utilising the exponential ansatz \eqref{eq:ansztz_g}, where for any $t\leq T$ the functions $a\left(t\right)\in\mathbb{R}$, $\boldsymbol{b}\left(t\right)\in\mathbb{R}^{d}$, and $\boldsymbol{C}\left(t\right)\in\mathbb{R}^{d\times d}$ satisfy the following system of ODEs:
\begin{align}
\frac{d a\left(t\right)}{d t} 
& =-\boldsymbol{b}^{\top}\left(t\right)\left(\frac{\gamma}{2\left(1-\gamma\right)}\mathbf{\Sigma}_{2}^{\top}\left(\mathbf{\Sigma}_{1}^{-1}-\mathbf{\Sigma}_{c}\right)\mathbf{\Sigma}_{2}+\frac{1}{2}\mathbf{\Sigma}_{3}\right)\boldsymbol{b}\left(t\right)\label{eq:ode_a_c}\\
& \quad\;-\frac{\gamma}{2\left(1-\gamma\right)}\boldsymbol{b}^{\top}\left(t\right)\mathbf{\Sigma}_{2}^{\top}\left(\mathbf{\Sigma}_{1}^{-1}-\mathbf{\Sigma}_{c}\right)\boldsymbol{\mu}-\frac{\gamma}{2\left(1-\gamma\right)}\boldsymbol{\mu}^{\top}\left(\mathbf{\Sigma}_{1}^{-1}-\mathbf{\Sigma}_{c}\right)\mathbf{\Sigma}_{2}\boldsymbol{b}\left(t\right)-\boldsymbol{\theta}^{\top}\boldsymbol{\delta}\boldsymbol{b}\left(t\right)\nonumber \\
& \quad\;-\frac{\gamma}{2\left(1-\gamma\right)}\boldsymbol{\mu}^{\top}\left(\mathbf{\Sigma}_{1}^{-1}-\mathbf{\Sigma}_{c}\right)\boldsymbol{\mu}-\mathrm{tr}\left[\mathbf{\Sigma}_{3}\boldsymbol{C}\left(t\right)\right]-r\gamma\,,\nonumber \\
a\left(T\right) & =0\,;\nonumber 
\end{align}
\begin{align}
\frac{d\boldsymbol{b}\left(t\right)}{d t} 
& =-\boldsymbol{C}\left(t\right)\left(\frac{2\gamma}{1-\gamma}\mathbf{\Sigma}_{2}^{\top}\left(\mathbf{\Sigma}_{1}^{-1}-\mathbf{\Sigma}_{c}\right)\mathbf{\Sigma}_{2}+2\mathbf{\Sigma}_{3}\right)\boldsymbol{b}\left(t\right)\label{eq:ode_b_c}\\
& \quad\;+\left(\frac{\gamma}{1-\gamma}\boldsymbol{\delta}\left(\mathbf{\Sigma}_{1}^{-1}-\mathbf{\Sigma}_{c}\right)\mathbf{\Sigma}_{2}+\boldsymbol{\delta}\right)\boldsymbol{b}\left(t\right)-\boldsymbol{C}\left(t\right)\left(\frac{2\gamma}{1-\gamma}\mathbf{\Sigma}_{2}^{\top}\left(\mathbf{\Sigma}_{1}^{-1}-\mathbf{\Sigma}_{c}\right)\boldsymbol{\mu}+2\boldsymbol{\delta}\boldsymbol{\theta}\right)\nonumber \\
& \quad\;+\frac{\gamma}{1-\gamma}\boldsymbol{\delta}\left(\mathbf{\Sigma}_{1}^{-1}-\mathbf{\Sigma}_{c}\right)\boldsymbol{\mu}\,,\nonumber \\
\boldsymbol{b}\left(T\right) & =0\,;\nonumber 
\end{align}
\begin{align}
\frac{d\boldsymbol{C}\left(t\right)}{d t} 
& =-\boldsymbol{C}\left(t\right)\left(\frac{2\gamma}{1-\gamma}\mathbf{\Sigma}_{2}^{\top}\left(\mathbf{\Sigma}_{1}^{-1}-\mathbf{\Sigma}_{c}\right)\mathbf{\Sigma}_{2}+2\mathbf{\Sigma}_{3}\right)\boldsymbol{C}\left(t\right)\label{eq:ode_c_c}\\
& \quad\;+\boldsymbol{C}\left(t\right)\left(\frac{\gamma}{1-\gamma}\mathbf{\Sigma}_{2}^{\top}\left(\mathbf{\Sigma}_{1}^{-1}-\mathbf{\Sigma}_{c}\right)\boldsymbol{\delta}+\boldsymbol{\delta}\right)+\left(\frac{\gamma}{1-\gamma}\mathbf{\Sigma}_{2}^{\top}\left(\mathbf{\Sigma}_{1}^{-1}-\mathbf{\Sigma}_{c}\right)\boldsymbol{\delta}+\boldsymbol{\delta}\right)^{\top}\boldsymbol{C}\left(t\right)\nonumber \\
& \quad\;-\frac{\gamma}{2\left(1-\gamma\right)}\boldsymbol{\delta}\left(\mathbf{\Sigma}_{1}^{-1}-\mathbf{\Sigma}_{c}\right)\boldsymbol{\delta}\,,\nonumber \\
\boldsymbol{C}\left(T\right) & =0\,.\nonumber 
\end{align}
\end{proposition}
\begin{proof}
\textit{The proof is similar to that for Proposition \ref{prop:pde_solution_u}.}
\end{proof}
\begin{remark}[Verification]
Utilising the solution of the HJB equation \eqref{eq:hjb_g_c} that is obtained from the exponential ansatz \eqref{eq:ansztz_g}, the optimal control $\left(\boldsymbol{\pi}_t^{*}\right)_{t\leq T}$ given by formula \eqref{eq:control_pi_c} belongs to the set of admissible controls $\mathcal{A}$ in equation \eqref{eq:admissible_control}, and maximises the expected utility function defined by equation \eqref{eq:valueFunction} subject to the market-neutral constraint $\boldsymbol{\pi}^{\top}\boldsymbol{\beta}=0$. The form of the model for this constrained stochastic optimal control problem fits into the framework set forth in \citet{DavisLleo2008} and \citet{DavisLleo2014}, and hence the same argument for verification applies here.
\end{remark}
Corresponding to the unconstrained stochastic optimal control problem of Sections \ref{subsec:unconstrained_hjb} and Section \ref{subsec:unconstrained_stability}, we perform the stability analyses for the solutions of the system of ODEs given by equation \eqref{eq:ode_a_c}, equation \eqref{eq:ode_b_c}, and equation \eqref{eq:ode_c_c}. We first study the behaviour of the solution to the matrix Riccati equation for $\boldsymbol{C}\left(t\right)$ of the constrained problem, which is expressed by equation \eqref{eq:ode_c_c}. We rewrite this matrix Riccati equation in the following way
\begin{align}
\frac{d\boldsymbol{C}\left(t\right)}{d t} 
& =-\boldsymbol{C}\left(t\right)\mathbf{Q}_{c}\boldsymbol{C}\left(t\right)-\boldsymbol{C}\left(t\right)\mathbf{A}_{c}-\mathbf{A}_{c}^{\top}\boldsymbol{C}\left(t\right)-\mathbf{P}_{c}\,,\label{eq:riccati_c}\\
\boldsymbol{C}\left(T\right) 
& =\boldsymbol{0}\,,\nonumber 
\end{align}
where 
\begin{align}
\mathbf{Q}_{c} 
& =\frac{2\gamma}{1-\gamma}\mathbf{\Sigma}_{2}^{\top}\left(\mathbf{\Sigma}_{1}^{-1}-\mathbf{\Sigma}_{c}\right)\mathbf{\Sigma}_{2}+2\mathbf{\Sigma}_{3}\,,\nonumber\\
\mathbf{A}_{c} 
& =-\frac{\gamma}{1-\gamma}\mathbf{\Sigma}_{2}^{\top}\left(\mathbf{\Sigma}_{1}^{-1}-\mathbf{\Sigma}_{c}\right)\boldsymbol{\delta}-\boldsymbol{\delta}\,,\nonumber \\
\mathbf{P}_{c} 
& =\frac{\gamma}{2\left(1-\gamma\right)}\boldsymbol{\delta}\left(\mathbf{\Sigma}_{1}^{-1}-\mathbf{\Sigma}_{c}\right)\boldsymbol{\delta}\,.\nonumber 
\end{align}

Similar to the unconstrained problem, we prove that the coefficient matrix $\mathbf{Q}_{c}$ for the quadratic term of the matrix Riccati equation \eqref{eq:riccati_c} for the constrained problem is symmetric positive definite.

\begin{proposition}
\label{prop:riccati_q_c}For $\gamma<0$, in the matrix Riccati equation \eqref{eq:riccati_c} for the constrained stochastic optimal control problem, the coefficient matrix $\mathbf{Q}_{c}$ of its quadratic term is symmetric positive definite, and the coefficient matrix $\mathbf{P}_{c}$ of its constant term is symmetric negative semi-definite.
\end{proposition}

\begin{proof}
\textit{Proposition \ref{prop:riccati_quadractic_term_u} has shown that matrix $\mathbf{Q}_{u}$ of equation \eqref{eq:riccati_u} for the unconstrained stochastic optimal control problem is symmetric positive definite. By utilising the formulae of equation \eqref{eq:notations}, $\mathbf{Q}_{c}$ has the following decomposition,
\begin{align}
\mathbf{Q}_{c} 
& =2\mathbf{\Sigma}_{3}+\frac{2\gamma}{1-\gamma}\mathbf{\Sigma}_{2}^{\top}\left(\mathbf{\Sigma}_{1}^{-1}-\mathbf{\Sigma}_{c}\right)\mathbf{\Sigma}_{2}\nonumber\\
& =2\mathbf{\Sigma}_{3}+\frac{2\gamma}{1-\gamma}\mathbf{\Sigma}_{2}^{\top}\left(\mathbf{\Sigma}_{1}^{-1}-\mathbf{\Sigma}_{1}^{-1}\boldsymbol{\beta}\left(\boldsymbol{\beta}^{\top}\mathbf{\Sigma}_{1}^{-1}\boldsymbol{\beta}\right)^{-1}\boldsymbol{\beta}^{\top}\mathbf{\Sigma}_{1}^{-1}\right)\mathbf{\Sigma}_{2}\nonumber \\
& =\mathbf{Q}_{u}-\frac{2\gamma}{1-\gamma}\mathbf{\Sigma}_{2}^{\top}\mathbf{\Sigma}_{1}^{-1}\boldsymbol{\beta}\left(\boldsymbol{\beta}^{\top}\mathbf{\Sigma}_{1}^{-1}\boldsymbol{\beta}\right)^{-1}\boldsymbol{\beta}^{\top}\mathbf{\Sigma}_{1}^{-1}\mathbf{\Sigma}_{2}\,.\nonumber
\end{align}
Because $\mathbf{\Sigma}_{1}$ is symmetric positive definite, so its inverse $\mathbf{\Sigma}_{1}^{-1}$ is symmetric positive definite as well. Hence $\left(\boldsymbol{\beta}^{\top}\mathbf{\Sigma}_{1}^{-1}\boldsymbol{\beta}\right)^{-1}>0$. Therefore, for $\gamma<0$, $\mathbf{Q}_{c}$ is symmetric positive definite.}

\textit{In order to prove matrix $\mathbf{P}_{c}$ is symmetric negative semi-definite, we need to examine the symmetric matrix $\mathbf{\Sigma}_{1}^{-1}-\mathbf{\Sigma}_{c}$. We observe that for any $\boldsymbol{x}\in\mathbb R^d$, we can decompose it as $\boldsymbol{x} = \boldsymbol{\beta}\boldsymbol{v} + \tilde{\boldsymbol{v}}$, where $\tilde{\boldsymbol{v}}^\top \mathbf{\Sigma}_1^{-1}\boldsymbol{\beta} =0$, and then
\begin{align*}
\boldsymbol{x}^{\top}\left(\mathbf{\Sigma}_{1}^{-1}-\mathbf{\Sigma}_{c}\right)\boldsymbol{x} 
& =\left(\boldsymbol{\beta}\boldsymbol{v}+\tilde{\boldsymbol{v}}\right)^{\top}\mathbf{\Sigma}_{1}^{-1}\left(\boldsymbol{\beta}\boldsymbol{v}+\tilde{\boldsymbol{v}}\right)\\
& \quad\;-\left(\boldsymbol{\beta}\boldsymbol{v}+\tilde{\boldsymbol{v}}\right)^{\top}\mathbf{\Sigma}_{1}^{-1}\boldsymbol{\beta}\left(\boldsymbol{\beta}^{\top}\mathbf{\Sigma}_{1}^{-1}\boldsymbol{\beta}\right)^{-1}\boldsymbol{\beta}^{\top}\mathbf{\Sigma}_{1}^{-1}\left(\boldsymbol{\beta}\boldsymbol{v}+\tilde{\boldsymbol{v}}\right)\\
& =\tilde{\boldsymbol{v}}^{\top}\mathbf{\Sigma}_{1}^{-1}\tilde{\boldsymbol{v}}\\
& \geq0\,.
\end{align*}
where the equality holds if and only if $\tilde{\boldsymbol{v}} = 0$. Hence, $\mathbf{\Sigma}_{1}^{-1}-\mathbf{\Sigma}_{c}$ is symmetric positive semi-definite. Consequently, $\boldsymbol{\delta}\left(\mathbf{\Sigma}_{1}^{-1}-\mathbf{\Sigma}_{c}\right)\boldsymbol{\delta}$ is a symmetric positive semi-definite matrix with the columns of $\boldsymbol{\delta}^{^{-1}}\boldsymbol{\beta}$ spanning its null space. Therefore, $-\mathbf{P}_{c}$ is symmetric positive semi-definite also with null-space spanned by $\boldsymbol{\delta}^{^{-1}}\boldsymbol{\beta}$.}
\end{proof}

Corresponding to Proposition \ref{prop:riccati_solution_u}, because $\mathbf{Q}_{c}$ is symmetric positive definite, $-\mathbf{P}_{c}$ is symmetric positive semi-definite, therefore Theorem 2.1 in \citet{Wonham1968} applies directly, consequently the solution to matrix Riccati equation \eqref{eq:ode_c_c} exists, is bounded, and is unique.

\begin{proposition}
\label{prop:riccati_solution_c}
For $\gamma<0$, if both matrices $\boldsymbol{\delta}-\boldsymbol{\beta}\left(\boldsymbol{\beta}^\top\boldsymbol{\beta}\right)^{-1}\boldsymbol{\beta}^\top\boldsymbol{\delta}$ and $\boldsymbol{\beta}^\top\boldsymbol{\beta}$ are full-rank, 
\begin{align}
\label{eq:beta_ode_c}
\mathrm{rank}\left(\boldsymbol{\delta}-\boldsymbol{\beta}\left(\boldsymbol{\beta}^\top\boldsymbol{\beta}\right)^{-1}\boldsymbol{\beta}^\top\boldsymbol{\delta}\right)&=d\,,\\
\nonumber
\mathrm{rank}\left(\boldsymbol{\beta}^\top\boldsymbol{\beta}\right)&=m\,,
\end{align}
then the coefficient matrix $\mathbf{Q}_{c}$ of the quadratic term in the matrix Riccati equation \eqref{eq:riccati_c} for the constrained stochastic optimal control problem is symmetric positive definite, and $-\boldsymbol{P}_{c}$ is symmetric positive semi-definite. Hence, there are matrices $\mathbf{B}_{c}$, $\mathbf{E}_{c}$, and $\mathbf{N}_{c}$ such that $\mathbf{Q}_{c}=\mathbf{B}_{c}\mathbf{N}_{c}^{-1}\mathbf{B}_{c}^{\top}$ and $-\mathbf{P}_{c}=\mathbf{E}_{c}^{\top}\mathbf{E}_{c}$, with the pair $\left(\mathbf{A}_{c},\:\mathbf{B}_{c}\right)$ being stabilisable, and the pair $\left(\mathbf{E}_{c},\:\mathbf{A}_{c}\right)$ being observable. Consequently, there is an unique solution $\boldsymbol{C}\left(t\right)$ to matrix Riccati equation \eqref{eq:riccati_c} that is negative semi-definite and bounded on $\left(-\infty,\:T\right]$, and there exists a unique limit $\bar{\mathbf{C}}=\lim_{t\rightarrow-\infty}\boldsymbol{C}\left(t\right)$.
\end{proposition}

\begin{proof}
\textit{The symmetric positive semi-definiteness of matrix $-\mathbf{P}_{c}$ comes from Proposition \ref{prop:riccati_q_c}, for which the proof shows that $\mathbf{\Sigma}_{1}^{-1}-\mathbf{\Sigma}_{1}^{-1}\boldsymbol{\beta}\left(\boldsymbol{\beta}^{\top}\mathbf{\Sigma}_{1}^{-1}\boldsymbol{\beta}\right)^{-1}\boldsymbol{\beta}^{\top}\mathbf{\Sigma}_{1}^{-1}$ is symmetric positive semi-definite. This matrix is diagonalizable, $\mathbf{\Sigma}_{1}^{-1}-\mathbf{\Sigma}_{1}^{-1}\boldsymbol{\beta}\left(\boldsymbol{\beta}^{\top}\mathbf{\Sigma}_{1}^{-1}\boldsymbol{\beta}\right)^{-1}\boldsymbol{\beta}^{\top}\mathbf{\Sigma}_{1}^{-1}=\mathbf{O}_{c}\mathbf{D}_{c}\mathbf{O}_{c}^{\top}$, where $\mathbf{O}_{c}\in\mathbb{R}^{d\times d}$ is orthonormal and $\mathbf{D}_{c}$ is a diagonal matrix with non-negative eigenvalues along its diagonal. Thus $-\mathbf{P}_{c}=\mathbf{E}_{c}^{\top}\mathbf{E}_{c}$, where $\mathbf{E}_{c}=\sqrt{\frac{-\gamma}{2\left(1-\gamma\right)}}\sqrt{\mathbf{D}_{c}}\mathbf{O}_{c}^{\top}\boldsymbol{\delta}$. As shown in the proof of Proposition \ref{prop:riccati_q_c}, there is a null space of $-\mathbf{P}_{c}$ that is spanned by the columns of matrix $\boldsymbol{\delta}^{-1}\boldsymbol{\beta}$, and hence the rank of matrix $\mathbf{E}_{c}$ is strictly less than $d$ with null space spanned by $\boldsymbol{\delta}^{-1}\boldsymbol{\beta}$. However, the rank of the controllability matrix $\boldsymbol{\Gamma}\left(\mathbf{A}_{c}^{\top},\:\mathbf{E}_{c}^{\top}\right)\in\mathbb{R}^{d\times d^{2}}$ given by Definition \ref{def:controllability} is $d$ if the equations of formula \eqref{eq:beta_ode_c} hold. }

\textit{It is indeed correct for the model that is proposed here in Section \ref{subsec:constrained_model}. Recall the matrix $\mathbf{A}_{c}$ defined for the matrix Riccati ODE \eqref{eq:riccati_c}, $\mathbf{A}_{c}^{\top}  =-\frac{\gamma}{1-\gamma}\boldsymbol{\delta}\left(\mathbf{\Sigma}_{1}^{-1}-\mathbf{\Sigma}_{c}\right)^{\top}\mathbf{\Sigma}_{2}-\boldsymbol{\delta}=2\mathbf{E}_{c}^{\top}\mathbf{E}_{c}\boldsymbol{\delta}^{-1}\mathbf{\Sigma}_{2}-\boldsymbol{\delta}$, which when multiplied on the left-hand side by $\boldsymbol{\delta}^{-1}\boldsymbol{\beta}$ and on the right by $\boldsymbol{E}_c^\top$,
\begin{align*}
\left(\boldsymbol{\delta}^{^{-1}}\boldsymbol{\beta}\right)^{\top}\mathbf{A}_{c}^{\top}\mathbf{E}_{c}^{\top} 
& =\left(\boldsymbol{\delta}^{^{-1}}\boldsymbol{\beta}\right)^{\top}\left(2\mathbf{E}_{c}^{\top}\mathbf{E}_{c}\boldsymbol{\delta}^{-1}\mathbf{\Sigma}_{2}-\boldsymbol{\delta}\right)\mathbf{E}_{c}^{\top}\\
& =2\left(\mathbf{E}_{c}\boldsymbol{\delta}^{^{-1}}\boldsymbol{\beta}\right)^{\top}\mathbf{E}_{c}\boldsymbol{\delta}^{-1}\mathbf{\Sigma}_{2}\mathbf{E}_{c}^{\top}-\left(\boldsymbol{\delta}^{^{-1}}\boldsymbol{\beta}\right)^{\top}\boldsymbol{\delta}\mathbf{E}_{c}^{\top}\\
& =-\left(\mathbf{E}_{c}\boldsymbol{\beta}\right)^{\top}\,.
\end{align*}
We can observe that, for any vector $\boldsymbol{x}$ we know that  $\mathbf{E}_{c}\boldsymbol{\beta}\boldsymbol{x}=0$ if and only if $\boldsymbol{x}^\top\boldsymbol{\beta}^\top\mathbf{P}_c\boldsymbol{\beta}\boldsymbol{x}=0$, which occurs if and only if $\exists\,\boldsymbol{v}$ such that $\boldsymbol{\beta}\boldsymbol{x} =\boldsymbol{\delta}^{-1}\boldsymbol{\beta}\boldsymbol{v}$. The nearest such $\boldsymbol{v}$ is $\widehat{\boldsymbol{v}}= \left(\boldsymbol{\beta}^\top\boldsymbol{\beta}\right)^{-1}\boldsymbol{\beta}^\top\boldsymbol{\delta}\boldsymbol{\beta}\boldsymbol{x}$. Therefore, $\mathbf{E}_{c}\boldsymbol{\beta}\boldsymbol{x}=0$ if we can find $\boldsymbol{x}$ such that $\boldsymbol{\delta}\boldsymbol{\beta}\boldsymbol{x} =\boldsymbol{\beta}\widehat{\boldsymbol{v}}$. This is the case if $\left(\mathbf{I}-\boldsymbol{\beta}\left(\boldsymbol{\beta}^\top\boldsymbol{\beta}\right)^{-1}\boldsymbol{\beta}^\top\right)\boldsymbol{\delta}\boldsymbol{\beta}\boldsymbol{x}=0$ where $\mathbf{I}\in\mathbb{R}^{d\times d}$ is the identity matrix, which occurs only for $\boldsymbol{x} =0$ if equations in formula \eqref{eq:beta_ode_c} hold, because $\boldsymbol{\beta}^\top\boldsymbol{\beta}$ is invertible implying that $\boldsymbol{\beta}$ has no right-hand null vector. Therefore, $\boldsymbol{\Gamma}\left(\mathbf{A}_{c}^{\top},\:\mathbf{E}_{c}^{\top}\right)\in\mathbb{R}^{d\times d^{2}}$ has full rank if formula \eqref{eq:beta_ode_c} holds. Thus, the pair $\left(\mathbf{A}_{u}^{\top},\:\mathbf{E}_{u}^{\top}\right)$ is controllable and the pair $\left(\mathbf{E}_{u},\:\mathbf{A}_{u}\right)$ is observable as per Definition \ref{def:observability}.}

\textit{The symmetric positive definiteness of matrix $\mathbf{Q}_{c}$ is proven in Proposition \ref{prop:riccati_q_c} as well. Corresponding to Proposition \ref{prop:riccati_solution_u}, matrix $\mathbf{Q}_{c}=\mathbf{B}_{c}\mathbf{N}_{c}^{-1}\mathbf{B}_{c}^{\top}$, where $\mathbf{B}_{c}$ is an orthogonal matrix, $\boldsymbol{\lambda}_{Q_{c}}$ are the eigenvalues of $\mathbf{Q}_{c}$, and $\mathbf{N}_{c}^{-1}=\text{diag}\left(\boldsymbol{\lambda}_{Q_{c}}\right)\in\mathbb{R}^{d\times d}$. The matrix $\mathbf{B}_{c}$ is invertible, hence we can find a constant matrix $\mathbf{M}_{c}\in\mathbb{R}^{d\times d}$ such that all eigenvalues of matrix $\mathbf{A}_{c}-\mathbf{B}_{c}\mathbf{M}_{c}$ have negative real parts, therefore the pair $\left(\mathbf{A}_{c},\:\mathbf{B}_{c}\right)$ is stabilisable.}

\textit{Finally, we let $\tilde{\boldsymbol{C}}\left(t\right)=-\boldsymbol{C}\left(t\right)$, so that the matrix Riccati equation \eqref{eq:riccati_c} becomes
\begin{align}
\frac{d\tilde{\boldsymbol{C}}\left(t\right)}{d t} 
& +\mathbf{A}_{c}^{\top}\tilde{\boldsymbol{C}}\left(t\right)+\tilde{\boldsymbol{C}}\left(t\right)\mathbf{A}_{c}-\tilde{\boldsymbol{C}}\left(t\right)\mathbf{B}_{c}\mathbf{N}_{c}^{-1}\mathbf{B}_{c}^{\top}\tilde{\boldsymbol{C}}\left(t\right)+\mathbf{E}_{c}^{\top}\mathbf{E}_{c}=0\,,\label{eq:riccati_q2_c}\\
\tilde{\boldsymbol{C}}\left(T\right) 
& =0\,,\nonumber 
\end{align}
and the above analyses of matrices $\mathbf{Q}_{c}$ and $-\mathbf{P}_{c}$ confirm that we can apply Theorem 2.1 from \citet{Wonham1968} again, to conclude that solution $\tilde{\boldsymbol{C}}\left(t\right)$ to equation \eqref{eq:riccati_q2_c} is unique, positive semidefinite, bounded on $\left(-\infty,\:T\right]$, and has unique limit as $t$ tends toward $-\infty$.}
\end{proof}

\begin{remark}
In Proposition \ref{prop:riccati_solution_c}, the implication of needing $\boldsymbol{\beta}^\top\boldsymbol{\beta}$ having full rank is that there needs to be at least as many co-integrated stocks as there are factors, as $\boldsymbol{\beta}\in\mathbb R^{d\times m}$. 
\end{remark}

\begin{remark}
In Proposition \ref{prop:riccati_solution_c}, a necessary condition for the first equation of formula  \eqref{eq:beta_ode_c} to hold is for $\boldsymbol{\beta}\left(\boldsymbol{\beta}^\top \boldsymbol{\beta}\right)^{-1}\boldsymbol{\beta}^\top$ to not commute with $\boldsymbol{\delta}$, which implies that $\boldsymbol{\delta}\not\propto \mathbf{I}$. 
\end{remark}

Next we analyse the behaviour of the solution for the ODE with respect to $\boldsymbol{b}\left(t\right)$ of the constrained stochastic optimal control problem that is described by equation \eqref{eq:ode_b_c}.

\begin{proposition}
\label{prop:ode_b_solution_c}
Let $\boldsymbol{R}_{c}\left(t\right)$ be the coefficient matrix for the homogeneous part of equation \eqref{eq:ode_b_c} for the constrained stochastic optimal control problem. For $\gamma<0$, there exists $t^{*}>-\infty$ such that $\boldsymbol{R}_{c}\left(t\right)$ has all positive eigenvalues for $t<t^{*}$. Therefore, the solution $\boldsymbol{b}\left(t\right)$ with respect to ODE \eqref{eq:ode_b_c} has a finite steady state.
\end{proposition}

\begin{proof}
\textit{By observing equation \eqref{eq:ode_b_c}, and following the notations that are denoted by formula \eqref{eq:notations}, also utilising the expression for $\boldsymbol{\beta}$ and for $\mathbf{\Sigma}_2$ in equations \eqref{eq:beta} and equation \eqref{eq:notations}, respectively, we have the coefficient matrix for the ODE,
\begin{align*}
\boldsymbol{R}_{c}\left(t\right) 
& =-2\boldsymbol{C}\left(t\right)\left(\frac{\gamma}{1-\gamma}\mathbf{\Sigma}_{2}^{\top}\left(\mathbf{\Sigma}_{1}^{-1}-\mathbf{\Sigma}_{c}\right)\mathbf{\Sigma}_{2}+\mathbf{\Sigma}_{3}\right)\\
& \quad\;+\left(\frac{\gamma}{1-\gamma}\boldsymbol{\delta}\left(\mathbf{\Sigma}_{1}^{-1}-\mathbf{\Sigma}_{c}\right)\mathbf{\Sigma}_{2}+\boldsymbol{\delta}\right)\\
& =-\boldsymbol{C}\left(t\right)\mathbf{Q}_{c}+\frac{1}{1-\gamma}\boldsymbol{\delta}-\frac{\gamma}{1-\gamma}\boldsymbol{\delta}\mathbf{\Sigma}_{c}\mathbf{\Sigma}_{1}\\
& \quad\;-\frac{\gamma}{1-\gamma}\boldsymbol{\delta}\left(\mathbf{\Sigma}_{1}^{-1}-\mathbf{\Sigma}_{c}\right)\left(\mathbf{\Psi}_1\mathbf{\Psi}_0^{\top}\mathbf{\Sigma}_{0}^{-1}\mathbf{\Psi}_0\mathbf{\Psi}_1^{\top}\right)\,.
\end{align*}
From Proposition \ref{prop:riccati_q_c}, we know that matrix $\mathbf{Q}_{c}$ is symmetric positive definite, from Proposition \ref{prop:riccati_solution_c}, we have that matrix $-\boldsymbol{C}\left(t\right)$ is positive semi-definite for $t<t^{*}$ with $t^{*}>-\infty$, and by the assumptions of the model, we have that matrices $\mathbf{\Sigma}_1$, $\mathbf{\Sigma}_c$, and $\boldsymbol{\delta}$ are positive definite. Also note that $\mathbf{\Sigma}_1^{-1}-\mathbf{\Sigma}_c$ is symmetric positive semi-definite, matrix $\mathbf{\Psi}_1  \mathbf{\Psi}_0^\top\mathbf{\Sigma}_0^{-1}\mathbf{\Psi}_0\mathbf{\Psi}_1^\top$ is symmetric positive semi-definite as well. Therefore, the above expression for $\boldsymbol{R}_{c}\left(t\right)$ is the summations and produces of positive semi-definite matrices, it has all positive eigenvalues, and
the solution of ODE \eqref{eq:ode_b_c} is stable.}
\end{proof}

\begin{remark}
For $\gamma<0$, the solution $a\left(t\right)$ with respect to equation \eqref{eq:ode_a_c} for the constrained stochastic optimal control problem has a finite steady state as $t\rightarrow-\infty$. The long-term  growth  rate  of  the  certainty equivalent of the constrained stochastic optimal control problem is proportional to the solution $a\left(t\right)$ of ODE \eqref{eq:ode_a_c}. The proof of such a proposition is similar to the proof of Proposition \ref{prop:long_term_ode_a_u}.
\end{remark}

%%%%%%%%%%%%%%%%%%%%%%%%%%%%%%%%%%%%%%%%%%%%%%%%%%%%%%%%%%%%%%%%%%%
\section{Numerical Experiments and Empirical Analyses\label{sec:numerical_empirical}}

This section presents sliding-window backtests on historical stock data. These tests demonstrate the performance of the optimal portfolios derived in Section \ref{sec:model_construction}. In Section \ref{subsec:eigenportfolio}, we first describe the method that is utilised for constructing eigenportfolios, which we then utilise as the factors $F_t^j$ seen in equation \eqref{eq:dF}. Then, we explain how to adjust the data for survivorship bias. The methods for parameter estimations and the approach for statistical testing of co-integration are described in Section \ref{subsec:parameter_estimation}. Lastly, the backtests and empirical analyses are presented in the Section \ref{subsec:portfolio_performance}.

%%%%%%%%%%%%%%%%%%%%%%%%%%%%%%%%%%%%%%%%%%%%%%%%%%%%%%%%%%%%%%%%%%%
\subsection{Eigenportfolios for Constructing Factors\label{subsec:eigenportfolio}}

In order to implement the portfolios proposed in Section \ref{sec:model_construction}, selecting the factors in equation \eqref{eq:dF} is the initial step. The principal eigenportfolio is a factor because it tracks the capitalisation-weighted market portfolio, see \citet{AvellanedaHPT2020}, which is closely tracked by the SPDR S\&P 500 Trust ETF, whose ticker symbol is SPY. Our additional factors are higher-order eigenportfolios, as is done in \citet{AvellanedaLee2010} and \citet{YeoPapanicolaou2017}.

Suppose $\boldsymbol{\rho}\in\mathbb{R}^{d\times d}$ is the correlation matrix of the returns for the stocks, 
\begin{align*}
\boldsymbol{\rho}=\left[\mathrm{corr}\left(\frac{dS_{t}^{i}}{S_{t}^{i}},\:\frac{dS_{t}^{k}}{S_{t}^{k}}\right)\right],\quad i,\:k=1,\:2,\:\cdots,\:d\,.
\end{align*}
Then, its eigenvalue decomposition is 
\begin{align}
\boldsymbol{\rho}=\mathbf{V}\mathbf{\Lambda}\mathbf{V}^{\top}\,,\label{eq:corr_matrix}
\end{align}
where matrix $\mathbf{V}=\left[\mathbf{v}_{1},\:\mathbf{v}_{2},\:,\cdots\:\mathbf{v}_{d}\right]\in\mathbb{R}^{d\times d}$ is composed by eigenvectors of $\boldsymbol{\rho}$ with orthonormal property $\mathbf{V}^{\top}\mathbf{V}=\mathbf{V}\mathbf{V}^{\top}=\mathbf{I}$, and $\mathbf{\Lambda}\in\mathbb{R}^{d\times d}$ is a diagonal matrix whose diagonal elements $\lambda^{ii}$ are the corresponding eigenvalues with $\lambda^{11}\geq\lambda^{22}\geq\cdots\geq\lambda^{ii}\geq\cdots\geq\lambda^{dd}>0$.
Therefore, the weight vectors for eigenportfolios are 
\begin{align*}
\boldsymbol{\omega}_{j}=\frac{1}{c^{j}}\boldsymbol{\sigma}^{-1}\mathbf{v}_{j}\in\mathbb{R}^d,\quad j=1,\:2,\:\cdots,\:m\,,
\end{align*}
where $c^{j}=\mathbf{1}^{\top}\boldsymbol{\sigma}^{-1}\mathbf{v}_{j}$, $\boldsymbol{\sigma}\in\mathbb{R}^{d\times d}$ is the diagonal matrix whose diagonal elements are the standard deviations of the returns for the stocks. Consequently, the $j^{\mathrm{th}}$ factor $F^{j}$ whose returns are given by the eigenportfolio is,
\begin{align*}
\frac{dF_{t}^{j}}{F_{t}^{j}}=\sum_{i=1}^{d}\omega^{ij}\frac{dS_{t}^{i}}{S_{t}^{i}}\,,
\end{align*}
where $\omega^{ij}$ is the $i^{\mathrm{th}}$ component of vector $\boldsymbol{\omega}_{j}$. Among all the factors that are constructed utilising the eigenportfolios returns, for $j=1$, the factor $F_{t}^{1}$ is called the principal eigenportfolio.

The principal eigenportfolio $F_{t}^{1}$ should track the SPDR S\&P 500 Trust ETF. However, there exists a survivorship bias if data is collected with future knowledge of S\&P 500 constituents. In other words, because S\&P Global Incorporated adjusts the constituents of S\&P 500 Index periodically, the latest constituents have already passed a selection process based on market capitalisation. For instance, the list that we utilise for numerical experiments was downloaded\footnote{S\&P 500 constituents list: \url{https://en.wikipedia.org/wiki/List_of_S\%26P_500_companies}} on 2021-05-06, but on 2021-04-20, PTC Incorporated, whose ticker symbol is PTC, replaced\footnote{S\&P Dow Jones Indices announcement for the changes to the S\&P 500: \url{https://www.spglobal.com/spdji/en/documents/indexnews/announcements/20210415-1358567/1358567_5var-6egov-pr.pdf}} Varian Medical Systems Incorporated, whose ticker symbol is VAR, in the S\&P 500 list of constituents. By observing Figure \ref{fig:survivorship_bias}, we can see clearly that in a long-time horizon the survivorship bias is significant. Consequently, to improve the interpretability of our numerical experiments, we should adjust for this survivorship bias in the data for returns of stocks,
\begin{align*}
\frac{\Delta F_{t}^{1}}{F_{t}^{1}} 
& =\alpha_{b}\Delta t+\beta_{b}\frac{\Delta S_{t}^{\mathrm{SPY}}}{S_{t}^{\mathrm{SPY}}}+\epsilon_{t}\,
\end{align*}
where $S_{t}^{\mathrm{SPY}}$ is the daily adjusted close price of the SPDR S\&P 500 Trust ETF and $\epsilon_t$ is the difference between the returns of the principal eigenportfolio and the SPY. The survivorship adjusted stock returns are
\begin{align*}
\frac{\Delta S_{t}^{i}}{S_{t}^{i}}
& \gets\frac{\Delta S_{t}^{i}}{S_{t}^{i}}-\alpha_{b}\Delta t\,,
\end{align*}
where $i=1,\:2,\:,\cdots,\:d$. In the above formulae, $\alpha_{b}$ is the survivorship bias, because the trackability with respect to the SPDR S\&P 500 Trust ETF of the principal eigenportfolio $F_{t}^{1}$ implies that the null hypotheses is $\alpha_{b}=0$. Therefore, any non-zero value of $\alpha_{b}$ from the regression should be subtracted from the stock returns.
\begin{figure}[ht]
\noindent \begin{centering}
\includegraphics[scale=0.3]{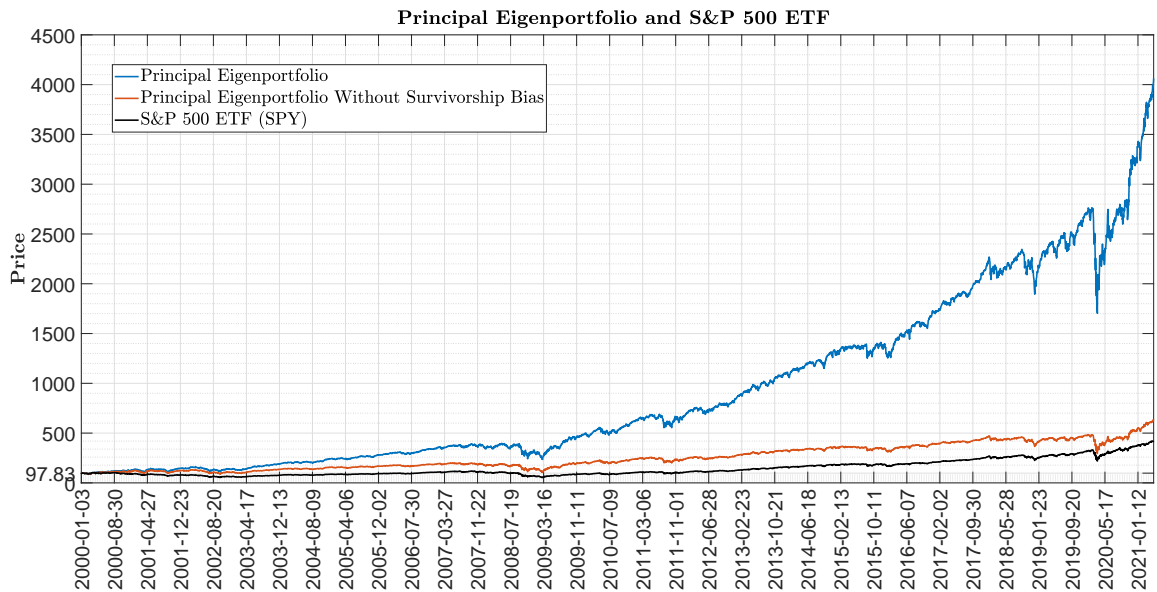}
\par\end{centering}
\caption{Survivorship bias of the principal eigenportfolio with respect to the SPDR S\&P 500 Trust ETF.\label{fig:survivorship_bias}}
\end{figure}

%%%%%%%%%%%%%%%%%%%%%%%%%%%%%%%%%%%%%%%%%%%%%%%%%%%%%%%%%%%%%%%%%%%%%%%%%%%%%%%%%%%
\subsection{Data and Parameter Estimations\label{subsec:parameter_estimation}}

We utilise Yahoo Finance\footnote{Data source for stock prices: \url{https://finance.yahoo.com}} as our data source. The data are the daily adjusted close stock prices of the S\&P 500 constituents from 2000-01-03 to 2021-05-06, which includes 5370 observations for each stock, and it also includes the SPDR S\&P 500 Trust ETF among these traded stocks. However, SPDR S\&P 500 Trust ETF is not utilised in the calculation of eigenportfolios. Hence, there are 506 ticker symbols in total, and after removing the ticker symbols that do not have full length data, we are left with 375 stocks. We assume that the interest rate $r$ is 0.01, and in every calendar year there are (around) 252 trading days. 

For the numerical experiments that we perform, the first step is to construct the eigenportfolios as the factors that are defined in equation \eqref{eq:dF}, which is described in Section \ref{subsec:eigenportfolio}. We select six factors, in other words $m=6$ in equation \eqref{eq:dF}. These six factors correspond to the largest six eigenvalues of the correlation matrix $\boldsymbol{\rho}$ of equation \eqref{eq:corr_matrix}. The drift coefficient vector $\boldsymbol{\eta}$ of the factors seen in equation \eqref{eq:dF} have $\eta_j = r$ for $j>1$ and $\eta_1$ of the principal eigenportfolio is estimated utilising the returns of the SPDR S\&P 500 Trust ETF, see \citet{boyle2014}, and the diffusion matrix $\mathbf{\Sigma}_{0}$ is estimated utilising the method described in \citet{LedoitWolf2004}.

The next step is to identify co-integrated stocks and estimate the parameters in equation \eqref{eq:factor} and equation \eqref{eq:dZ}. We first apply linear regression on the factor model of equation \eqref{eq:factor} to get the residual process $Z_{t}$, then we utilise the augmented Dickey-Fuller test to detect stationarity. Accordingly, the $\boldsymbol{\alpha}$ vector and the co-integration coefficient matrix $\boldsymbol{\beta}$ are estimated by the least squares approach of the linear regression of stock returns onto the factor returns. In equation \eqref{eq:dZ}, the reversion-mean parameter vector $\boldsymbol{\theta}$ is estimated by the time-series averages of the co-integrated processes $Z_{t}^{i}$ for $i=1,\:2,\:\cdots,\:d$. Estimating the mean-reversion speed parameter vector $\boldsymbol{\delta}$ is important because each inverse $1/\delta_{i}$ of its component represents the characteristic time-scale for mean reversion. Consequently, it affects the profit of portfolio wealth significantly, because it quantifies the duration of the trade and the inventory-risk exposure. Since the co-integrated process $Z_{t}$ is assumed to follow a stationary OU process that is described by equation \eqref{eq:dZ}, therefore, by applying the ergodic theory, the mean-reversion speed parameter $\delta$ can be estimated by the order-one auto-correlation of the process $Z_{t}$,
\begin{align*}
%\hat{\delta}=\frac{1}{\Delta t}\left(1-\frac{\sum_{t=1}^{T-1}\left(Z_{t+1}-\hat{\theta}\right)\left(Z_{t}-\hat{\theta}\right)}{\sum_{t=1}^{T}\left(Z_{t}-\hat{\theta}\right)^{2}}\right)
\hat{\delta}=-\frac{1}{\Delta t}\ln\left(\frac{\sum_{t=1}^{T-1}\left(Z_{t+1}-\hat{\theta}\right)\left(Z_{t}-\hat{\theta}\right)}{\sum_{t=1}^{T}\left(Z_{t}^{i}-\hat{\theta}^{i}\right)^{2}}\right)\,,%\label{eq:delta_estimate}
\end{align*}
where $\hat{\theta}=\frac{1}{T}\sum_{t=1}^{T}Z_{t}$ is the estimation of the mean-reversion parameter $\theta$. %For $\delta\Delta t\ll1$, we have $-\ln\left(1-\delta\Delta t\right)\thickapprox\delta\Delta t$, hence, equation \eqref{eq:delta_estimate} can be replaced by
% \begin{align*}
% \hat{\delta}=-\frac{1}{\Delta t}\ln\left(\frac{\sum_{t=1}^{T-1}\left(Z_{t+1}-\hat{\theta}\right)\left(Z_{t}-\hat{\theta}\right)}{\sum_{t=1}^{T}\left(Z_{t}^{i}-\hat{\theta}^{i}\right)^{2}}\right)\,.
% \end{align*}
Both of the estimators $\hat{\theta}$ and $\hat{\delta}$ converge to their true values in probability as $T\rightarrow+\infty$, separately. The error for this estimator is asymptotically Gaussian distributed with expectation zero and
variance approximately $2\delta$, see \citet{Kutoyants2004} for more mathematical details. 

For the data set that we utilise, in each in-sample training window, the number of co-integrated stocks varies approximately in the range of ten to eighty. We set $d\leq15$, which means that for each in-sample training window, we first sort all the co-integrated stocks in a descending order with respect to the values of the mean reversion speed $\delta$, and then select only fifteen stocks whose values of $\delta$ are the largest. If the totally number of the co-integrated stocks for this window is less than fifteen, we then choose them all. For the out-of-sample testing window that follows, we utilise the same selected stocks without knowing which stocks will remain co-integrated in the test, in other words, there is model risk in the out-of-sample testing window. By setting up the backtests in this way, we can allow both the number of co-integrated stocks and the ticker symbols to change in all in-sample training windows and out-of-sample testing windows. Figure \ref{fig:z_process} illustrates fifteen co-integrated processes $Z_{t}$ that have the largest mean-reversion speeds $\delta$ of an in-sample training window, and their $p$-values for the augmented Dickey-Fuller test are smaller than 0.01. Table \ref{tab:ticker_speed} lists the fifteen ticker symbols with fastest mean-reversion time $252/\delta$ of the selected co-integrated stocks for three in-sample training windows.

\begin{figure}[h]
\noindent \begin{centering}
\includegraphics[scale=0.3]{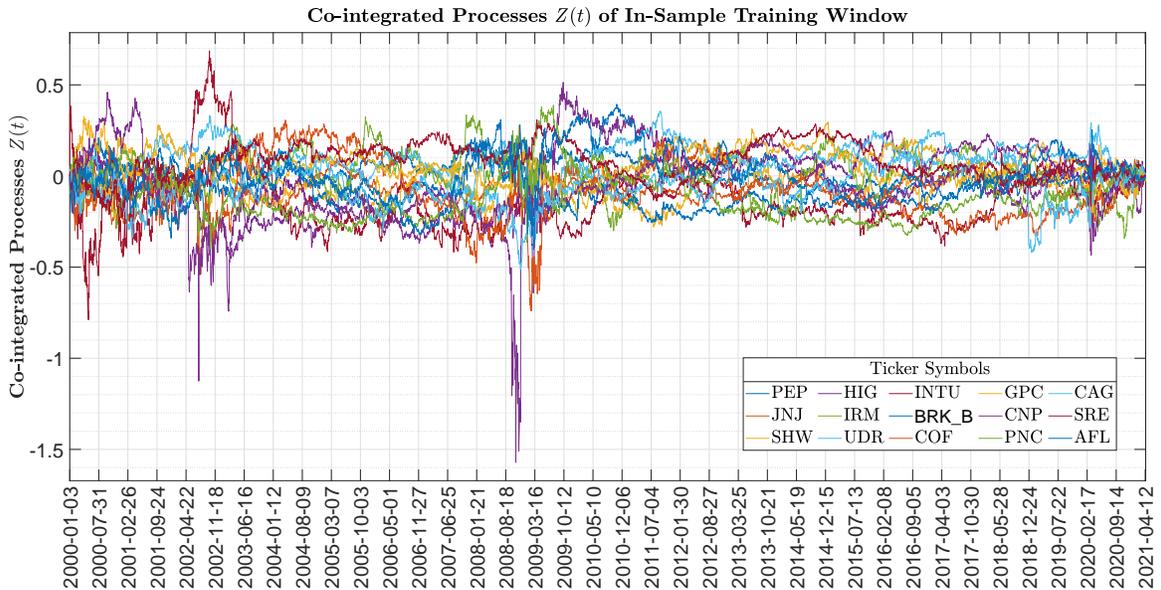}
\par\end{centering}
\caption{Co-integrated processes $Z_{t}$ of the in-sample training period 2000-01-03 to 2021-04-15. Their mean-reversion speeds $\delta$ are the fifteen largest among all the OU processes $Z_{t}$ whose augmented Dickey-Fuller tests reject the unit-root hypothesis with $p$-value $\protect\leq0.01$.\label{fig:z_process}}
\end{figure}

\begin{table}[h]
%\small
%\footnotesize
%\scriptsize
%\tiny
\setlength{\tabcolsep}{2pt}
\begin{centering}
\begin{tabular}{c|c|cccccccc}
\hline 
\multirow{3}{*}{From 2000-01-03} & Ticker Symbol & AME & AIG & MMC & ALK & MTD & KMB & SYY & MXIM\tabularnewline
 & Mean Reversion & \multirow{2}{*}{2.910} & \multirow{2}{*}{4.213} & \multirow{2}{*}{5.524} & \multirow{2}{*}{5.626} & \multirow{2}{*}{5.753} & \multirow{2}{*}{5.963} & \multirow{2}{*}{6.082} & \multirow{2}{*}{6.360}\tabularnewline
 & $\mathrm{Time}=\frac{252}{\delta}$ (Days) &  &  &  &  &  &  &  & \tabularnewline
\cline{2-10} \cline{3-10} \cline{4-10} \cline{5-10} \cline{6-10} \cline{7-10} \cline{8-10} \cline{9-10} \cline{10-10} 
\multirow{3}{*}{To 2000-11-13} & Ticker Symbol & TROW & SCHW & PH & YUM & ABT & ECL & BLK & \tabularnewline
 & Mean Reversion & \multirow{2}{*}{6.673} & \multirow{2}{*}{7.598} & \multirow{2}{*}{7.672} & \multirow{2}{*}{7.890} & \multirow{2}{*}{7.996} & \multirow{2}{*}{8.094} & \multirow{2}{*}{8.231} & \multirow{2}{*}{}\tabularnewline
 & $\mathrm{Time}=\frac{252}{\delta}$ (Days) &  &  &  &  &  &  &  & \tabularnewline
\hline 
\multirow{3}{*}{From 2010-09-07} & Ticker Symbol & ALK & NLOK & HON & BKR & RTX & PEG & CPB & OKE\tabularnewline
 & Mean Reversion & \multirow{2}{*}{6.197} & \multirow{2}{*}{6.584} & \multirow{2}{*}{7.039} & \multirow{2}{*}{7.562} & \multirow{2}{*}{7.658} & \multirow{2}{*}{8.034} & \multirow{2}{*}{8.078} & \multirow{2}{*}{9.018}\tabularnewline
 & $\mathrm{Time}=\frac{252}{\delta}$ (Days) &  &  &  &  &  &  &  & \tabularnewline
\cline{2-10} \cline{3-10} \cline{4-10} \cline{5-10} \cline{6-10} \cline{7-10} \cline{8-10} \cline{9-10} \cline{10-10} 
\multirow{3}{*}{To 2011-07-20} & Ticker Symbol & HBAN & LEN & ADI & MOS & QCOM & ED & CL & \tabularnewline
 & Mean Reversion & \multirow{2}{*}{10.460} & \multirow{2}{*}{10.708} & \multirow{2}{*}{10.718} & \multirow{2}{*}{11.018} & \multirow{2}{*}{11.021} & \multirow{2}{*}{11.365} & \multirow{2}{*}{11.521} & \multirow{2}{*}{}\tabularnewline
 & $\mathrm{Time}=\frac{252}{\delta}$ (Days) &  &  &  &  &  &  &  & \tabularnewline
\hline 
\multirow{3}{*}{From 2020-05-26} & Ticker Symbol & T & A & TSN & IBM & XRAY & UNH & PRGO & OMC\tabularnewline
 & Mean Reversion & \multirow{2}{*}{4.768} & \multirow{2}{*}{4.931} & \multirow{2}{*}{6.224} & \multirow{2}{*}{6.570} & \multirow{2}{*}{6.805} & \multirow{2}{*}{7.262} & \multirow{2}{*}{8.479} & \multirow{2}{*}{8.525}\tabularnewline
 & $\mathrm{Time}=\frac{252}{\delta}$ (Days) &  &  &  &  &  &  &  & \tabularnewline
\cline{2-10} \cline{3-10} \cline{4-10} \cline{5-10} \cline{6-10} \cline{7-10} \cline{8-10} \cline{9-10} \cline{10-10} 
\multirow{3}{*}{To 2021-04-08} & Ticker Symbol & IEX & SBAC & JPM & NEM & INCY & ECL & ANSS & \tabularnewline
 & Mean Reversion & \multirow{2}{*}{8.865} & \multirow{2}{*}{8.936} & \multirow{2}{*}{9.092} & \multirow{2}{*}{9.183} & \multirow{2}{*}{9.481} & \multirow{2}{*}{9.727} & \multirow{2}{*}{9.769} & \multirow{2}{*}{}\tabularnewline
 & $\mathrm{Time}=\frac{252}{\delta}$ (Days) &  &  &  &  &  &  &  & \tabularnewline
\hline 
\end{tabular}
\par\end{centering}
\caption{Ticker symbols and mean-reversion speeds in days for three different in-sample training windows of a sliding window in-sample training and out-of-sample testing for the time period [2000-01-03, 2021-05-06]. For each training window, the data has length 220 trading days. The selected co-integrated stocks for constructing the portfolio are the ones whose mean-reversion speeds $\delta$ are the fifteen fastest among all the OU processes $Z_{t}$ whose augmented Dickey-Fuller tests reject the unit-root hypothesis with $p$-value$\:\protect\leq0.01$.\label{tab:ticker_speed}}
\end{table}

%%%%%%%%%%%%%%%%%%%%%%%%%%%%%%%%%%%%%%%%%%%%%%%%%%%%%%%%%%%%%%%%%%%%%%%%%%%%%%%
\subsection{Portfolio Performance}
\label{subsec:portfolio_performance}

The final step before evaluating the performances of optimal portfolios is to solve the system of ordinary differential equations for the unconstrained and constrained portfolios. Because we work in the setting of very large terminal time $T$, therefore we can utilise the steady-state portfolios to demonstrate the results. In other words, we work with the limiting vector $\bar{\boldsymbol{\pi}}^{*}$ that is calculated utilising $\bar{\mathbf{C}}=\lim_{t\rightarrow-\infty}\boldsymbol{C}\left(t\right)$ and $\bar{\mathbf{b}}=\lim_{t\rightarrow-\infty}\boldsymbol{b}\left(t\right)$. In this situation, for the unconstrained problem with HJB equation \eqref{eq:unconstrainedHJB}, the matrix Riccati ODE \eqref{eq:ode_c_u} has a steady state $\bar{\mathbf{C}}$ that solves a continuous-time algebraic Riccati equation, which can be solved numerically, see \citet{Dooren1981} and \citet{ArnoldLaub1984}. The ODE \eqref{eq:ode_b_u} has a steady state $\bar{\mathbf{b}}$ that solves a linear system. The optimal control \eqref{eq:control_pi_u2} for the unconstrained portfolio under the steady state is
\begin{align*}
\bar{\boldsymbol{\pi}}^{*}\left(\boldsymbol{z}\right)=\frac{1}{1-\gamma}\mathbf{\Sigma}_{1}^{-1}\left(\boldsymbol{\mu}+\mathbf{\Sigma}_{2}\bar{\mathbf{b}}+\left(-\boldsymbol{\delta}+2\mathbf{\Sigma}_{2}\bar{\mathbf{C}}\right)\boldsymbol{z}\right)\,.
\end{align*}
Correspondingly, the optimal control \eqref{eq:control_pi2_c2} for the market-neutral constrained portfolio under the steady state is
\begin{align*}
\bar{\boldsymbol{\pi}}^{*}\left(\boldsymbol{z}\right)=\frac{1}{1-\gamma}\left(\mathbf{\Sigma}_{1}^{-1}-\mathbf{\Sigma}_{c}\right)\left(\boldsymbol{\mu}+\mathbf{\Sigma}_{2}\bar{\mathbf{b}}+\left(-\boldsymbol{\delta}+2\mathbf{\Sigma}_{2}\bar{\mathbf{C}}\right)\boldsymbol{z}\right)\,.
\end{align*}
The wealth of the portfolio calculated utilising steady-state optimal control $\bar{\boldsymbol{\pi}}^{*}$ is denoted by $W_{t}\left(\bar{\boldsymbol{\pi}}^{*}\right)$.

In order to perform comprehensive comparisons among different data sets and parameter settings, we also consider the myopic wealth process $W_{t}\left(\boldsymbol{\pi}_{m}^{*}\right)$. Utilising the myopic portfolio $\boldsymbol{\pi}_{m}^{*}$ given by equation \eqref{eq:myopic_control_u} for the unconstrained portfolio and equation \eqref{eq:myopic_control_c} for the market-neutral constrained portfolio, a myopic wealth process is computed utilising equation \eqref{eq:dW} with the optimal portfolio $\boldsymbol{\pi}^{*}$ substituted by $\boldsymbol{\pi}_{m}^{*}$. Correspondingly, the wealth of the portfolio that is calculated utilising $\bar{\boldsymbol{\pi}}_{m}^{*}$ is called the myopic wealth of steady-state portfolio $W_{t}\left(\bar{\boldsymbol{\pi}}^{*}_{m}\right)$.

We perform sliding window in-sample training and out-of-sample testing, see Figure \ref{fig:sliding_window}, for both of the unconstrained and constrained portfolios. In each in-sample training window, we identify and select co-integrated stocks, estimate parameters, and then solve the equations for steady-states values $\bar{\mathbf{C}}$ and $\bar{\mathbf{b}}$. Afterwards, in each out-of-sample testing window, we utilise the estimated parameters along with $\bar{\mathbf{C}}$ and $\bar{\mathbf{b}}$, to calculate myopic control $\bar{\boldsymbol{\pi}}_{m}^{*}$ and optimal control $\bar{\boldsymbol{\pi}}^{*}$. Over time, we compute and record the myopic wealth $W_{t}\left(\bar{\boldsymbol{\pi}}_{m}^{*}\right)$ and the optimal wealth $W_{t}\left(\bar{\boldsymbol{\pi}}^{*}\right)$.

\begin{figure}[h]
\begin{centering}
\includegraphics[scale=0.2]{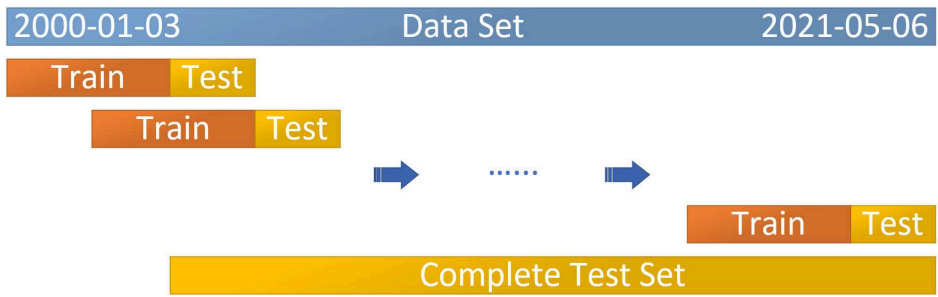}
\par\end{centering}
\caption{Sliding window in-sample training and out-of-sample testing for both of unconstrained and constrained portfolios.\label{fig:sliding_window}}
\end{figure}

We analyse the wealth of portfolios to demonstrate the long-term profitability of the multiple co-integrated stocks strategies proposed in Section \ref{sec:model_construction}. Figure \ref{fig:sliding_window_U} and Figure \ref{fig:sliding_window_C} illustrate the myopic wealth and optimal wealth with respect to unconstrained portfolio and market-neutral constrained portfolio with a fixed parameter setting, respectively. Figure \ref{fig:wealth_u} and Figure \ref{fig:wealth_c} exhibit the optimal wealth of unconstrained portfolio and market-neutral constrained portfolio for varying sliding-window lengths. Table \ref{tab:statistics_u} and Table \ref{tab:statistics_c} display the annualised statistics of expected return, volatility, profit percentage, Sharpe ratio, and maximum drawdown of myopic and optimal wealth for unconstrained and market-neutral constrained portfolios with differing sliding-window lengths. Table \ref{tab:statistics_u_1} to Table \ref{tab:statistics_u_4} and Table \ref{tab:statistics_c_1} to Table \ref{tab:statistics_c_4} show the annualised statistics of myopic and optimal wealth for unconstrained and market-neutral constrained portfolios among four different sub-intervals from year 2000 to 2021. 

As we can observe from these figures and tables, for different in-sample training window and out-of-sample testing window measured in days, comparing with the myopic portfolios, the optimal portfolios usually have bigger profits, higher expected returns, larger volatilities, and better Sharpe ratios\footnote{In this article, the annual return that is utilised for calculating Sharpe ratio is the annualised profit: $\left(1+\mathrm{Profit}\right)^{\nicefrac{365}{T}}-1$.}. But their maximum drawdowns are worse than those of myopic portfolios, which can be attributed to their higher volatility. These empirical results demonstrate that the non-myopic component in the optimal control variable $\boldsymbol{\pi}^{*}$ provides improvement to the profits and expected returns. It is also clear that the optimal market-neutral constrained portfolios, when compared to the optimal unconstrained portfolios, have lower volatilites, smaller profits and expected returns, weaker Sharpe ratios, and lower maximum drawdowns. We can also notice that the annualised statistics are sensitive to both the lengths of in-sample training window and out-of-sample testing window. This is an indication that portfolios are sensitive to parameter estimation, as window length affects the procedures for co-integration selection and parameter estimation. The results also show that trading in multiple co-integrated stocks can generate significant profits during the time periods of high volatility, for example the post Internet bubble period among 2000-2003, the financial crisis during 2007-2008, the European sovereign debt crisis between 2010-2012, and the coronavirus pandemic of 2020. The slowdown in statistical arbitrage performance in the post-bubble period running up to 2007 was analysed in \citet{khandani2011happened}. Furthermore, from the tables and figures, we can also observe that although the portfolios have high risk aversion $\gamma=-70$, the portfolios we have constructed still out-preform the SPDR S\&P 500 Trust ETF during those time periods of high volatility mentioned above. These results are a general indication that statistical arbitrage strategies have their greatest potential for profit generation when the overall level of market volatility is higher than normal.

\begin{figure}[h] 
\begin{centering}
\includegraphics[scale=0.3]{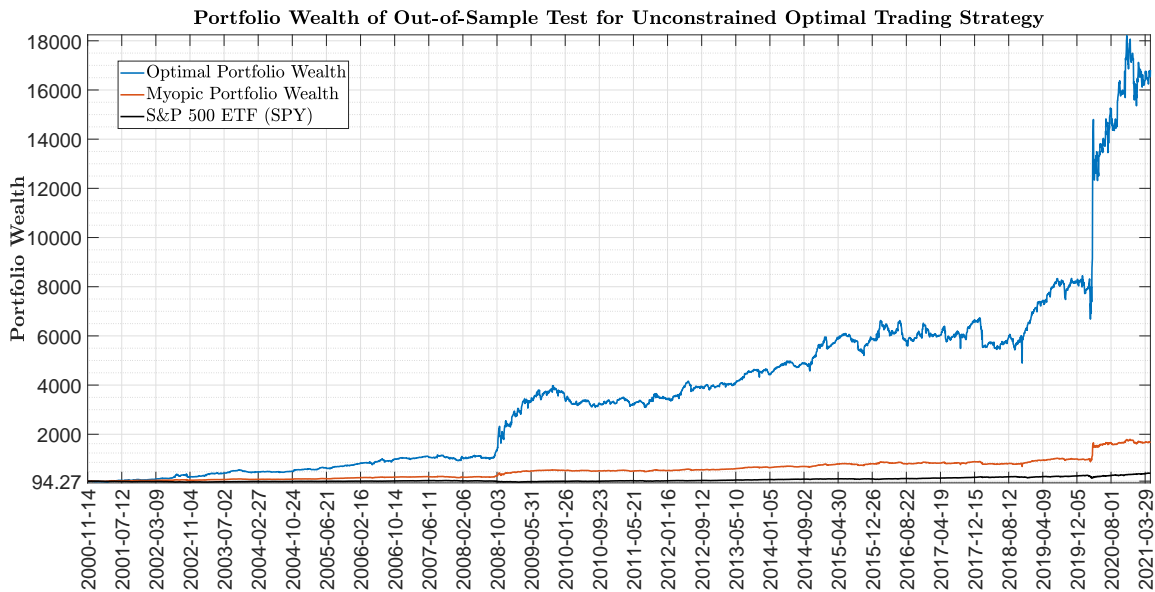}
\par\end{centering}
\caption{Sliding window out-of-sample testing for unconstrained portfolio. Training window is 220 days, testing window is 15 days, $\gamma = -100$, factor number is 6, and interest rate is $1\%$. The annualised statistics for myopic wealth are profit of $1703.167\%$, expected return of $15.458\%$, volatility of $0.163$, Sharpe ratio of $1.335$, maximum drawdown of $0.231$. The statistics for optimal wealth are profit of $17735.254\%$, expected return of $30.011\%$, volatility of $0.306$, Sharpe ratio of $1.417$, maximum drawdown of $0.445$. The statistics for the S\&P 500 ETF are profit of $340.999\%$, expected return of $9.196\%$, volatility of $0.196$, Sharpe ratio of $0.514$, maximum drawdown of $0.552$. \label{fig:sliding_window_U}}
\end{figure}

\begin{figure}[h]
\begin{centering}
\includegraphics[scale=0.3]{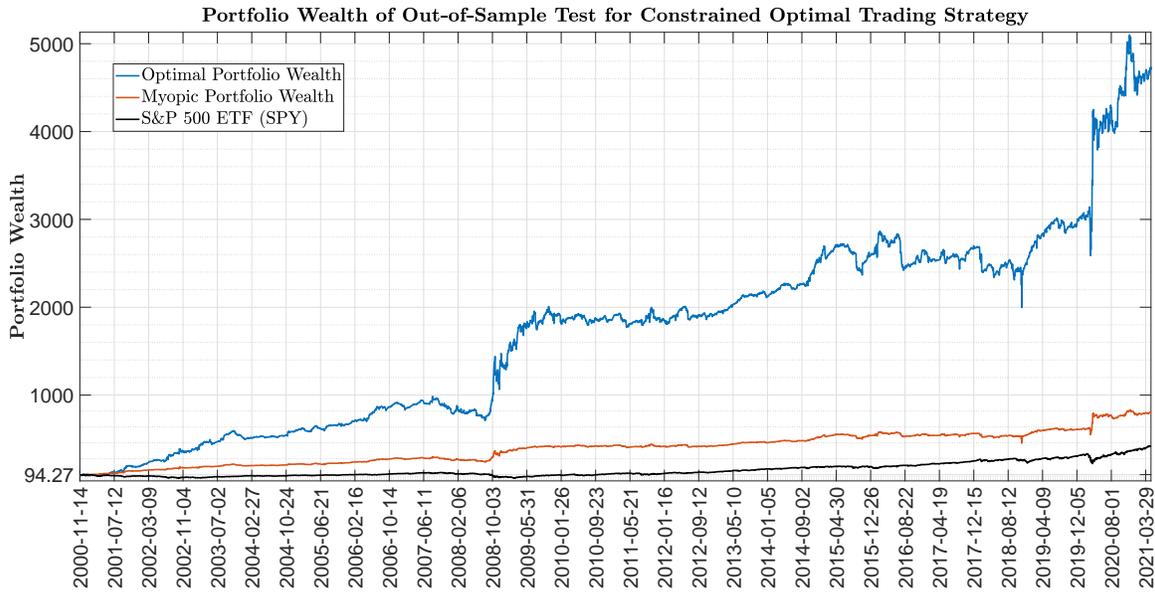}
\par\end{centering}
\caption{Sliding window out-of-sample testing for constrained portfolio to be market neutral. Training window is 220 days, testing window is 15 days, and $\gamma = -100$, factor number is 6, and interest rate is $1\%$. The annualised statistics for myopic wealth are profit of $753.696\%$, expected return of $11.195\%$, volatility of $0.119$, Sharpe ratio of $1.296$, maximum drawdown of $0.221$. The statistics for optimal wealth are profit of $4922.455\%$, expected return of $21.679\%$, volatility of $0.226$, Sharpe ratio of $1.375$, maximum drawdown of $0.303$. The statistics for the S\&P 500 ETF are profit of $340.999\%$, expected return of $9.196\%$, volatility of $0.196$, Sharpe ratio of $0.514$, maximum drawdown of $0.552$.\label{fig:sliding_window_C}}
\end{figure}

\begin{figure}[ht]
\begin{centering}
\begin{minipage}[t]{0.47\columnwidth}%
\includegraphics[angle=90,width=7.2cm,height=10cm]{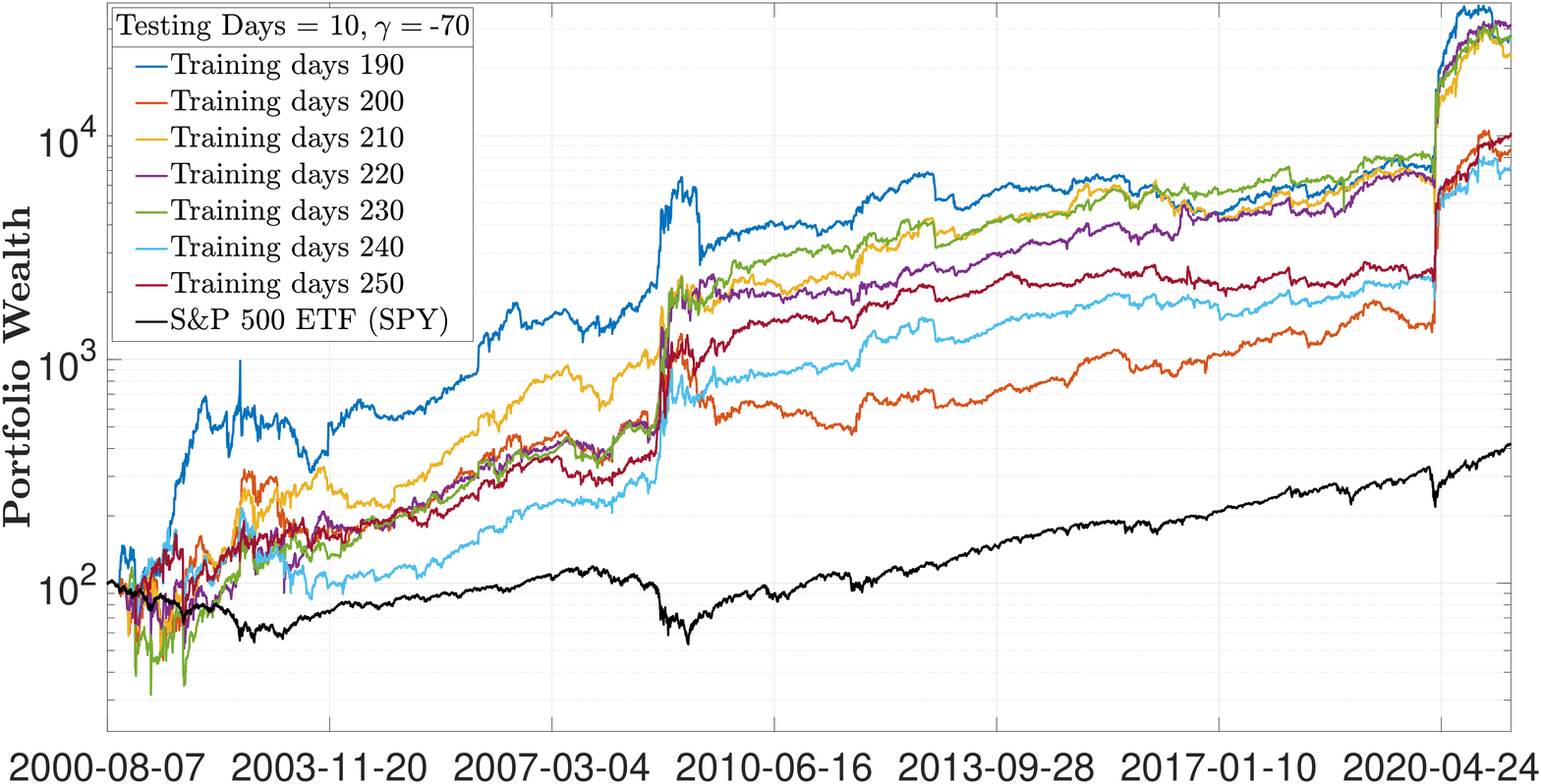}%
\end{minipage}\hspace{0.2cm}%
\begin{minipage}[t]{0.47\columnwidth}%
\includegraphics[angle=90,width=7.2cm,height=10cm]{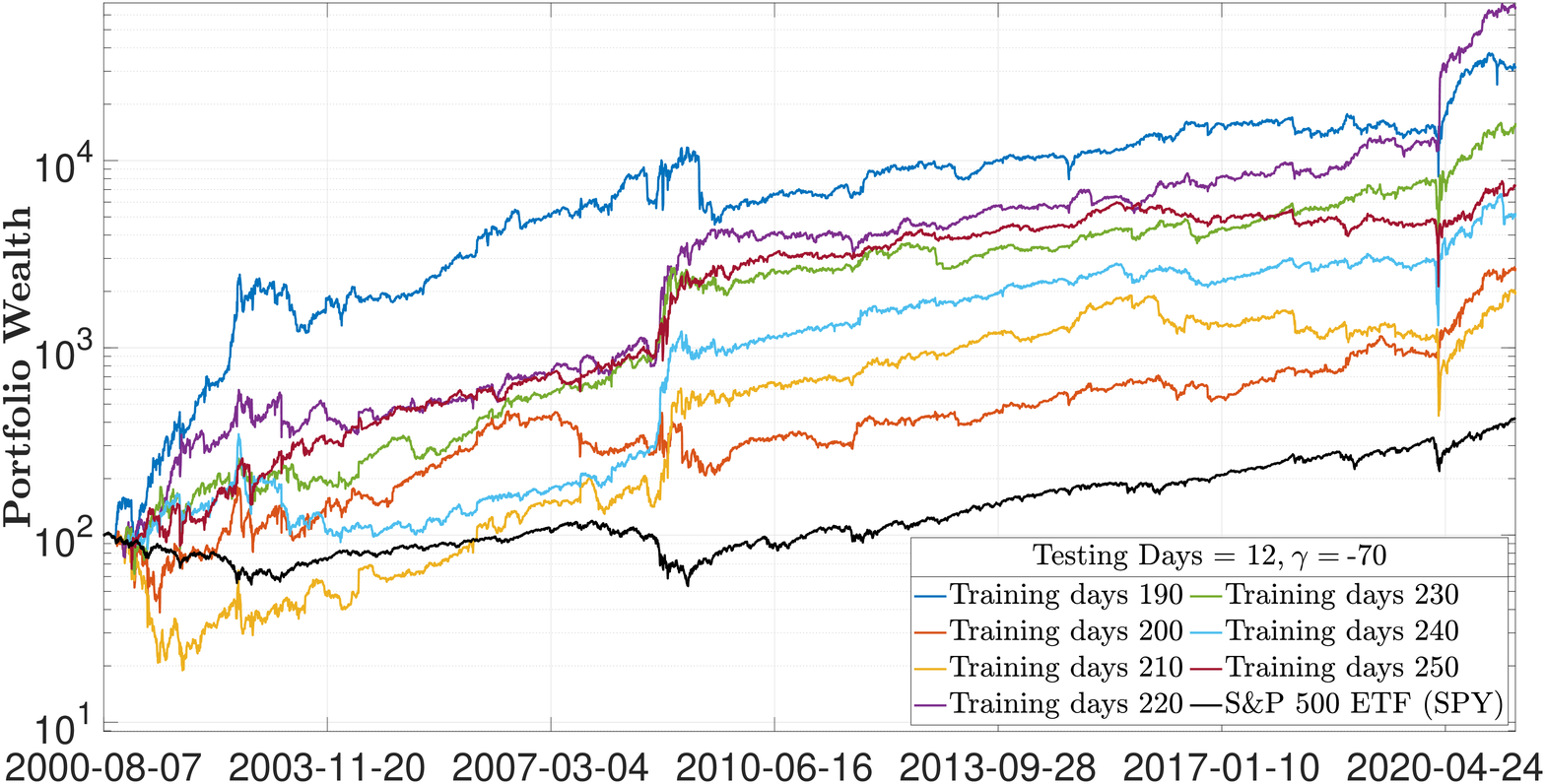}%
\end{minipage}
\par\end{centering}
\begin{centering}
\begin{minipage}[t]{0.47\columnwidth}%
\includegraphics[angle=90,width=7.2cm,height=10cm]{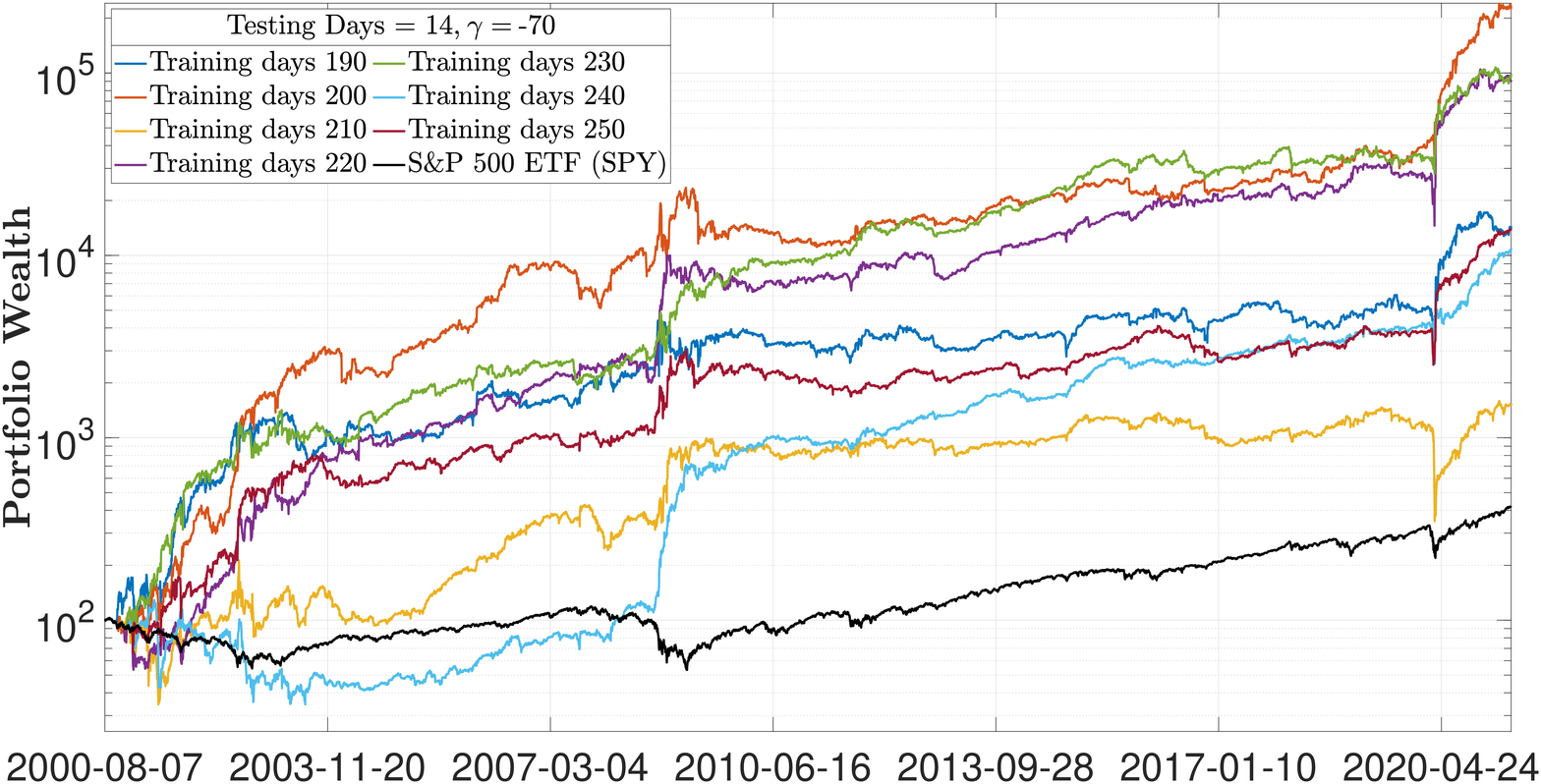}%
\end{minipage}\hspace{0.2cm}%
\begin{minipage}[t]{0.47\columnwidth}%
\includegraphics[angle=90,width=7.2cm,height=10cm]{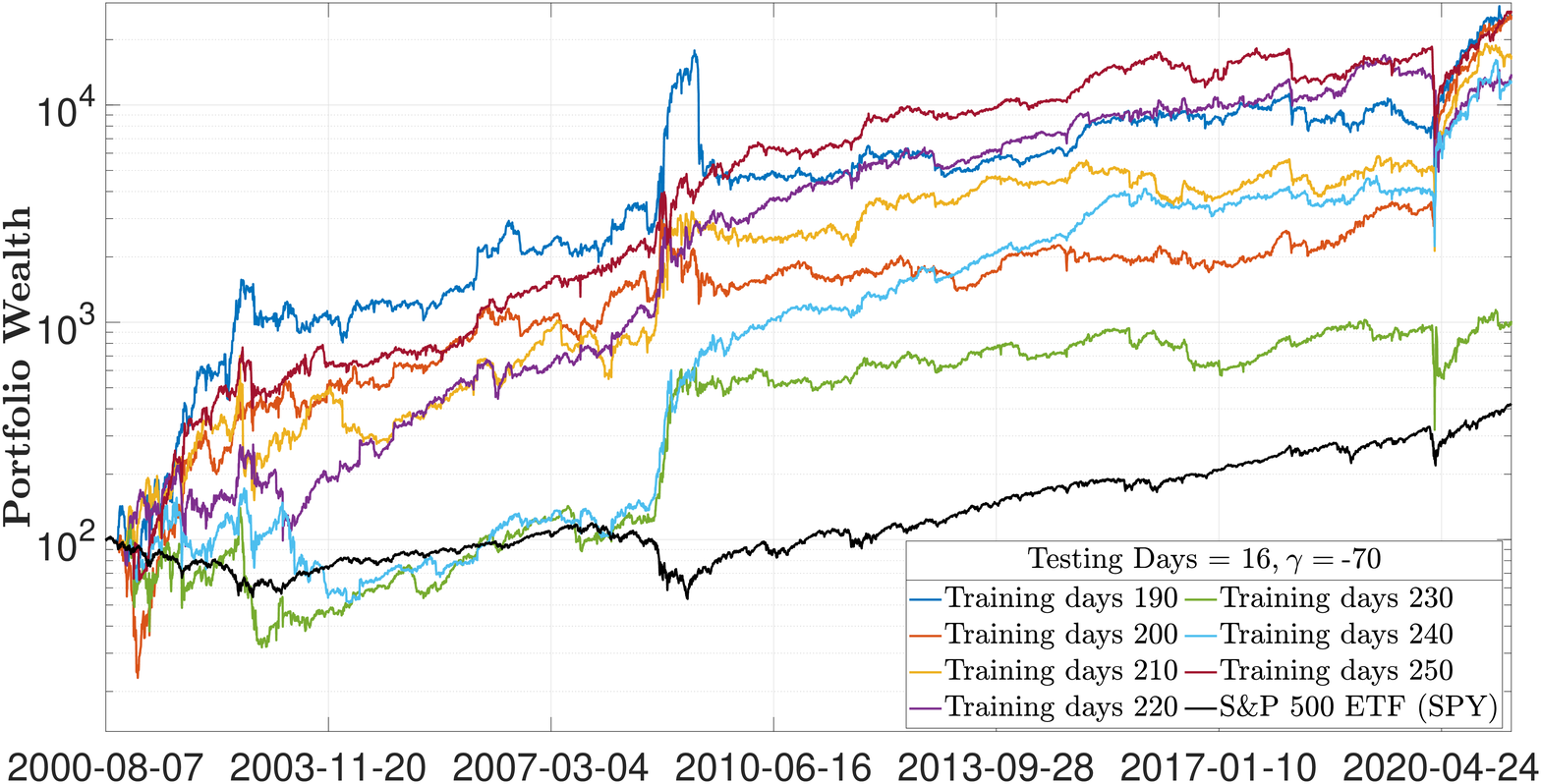}%
\end{minipage}
\par\end{centering}
\centering{}\caption{Optimal wealth trajectories for $\gamma=-70$ of the unconstrained portfolio. The vertical axis is in logarithmic scale with base 10. Interest rate is $r=1\%$. Factor number is 6.
\label{fig:wealth_u}}
\end{figure}

\begin{figure}[ht]
\begin{centering}
\begin{minipage}[t]{0.47\columnwidth}%
\includegraphics[angle=90,width=7.2cm,height=10cm]{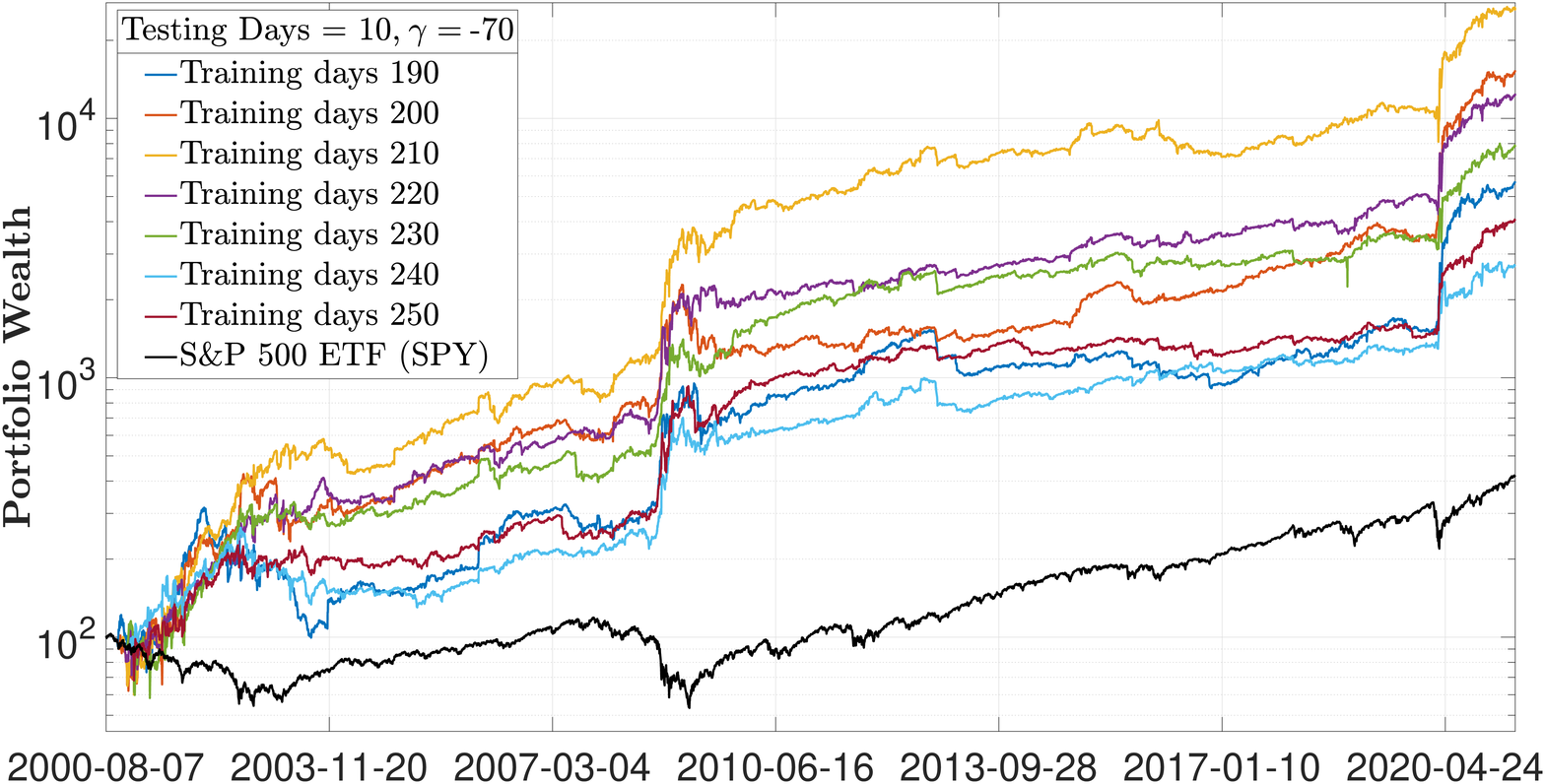}%
\end{minipage}\hspace{0.2cm}%
\begin{minipage}[t]{0.47\columnwidth}%
\includegraphics[angle=90,width=7.2cm,height=10cm]{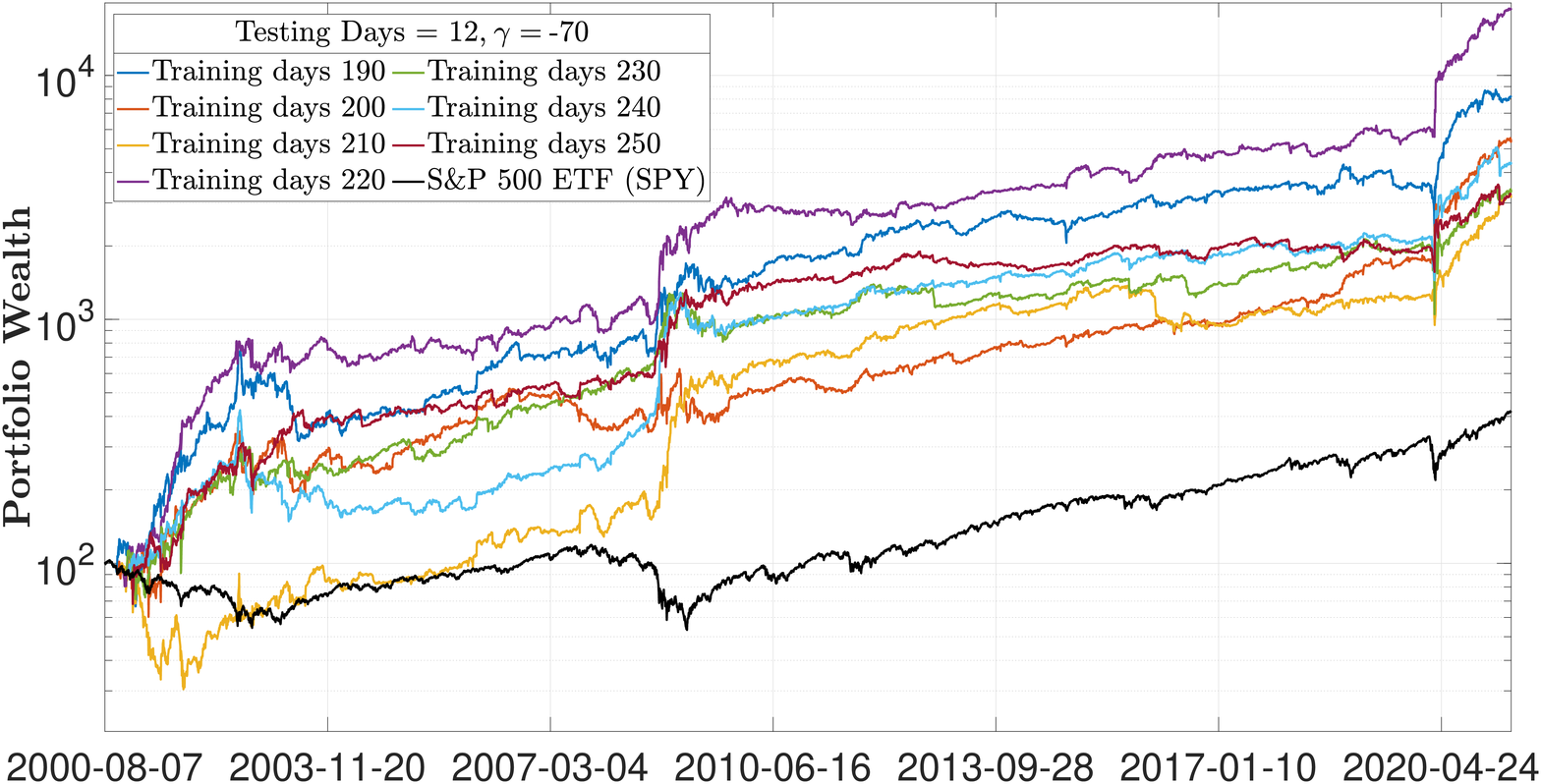}%
\end{minipage}
\par\end{centering}
\begin{centering}
\begin{minipage}[t]{0.47\columnwidth}%
\includegraphics[angle=90,width=7.2cm,height=10cm]{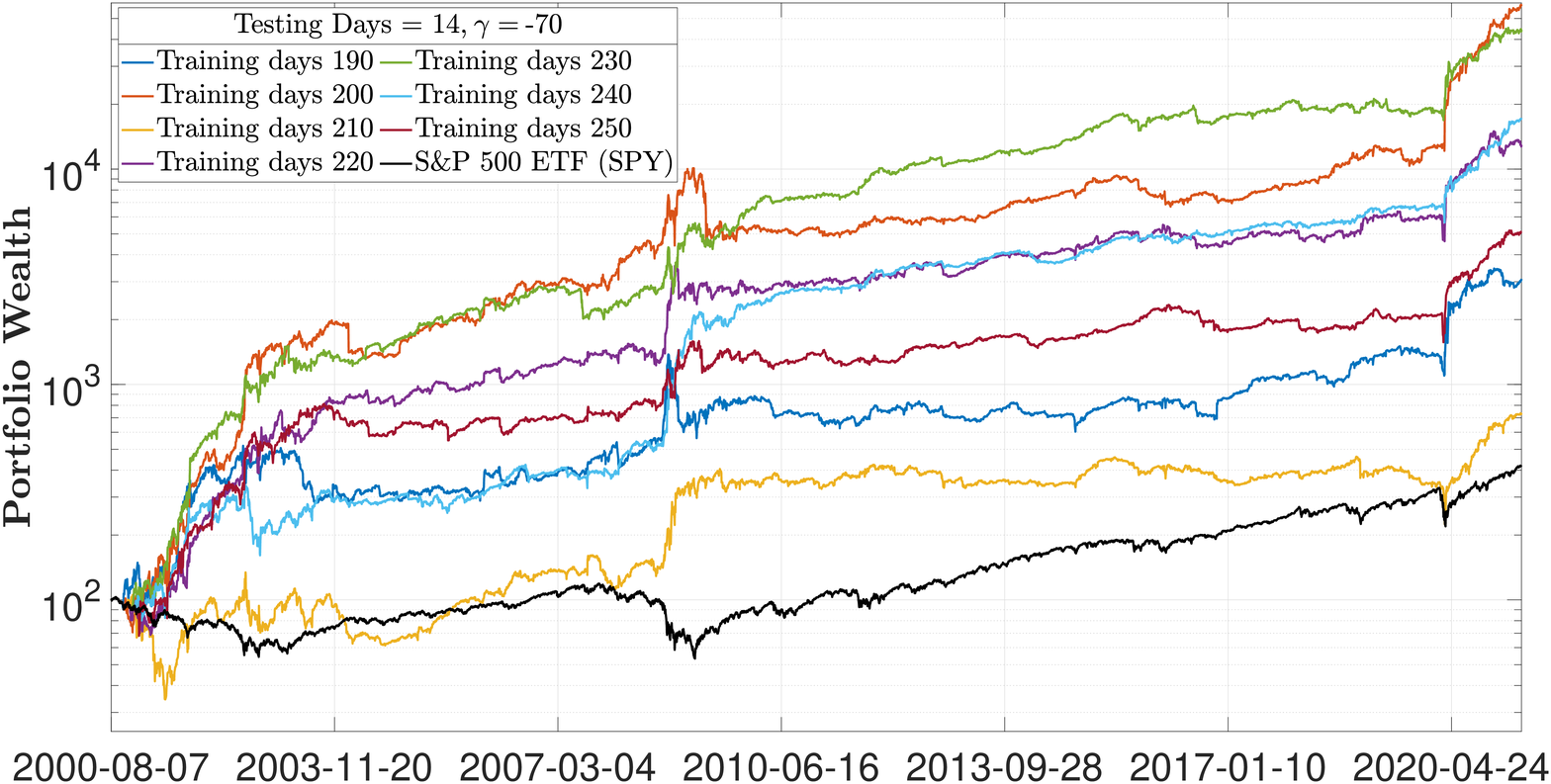}%
\end{minipage}\hspace{0.2cm}%
\begin{minipage}[t]{0.47\columnwidth}%
\includegraphics[angle=90,width=7.2cm,height=10cm]{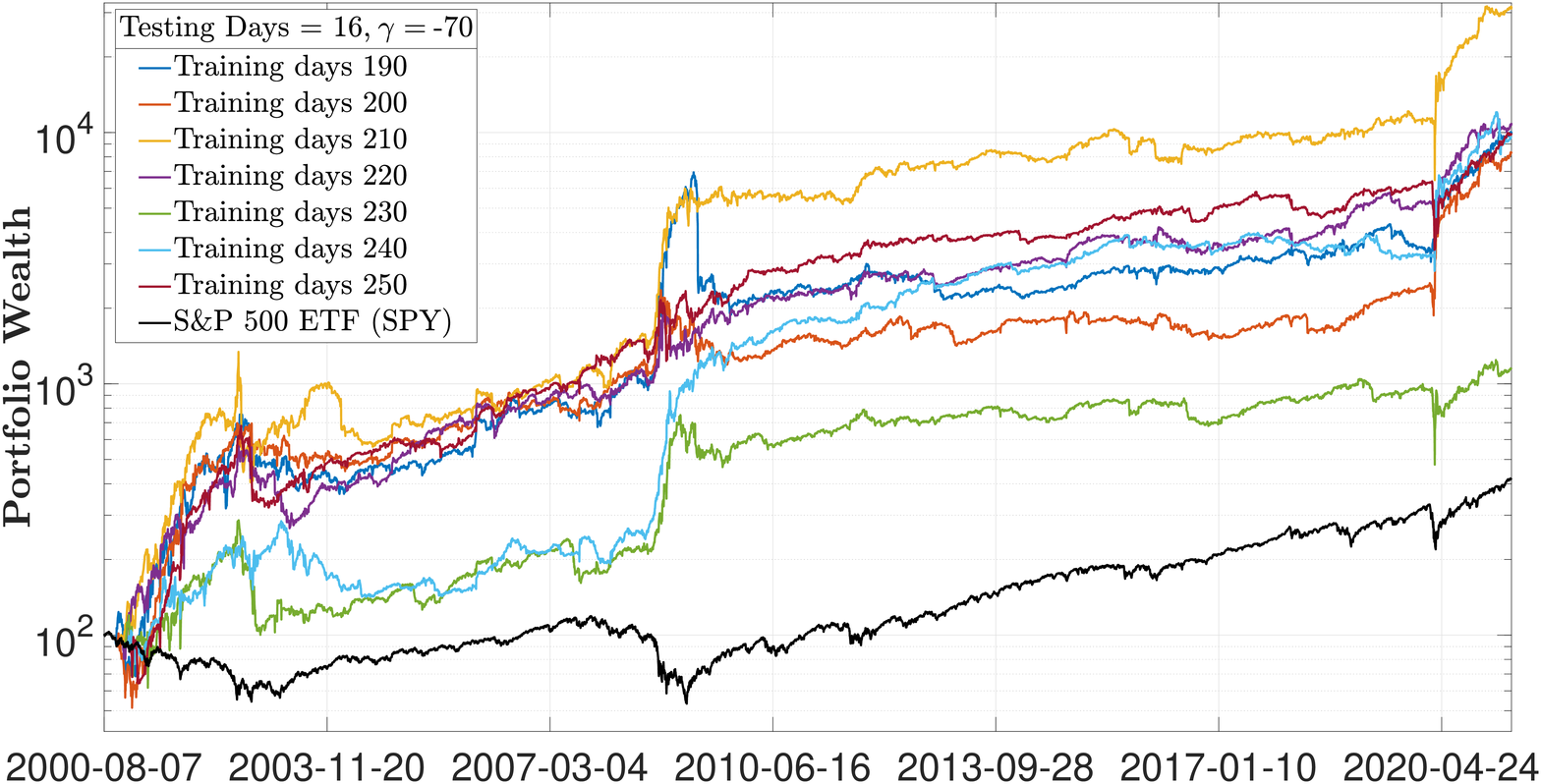}%
\end{minipage}
\par\end{centering}
\caption{Optimal wealth trajectories for $\gamma=-70$ of the constrained portfolio to be market neutral. The vertical axis is in logarithmic scale with base 10. Interest rate is $r=1\%$. Factor number is 6.
\label{fig:wealth_c}}
\end{figure}

\begin{sidewaystable}[ht]
%\begin{table}[ht]
%\small
%\footnotesize
%\scriptsize
%\tiny
\begin{centering}
\begin{tabular}{cc|cc|cc|cc|cc|cc}
\hline 
\multicolumn{2}{c|}{\diagbox[width=2.8cm]{Windows}{Statistics}} & \multicolumn{2}{c|}{Profit $\left(\%\right)$} & \multicolumn{2}{c|}{Volatility} & \multicolumn{2}{c|}{Expected Return $\left(\%\right)$} & \multicolumn{2}{c|}{Sharpe Ratio} & \multicolumn{2}{c}{Maximum Drawdown}\tabularnewline
\multicolumn{2}{c|}{} & \multicolumn{2}{c|}{$100\ensuremath{\left(\frac{W_{T}-W_{0}}{W_{0}}\right)}$} & \multicolumn{2}{c|}{$\ensuremath{\frac{1}{\sqrt{\Delta t}}\sigma\left(\frac{\Delta W_{t}}{W_{t}}\right)}$} & \multicolumn{2}{c|}{$\ensuremath{\frac{100}{\Delta t}\mathbb{E}\left(\frac{\Delta W_{t}}{W_{t}}\right)}$} & \multicolumn{2}{c|}{$\frac{0.01\left(\mathrm{Return}-r\right)}{\mathrm{Volatility}}$} & \multicolumn{2}{c}{$\underset{t}{\max}\frac{\max_{s\leq t}W_{s}-W_{t}}{\max_{s\leq t}W_{s}}$}\tabularnewline
\hline 
Train & Test & Myopic & Optimal & Myopic & Optimal & Myopic & Optimal & Myopic & Optimal & Myopic & Optimal\tabularnewline
\hline 
190 & 10 & 5614.599 & 28430.844 & 0.193 & 0.413 & 21.507 & 36.079 & 1.658 & 1.162 & 0.264 & 0.683 \tabularnewline
200 & 10 & 3374.815 & 9468.664 & 0.206 & 0.433 & 19.350 & 31.453 & 1.334 & 0.854 & 0.252 & 0.647 \tabularnewline
210 & 10 & 3676.837 & 23666.221 & 0.208 & 0.390 & 19.817 & 34.231 & 1.361 & 1.186 & 0.303 & 0.586 \tabularnewline
220 & 10 & 5610.682 & 33097.813 & 0.191 & 0.401 & 21.577 & 36.662 & 1.686 & 1.245 & 0.213 & 0.509 \tabularnewline
230 & 10 & 5290.267 & 30401.797 & 0.200 & 0.396 & 21.492 & 36.018 & 1.588 & 1.241 & 0.288 & 0.731 \tabularnewline
240 & 10 & 3652.486 & 7649.431 & 0.181 & 0.356 & 19.416 & 27.663 & 1.568 & 0.990 & 0.230 & 0.610 \tabularnewline
250 & 10 & 2845.183 & 11246.119 & 0.167 & 0.323 & 18.029 & 28.594 & 1.569 & 1.212 & 0.243 & 0.606 \tabularnewline
\hline 
190 & 12 & 4420.283 & 32891.053 & 0.264 & 0.471 & 21.678 & 38.674 & 1.130 & 1.050 & 0.473 & 0.612 \tabularnewline
200 & 12 & 2746.609 & 2819.259 & 0.218 & 0.435 & 18.566 & 25.718 & 1.177 & 0.596 & 0.327 & 0.644 \tabularnewline
210 & 12 & 1128.691 & 1948.895 & 0.207 & 0.417 & 14.431 & 23.536 & 0.889 & 0.548 & 0.634 & 0.820 \tabularnewline
220 & 12 & 7008.658 & 70308.312 & 0.197 & 0.368 & 22.741 & 38.813 & 1.741 & 1.582 & 0.208 & 0.443 \tabularnewline
230 & 12 & 5259.303 & 17427.372 & 0.240 & 0.418 & 22.252 & 33.835 & 1.323 & 1.036 & 0.541 & 0.650 \tabularnewline
240 & 12 & 2828.002 & 5578.698 & 0.179 & 0.343 & 18.182 & 25.737 & 1.460 & 0.941 & 0.444 & 0.735 \tabularnewline
250 & 12 & 1840.287 & 8077.588 & 0.196 & 0.354 & 16.401 & 27.798 & 1.152 & 1.013 & 0.467 & 0.645 \tabularnewline
\hline 
190 & 14 & 2414.485 & 14856.244 & 0.321 & 0.551 & 20.182 & 38.205 & 0.764 & 0.750 & 0.514 & 0.581 \tabularnewline
200 & 14 & 12354.895 & 255582.935 & 0.233 & 0.456 & 26.111 & 48.335 & 1.697 & 1.602 & 0.350 & 0.582 \tabularnewline
210 & 14 & 1764.904 & 1488.156 & 0.221 & 0.425 & 16.736 & 22.622 & 0.997 & 0.485 & 0.641 & 0.762 \tabularnewline
220 & 14 & 9537.631 & 96750.772 & 0.228 & 0.395 & 24.794 & 41.083 & 1.632 & 1.564 & 0.432 & 0.552 \tabularnewline
230 & 14 & 8419.730 & 109588.211 & 0.195 & 0.376 & 23.674 & 41.291 & 1.854 & 1.688 & 0.259 & 0.397 \tabularnewline
240 & 14 & 5833.636 & 11760.715 & 0.205 & 0.371 & 22.012 & 30.122 & 1.593 & 1.064 & 0.338 & 0.732 \tabularnewline
250 & 14 & 3297.981 & 15333.644 & 0.187 & 0.361 & 19.065 & 31.266 & 1.474 & 1.171 & 0.266 & 0.444  \tabularnewline
\hline 
190 & 16 & 4506.404 & 26938.729 & 0.299 & 0.555 & 22.667 & 41.564 & 1.004 & 0.855 & 0.476 & 0.771 \tabularnewline
200 & 16 & 4562.667 & 28012.313 & 0.234 & 0.447 & 21.398 & 37.352 & 1.288 & 1.071 & 0.402 & 0.779 \tabularnewline
210 & 16 & 5674.993 & 17314.723 & 0.319 & 0.517 & 24.260 & 37.600 & 1.009 & 0.832 & 0.553 & 0.783 \tabularnewline
220 & 16 & 5116.692 & 14598.738 & 0.231 & 0.456 & 21.968 & 34.896 & 1.360 & 0.908 & 0.538 & 0.709 \tabularnewline
230 & 16 & 1478.791 & 1021.100 & 0.247 & 0.472 & 16.477 & 23.045 & 0.835 & 0.375 & 0.600 & 0.772 \tabularnewline
240 & 16 & 5442.467 & 13908.472 & 0.227 & 0.411 & 22.187 & 32.524 & 1.415 & 1.002 & 0.420 & 0.703 \tabularnewline
250 & 16 & 4273.897 & 29776.589 & 0.189 & 0.364 & 20.417 & 34.671 & 1.580 & 1.351 & 0.501 & 0.622  \tabularnewline
\hline 
\end{tabular}
\par\end{centering}
\centering{}\caption{Annualised statistics of time interval [2000-01-03, 2021-05-06] of myopic wealth and optimal wealth for unconstrained portfolio with $\gamma=-70$. Interest rate is $r=1\%$. Factor number is 6. Among all the different testing-training window combinations, the statistics for the S\&P 500 ETF, which has the largest profit are profit of $359.795\%$, volatility of $0.196$, expected return of $9.414\%$, Sharpe ratio of $0.532$, and maximum drawdown of $0.552$.}\label{tab:statistics_u}
%\end{table}
\end{sidewaystable}
%\footnotetext{aaaa}

\begin{sidewaystable}[ht]
%\begin{table}[ht]
%\small
%\footnotesize
%\scriptsize
%\tiny
\begin{centering}
\begin{tabular}{cc|cc|cc|cc|cc|cc}
\hline 
\multicolumn{2}{c|}{\diagbox[width=2.8cm]{Windows}{Statistics}} & \multicolumn{2}{c|}{Profit $\left(\%\right)$} & \multicolumn{2}{c|}{Volatility} & \multicolumn{2}{c|}{Expected Return $\left(\%\right)$} & \multicolumn{2}{c|}{Sharpe Ratio} & \multicolumn{2}{c}{Maximum Drawdown}\tabularnewline
\multicolumn{2}{c|}{} & \multicolumn{2}{c|}{$100\ensuremath{\left(\frac{W_{T}-W_{0}}{W_{0}}\right)}$} & \multicolumn{2}{c|}{$\ensuremath{\frac{1}{\sqrt{\Delta t}}\sigma\left(\frac{\Delta W_{t}}{W_{t}}\right)}$} & \multicolumn{2}{c|}{$\ensuremath{\frac{100}{\Delta t}\mathbb{E}\left(\frac{\Delta W_{t}}{W_{t}}\right)}$} & \multicolumn{2}{c|}{$\frac{0.01\left(\mathrm{Return}-r\right)}{\mathrm{Volatility}}$} & \multicolumn{2}{c}{$\underset{t}{\max}\frac{\max_{s\leq t}W_{s}-W_{t}}{\max_{s\leq t}W_{s}}$}\tabularnewline
\hline 
Train & Test & Myopic & Optimal & Myopic & Optimal & Myopic & Optimal & Myopic & Optimal & Myopic & Optimal\tabularnewline
\hline 
190 & 10 & 269.189 & 762.973 & 0.202 & 0.599 & 27.514 & 60.595 & 2.153 & 1.379 & 0.236 & 0.683 \tabularnewline
200 & 10 & 175.525 & 239.890 & 0.224 & 0.652 & 22.285 & 45.156 & 1.435 & 0.618 & 0.252 & 0.604 \tabularnewline
210 & 10 & 151.609 & 377.430 & 0.198 & 0.532 & 20.026 & 45.037 & 1.455 & 1.029 & 0.226 & 0.586 \tabularnewline
220 & 10 & 161.404 & 217.664 & 0.201 & 0.599 & 20.855 & 41.913 & 1.506 & 0.631 & 0.213 & 0.509 \tabularnewline
230 & 10 & 126.343 & 222.963 & 0.199 & 0.593 & 18.022 & 41.636 & 1.263 & 0.650 & 0.231 & 0.731 \tabularnewline
240 & 10 & 105.592 & 64.080 & 0.181 & 0.494 & 15.820 & 22.042 & 1.203 & 0.286 & 0.230 & 0.610 \tabularnewline
250 & 10 & 156.828 & 189.011 & 0.179 & 0.460 & 20.217 & 32.144 & 1.673 & 0.747 & 0.243 & 0.606 \tabularnewline
\hline 
190 & 12 & 434.238 & 3042.101 & 0.232 & 0.616 & 35.367 & 86.106 & 2.560 & 2.652 & 0.268 & 0.573 \tabularnewline
200 & 12 & 285.006 & 215.004 & 0.203 & 0.625 & 28.381 & 42.271 & 2.237 & 0.597 & 0.212 & 0.644 \tabularnewline
210 & 12 & 65.927 & -12.171 & 0.219 & 0.598 & 12.299 & 15.673 & 0.657 & -0.077 & 0.244 & 0.820 \tabularnewline
220 & 12 & 198.079 & 476.881 & 0.218 & 0.521 & 23.721 & 48.340 & 1.619 & 1.217 & 0.208 & 0.443 \tabularnewline
230 & 12 & 187.898 & 300.986 & 0.241 & 0.573 & 23.615 & 43.701 & 1.413 & 0.828 & 0.206 & 0.531 \tabularnewline
240 & 12 & 114.100 & 43.943 & 0.173 & 0.475 & 16.471 & 18.614 & 1.341 & 0.209 & 0.205 & 0.735 \tabularnewline
250 & 12 & 204.540 & 544.735 & 0.172 & 0.467 & 23.447 & 48.015 & 2.120 & 1.483 & 0.178 & 0.435  \tabularnewline
\hline 
190 & 14 & 203.056 & 1184.984 & 0.245 & 0.649 & 24.675 & 71.424 & 1.459 & 1.609 & 0.346 & 0.552 \tabularnewline
200 & 14 & 450.696 & 4240.781 & 0.244 & 0.634 & 36.294 & 93.593 & 2.496 & 2.985 & 0.302 & 0.422 \tabularnewline
210 & 14 & 253.394 & 92.338 & 0.239 & 0.607 & 27.449 & 31.511 & 1.755 & 0.319 & 0.269 & 0.673 \tabularnewline
220 & 14 & 260.782 & 1365.562 & 0.200 & 0.503 & 27.155 & 65.187 & 2.144 & 2.251 & 0.233 & 0.465 \tabularnewline
230 & 14 & 430.667 & 1917.827 & 0.195 & 0.511 & 34.681 & 72.071 & 3.069 & 2.622 & 0.171 & 0.368 \tabularnewline
240 & 14 & 94.496 & -38.989 & 0.207 & 0.523 & 15.209 & 4.384 & 0.958 & -0.270 & 0.217 & 0.732 \tabularnewline
250 & 14 & 274.299 & 855.966 & 0.200 & 0.527 & 28.009 & 58.516 & 2.242 & 1.697 & 0.239 & 0.405  \tabularnewline
\hline 
190 & 16 & 168.095 & 1400.592 & 0.253 & 0.677 & 22.445 & 76.018 & 1.226 & 1.678 & 0.276 & 0.567 \tabularnewline
200 & 16 & 175.055 & 825.837 & 0.245 & 0.635 & 22.742 & 63.564 & 1.309 & 1.362 & 0.402 & 0.779 \tabularnewline
210 & 16 & 361.257 & 426.935 & 0.242 & 0.653 & 32.849 & 54.512 & 2.197 & 0.906 & 0.320 & 0.783 \tabularnewline
220 & 16 & 267.150 & 467.307 & 0.229 & 0.649 & 28.128 & 55.330 & 1.900 & 0.964 & 0.282 & 0.638 \tabularnewline
230 & 16 & 210.867 & -9.694 & 0.230 & 0.690 & 24.924 & 22.927 & 1.613 & -0.056 & 0.269 & 0.772 \tabularnewline
240 & 16 & 63.550 & -11.986 & 0.245 & 0.577 & 12.623 & 14.244 & 0.573 & -0.079 & 0.338 & 0.703 \tabularnewline
250 & 16 & 425.794 & 1088.499 & 0.200 & 0.509 & 34.719 & 61.585 & 2.975 & 1.998 & 0.165 & 0.474  \tabularnewline
\hline 
\end{tabular}
\par\end{centering}
\centering{}\caption{Annualised statistics of sub-interval [2000-08-07, 2005-10-13] of myopic wealth and optimal wealth for unconstrained portfolio with $\gamma=-70$. Factor number is 6. Interest rate is $r=1\%$. Among all the different testing-training window combinations, the statistics for the S\&P 500 ETF, which has the largest profit are profit of $5.078\%$, volatility of $0.186$, expected return of $2.693\%$, Sharpe ratio of $0.023$, and maximum drawdown of $0.421$.}\label{tab:statistics_u_1}
%\end{table}
\end{sidewaystable}

\begin{sidewaystable}[ht]
%\begin{table}[ht]
%\small
%\footnotesize
%\scriptsize
%\tiny
\begin{centering}
\begin{tabular}{cc|cc|cc|cc|cc|cc}
\hline 
\multicolumn{2}{c|}{\diagbox[width=2.8cm]{Windows}{Statistics}} & \multicolumn{2}{c|}{Profit $\left(\%\right)$} & \multicolumn{2}{c|}{Volatility} & \multicolumn{2}{c|}{Expected Return $\left(\%\right)$} & \multicolumn{2}{c|}{Sharpe Ratio} & \multicolumn{2}{c}{Maximum Drawdown}\tabularnewline
\multicolumn{2}{c|}{} & \multicolumn{2}{c|}{$100\ensuremath{\left(\frac{W_{T}-W_{0}}{W_{0}}\right)}$} & \multicolumn{2}{c|}{$\ensuremath{\frac{1}{\sqrt{\Delta t}}\sigma\left(\frac{\Delta W_{t}}{W_{t}}\right)}$} & \multicolumn{2}{c|}{$\ensuremath{\frac{100}{\Delta t}\mathbb{E}\left(\frac{\Delta W_{t}}{W_{t}}\right)}$} & \multicolumn{2}{c|}{$\frac{0.01\left(\mathrm{Return}-r\right)}{\mathrm{Volatility}}$} & \multicolumn{2}{c}{$\underset{t}{\max}\frac{\max_{s\leq t}W_{s}-W_{t}}{\max_{s\leq t}W_{s}}$}\tabularnewline
\hline 
Train & Test & Myopic & Optimal & Myopic & Optimal & Myopic & Optimal & Myopic & Optimal & Myopic & Optimal\tabularnewline
\hline 
190 & 10 & 328.705 & 373.889 & 0.211 & 0.430 & 30.559 & 39.425 & 2.357 & 1.259 & 0.196 & 0.595 \tabularnewline
200 & 10 & 121.789 & 70.556 & 0.218 & 0.413 & 17.902 & 18.758 & 1.110 & 0.370 & 0.240 & 0.608 \tabularnewline
210 & 10 & 141.629 & 384.428 & 0.225 & 0.387 & 19.781 & 38.253 & 1.219 & 1.429 & 0.303 & 0.370 \tabularnewline
220 & 10 & 246.856 & 555.269 & 0.216 & 0.392 & 26.671 & 44.349 & 1.912 & 1.771 & 0.186 & 0.338 \tabularnewline
230 & 10 & 362.727 & 946.501 & 0.192 & 0.347 & 31.923 & 52.030 & 2.795 & 2.709 & 0.250 & 0.324 \tabularnewline
240 & 10 & 287.123 & 478.952 & 0.201 & 0.373 & 28.649 & 41.468 & 2.285 & 1.714 & 0.161 & 0.322 \tabularnewline
250 & 10 & 226.053 & 470.042 & 0.187 & 0.340 & 25.026 & 40.038 & 2.096 & 1.860 & 0.190 & 0.371 
 \tabularnewline
\hline 
190 & 12 & 205.653 & 115.748 & 0.241 & 0.474 & 24.623 & 26.048 & 1.493 & 0.490 & 0.264 & 0.612 \tabularnewline
200 & 12 & 61.391 & 8.194 & 0.242 & 0.429 & 12.243 & 10.782 & 0.556 & 0.029 & 0.282 & 0.546 \tabularnewline
210 & 12 & 241.410 & 685.224 & 0.232 & 0.428 & 26.634 & 48.988 & 1.750 & 1.830 & 0.237 & 0.361 \tabularnewline
220 & 12 & 241.752 & 625.207 & 0.224 & 0.386 & 26.510 & 45.952 & 1.819 & 1.929 & 0.168 & 0.298 \tabularnewline
230 & 12 & 353.907 & 624.522 & 0.186 & 0.352 & 31.415 & 44.989 & 2.836 & 2.117 & 0.108 & 0.308 \tabularnewline
240 & 12 & 358.201 & 823.563 & 0.186 & 0.330 & 31.641 & 48.959 & 2.858 & 2.650 & 0.155 & 0.291 \tabularnewline
250 & 12 & 217.307 & 427.378 & 0.206 & 0.355 & 24.851 & 38.956 & 1.848 & 1.683 & 0.226 & 0.369  \tabularnewline
\hline 
190 & 14 & 208.058 & 155.534 & 0.403 & 0.675 & 28.880 & 37.660 & 0.903 & 0.434 & 0.345 & 0.504 \tabularnewline
200 & 14 & 195.598 & 192.751 & 0.242 & 0.466 & 24.077 & 31.892 & 1.439 & 0.740 & 0.350 & 0.582 \tabularnewline
210 & 14 & 232.307 & 335.514 & 0.243 & 0.430 & 26.355 & 37.557 & 1.624 & 1.178 & 0.327 & 0.436 \tabularnewline
220 & 14 & 274.081 & 451.882 & 0.225 & 0.386 & 28.314 & 40.697 & 1.974 & 1.588 & 0.180 & 0.371 \tabularnewline
230 & 14 & 166.951 & 397.124 & 0.229 & 0.423 & 21.891 & 40.281 & 1.364 & 1.342 & 0.259 & 0.397 \tabularnewline
240 & 14 & 463.927 & 1513.044 & 0.190 & 0.333 & 35.805 & 60.064 & 3.296 & 3.593 & 0.172 & 0.195 \tabularnewline
250 & 14 & 115.122 & 125.401 & 0.219 & 0.374 & 17.422 & 22.873 & 1.072 & 0.671 & 0.177 & 0.379  \tabularnewline
\hline 
190 & 16 & 412.601 & 216.950 & 0.365 & 0.662 & 37.871 & 42.750 & 1.578 & 0.566 & 0.277 & 0.771 \tabularnewline
200 & 16 & 138.883 & 97.803 & 0.251 & 0.458 & 20.097 & 23.752 & 1.070 & 0.442 & 0.237 & 0.591 \tabularnewline
210 & 16 & 240.030 & 402.973 & 0.292 & 0.470 & 28.095 & 42.166 & 1.384 & 1.212 & 0.341 & 0.462 \tabularnewline
220 & 16 & 253.144 & 722.470 & 0.234 & 0.390 & 27.461 & 48.807 & 1.793 & 2.070 & 0.267 & 0.339 \tabularnewline
230 & 16 & 214.321 & 505.117 & 0.208 & 0.391 & 24.634 & 42.820 & 1.802 & 1.682 & 0.327 & 0.381 \tabularnewline
240 & 16 & 484.534 & 1273.849 & 0.215 & 0.377 & 37.044 & 58.652 & 2.989 & 2.912 & 0.195 & 0.358 \tabularnewline
250 & 16 & 216.752 & 486.400 & 0.218 & 0.381 & 25.075 & 41.991 & 1.744 & 1.698 & 0.261 & 0.425  \tabularnewline
\hline 
\end{tabular}
\par\end{centering}
\centering{}\caption{Annualised statistics of sub-interval [2005-10-14, 2010-12-17] of myopic wealth and optimal wealth for unconstrained portfolio with $\gamma=-70$. Interest rate is $r=1\%$. Factor number is 6. Among all the different testing-training window combinations, the statistics for the S\&P 500 ETF, which has the largest profit are profit of $17.725\%$, volatility of $0.247$, expected return of $6.232\%$, Sharpe ratio of $0.152$, and maximum drawdown of $0.552$.}\label{tab:statistics_u_2}
%\end{table}
\end{sidewaystable}

\begin{sidewaystable}[ht]
%\begin{table}[ht]
%\small
%\footnotesize
%\scriptsize
%\tiny
\begin{centering}
\begin{tabular}{cc|cc|cc|cc|cc|cc}
\hline 
\multicolumn{2}{c|}{\diagbox[width=2.8cm]{Windows}{Statistics}} & \multicolumn{2}{c|}{Profit $\left(\%\right)$} & \multicolumn{2}{c|}{Volatility} & \multicolumn{2}{c|}{Expected Return $\left(\%\right)$} & \multicolumn{2}{c|}{Sharpe Ratio} & \multicolumn{2}{c}{Maximum Drawdown}\tabularnewline
\multicolumn{2}{c|}{} & \multicolumn{2}{c|}{$100\ensuremath{\left(\frac{W_{T}-W_{0}}{W_{0}}\right)}$} & \multicolumn{2}{c|}{$\ensuremath{\frac{1}{\sqrt{\Delta t}}\sigma\left(\frac{\Delta W_{t}}{W_{t}}\right)}$} & \multicolumn{2}{c|}{$\ensuremath{\frac{100}{\Delta t}\mathbb{E}\left(\frac{\Delta W_{t}}{W_{t}}\right)}$} & \multicolumn{2}{c|}{$\frac{0.01\left(\mathrm{Return}-r\right)}{\mathrm{Volatility}}$} & \multicolumn{2}{c}{$\underset{t}{\max}\frac{\max_{s\leq t}W_{s}-W_{t}}{\max_{s\leq t}W_{s}}$}\tabularnewline
\hline 
Train & Test & Myopic & Optimal & Myopic & Optimal & Myopic & Optimal & Myopic & Optimal & Myopic & Optimal\tabularnewline
\hline 
190 & 10 & 14.969 & 31.075 & 0.113 & 0.164 & 3.364 & 6.628 & 0.266 & 0.422 & 0.242 & 0.328 \tabularnewline
200 & 10 & 30.881 & 67.641 & 0.114 & 0.176 & 5.908 & 11.620 & 0.603 & 0.835 & 0.173 & 0.219 \tabularnewline
210 & 10 & 83.136 & 132.421 & 0.118 & 0.186 & 12.539 & 18.225 & 1.495 & 1.393 & 0.168 & 0.215 \tabularnewline
220 & 10 & 64.793 & 87.488 & 0.111 & 0.169 & 10.407 & 13.744 & 1.278 & 1.093 & 0.160 & 0.203 \tabularnewline
230 & 10 & 59.327 & 81.535 & 0.109 & 0.157 & 9.748 & 12.943 & 1.202 & 1.113 & 0.170 & 0.254 \tabularnewline
240 & 10 & 71.362 & 102.627 & 0.103 & 0.157 & 11.131 & 15.124 & 1.508 & 1.354 & 0.152 & 0.236 \tabularnewline
250 & 10 & 32.537 & 54.446 & 0.091 & 0.131 & 5.970 & 9.426 & 0.809 & 0.934 & 0.126 & 0.163  \tabularnewline
\hline 
190 & 12 & 83.201 & 99.816 & 0.125 & 0.177 & 12.579 & 15.042 & 1.411 & 1.165 & 0.182 & 0.271 \tabularnewline
200 & 12 & 80.920 & 112.728 & 0.117 & 0.170 & 12.257 & 16.176 & 1.474 & 1.338 & 0.097 & 0.155 \tabularnewline
210 & 12 & 83.273 & 108.217 & 0.115 & 0.158 & 12.515 & 15.597 & 1.543 & 1.401 & 0.212 & 0.293 \tabularnewline
220 & 12 & 62.768 & 82.227 & 0.106 & 0.153 & 10.108 & 12.927 & 1.304 & 1.150 & 0.171 & 0.248 \tabularnewline
230 & 12 & 71.306 & 74.294 & 0.108 & 0.153 & 11.154 & 12.081 & 1.445 & 1.056 & 0.183 & 0.271 \tabularnewline
240 & 12 & 86.698 & 112.837 & 0.108 & 0.161 & 12.868 & 16.154 & 1.704 & 1.425 & 0.161 & 0.212 \tabularnewline
250 & 12 & 46.957 & 78.187 & 0.086 & 0.125 & 7.962 & 12.181 & 1.232 & 1.349 & 0.095 & 0.128  \tabularnewline
\hline 
190 & 14 & 28.173 & 47.050 & 0.144 & 0.213 & 5.879 & 9.761 & 0.433 & 0.493 & 0.205 & 0.323 \tabularnewline
200 & 14 & 51.696 & 78.919 & 0.121 & 0.176 & 8.872 & 12.891 & 0.947 & 0.961 & 0.199 & 0.254 \tabularnewline
210 & 14 & 52.760 & 57.415 & 0.126 & 0.188 & 9.088 & 10.642 & 0.929 & 0.674 & 0.162 & 0.203 \tabularnewline
220 & 14 & 127.425 & 153.183 & 0.124 & 0.184 & 16.867 & 19.897 & 2.044 & 1.583 & 0.188 & 0.291 \tabularnewline
230 & 14 & 152.067 & 239.617 & 0.107 & 0.150 & 18.739 & 25.149 & 2.709 & 2.702 & 0.132 & 0.202 \tabularnewline
240 & 14 & 117.853 & 188.757 & 0.104 & 0.145 & 15.863 & 21.920 & 2.288 & 2.358 & 0.111 & 0.145 \tabularnewline
250 & 14 & 64.242 & 96.591 & 0.101 & 0.143 & 10.298 & 14.353 & 1.402 & 1.418 & 0.123 & 0.184  \tabularnewline
\hline 
190 & 16 & 60.928 & 94.575 & 0.132 & 0.188 & 10.138 & 14.728 & 1.016 & 1.047 & 0.195 & 0.279 \tabularnewline
200 & 16 & 28.547 & 30.225 & 0.140 & 0.193 & 5.886 & 7.009 & 0.453 & 0.350 & 0.263 & 0.309 \tabularnewline
210 & 16 & 30.035 & 43.531 & 0.137 & 0.191 & 6.080 & 8.888 & 0.490 & 0.512 & 0.249 & 0.336 \tabularnewline
220 & 16 & 111.749 & 124.274 & 0.127 & 0.186 & 15.504 & 17.544 & 1.793 & 1.333 & 0.120 & 0.204 \tabularnewline
230 & 16 & 46.677 & 62.443 & 0.121 & 0.170 & 8.256 & 10.966 & 0.868 & 0.813 & 0.150 & 0.196 \tabularnewline
240 & 16 & 138.608 & 211.928 & 0.109 & 0.153 & 17.705 & 23.555 & 2.488 & 2.434 & 0.120 & 0.175 \tabularnewline
250 & 16 & 98.097 & 151.178 & 0.098 & 0.137 & 13.959 & 19.097 & 2.102 & 2.129 & 0.118 & 0.132  \tabularnewline
\hline 
\end{tabular}
\par\end{centering}
\centering{}\caption{Annualised statistics of sub-interval [2010-12-20, 2016-02-25] of myopic wealth and optimal wealth for unconstrained portfolio with $\gamma=-70$. Interest rate is $r=1\%$. Factor number is 6. Among all the different testing-training window combinations, the statistics for the S\&P 500 ETF, which has the largest profit are profit of $75.089\%$, volatility of $0.156$, expected return of $12.066\%$, Sharpe ratio of $1.028$, and maximum drawdown of $0.186$.}\label{tab:statistics_u_3}
%\end{table}
\end{sidewaystable}

\begin{sidewaystable}[ht]
%\begin{table}[ht]
%\small
%\footnotesize
%\scriptsize
%\tiny
\begin{centering}
\begin{tabular}{cc|cc|cc|cc|cc|cc}
\hline 
\multicolumn{2}{c|}{\diagbox[width=2.8cm]{Windows}{Statistics}} & \multicolumn{2}{c|}{Profit $\left(\%\right)$} & \multicolumn{2}{c|}{Volatility} & \multicolumn{2}{c|}{Expected Return $\left(\%\right)$} & \multicolumn{2}{c|}{Sharpe Ratio} & \multicolumn{2}{c}{Maximum Drawdown}\tabularnewline
\multicolumn{2}{c|}{} & \multicolumn{2}{c|}{$100\ensuremath{\left(\frac{W_{T}-W_{0}}{W_{0}}\right)}$} & \multicolumn{2}{c|}{$\ensuremath{\frac{1}{\sqrt{\Delta t}}\sigma\left(\frac{\Delta W_{t}}{W_{t}}\right)}$} & \multicolumn{2}{c|}{$\ensuremath{\frac{100}{\Delta t}\mathbb{E}\left(\frac{\Delta W_{t}}{W_{t}}\right)}$} & \multicolumn{2}{c|}{$\frac{0.01\left(\mathrm{Return}-r\right)}{\mathrm{Volatility}}$} & \multicolumn{2}{c}{$\underset{t}{\max}\frac{\max_{s\leq t}W_{s}-W_{t}}{\max_{s\leq t}W_{s}}$}\tabularnewline
\hline 
Train & Test & Myopic & Optimal & Myopic & Optimal & Myopic & Optimal & Myopic & Optimal & Myopic & Optimal\tabularnewline
\hline 
190 & 10 & 212.181 & 426.093 & 0.225 & 0.333 & 24.517 & 37.519 & 1.634 & 1.756 & 0.233 & 0.317 \tabularnewline
200 & 10 & 339.630 & 901.051 & 0.242 & 0.351 & 31.567 & 50.656 & 2.102 & 2.580 & 0.203 & 0.280 \tabularnewline
210 & 10 & 233.906 & 331.616 & 0.263 & 0.376 & 26.641 & 35.009 & 1.506 & 1.332 & 0.264 & 0.375 \tabularnewline
220 & 10 & 280.673 & 748.468 & 0.216 & 0.320 & 28.341 & 46.668 & 2.090 & 2.575 & 0.135 & 0.202 \tabularnewline
230 & 10 & 221.258 & 390.567 & 0.267 & 0.361 & 26.206 & 37.282 & 1.430 & 1.552 & 0.288 & 0.354 \tabularnewline
240 & 10 & 172.723 & 298.623 & 0.218 & 0.318 & 21.933 & 31.884 & 1.470 & 1.486 & 0.162 & 0.230 \tabularnewline
250 & 10 & 163.417 & 341.385 & 0.193 & 0.269 & 20.793 & 32.623 & 1.596 & 1.921 & 0.141 & 0.259  \tabularnewline
\hline 
190 & 12 & 50.150 & 140.883 & 0.388 & 0.503 & 14.086 & 27.401 & 0.286 & 0.538 & 0.473 & 0.535 \tabularnewline
200 & 12 & 161.066 & 318.941 & 0.277 & 0.391 & 22.017 & 34.466 & 1.085 & 1.248 & 0.327 & 0.470 \tabularnewline
210 & 12 & 18.386 & 41.694 & 0.238 & 0.361 & 6.337 & 13.828 & 0.163 & 0.259 & 0.602 & 0.726 \tabularnewline
220 & 12 & 326.129 & 815.175 & 0.215 & 0.312 & 30.553 & 47.938 & 2.311 & 2.761 & 0.165 & 0.222 \tabularnewline
230 & 12 & 141.034 & 248.719 & 0.354 & 0.473 & 23.007 & 34.789 & 0.772 & 0.877 & 0.541 & 0.650 \tabularnewline
240 & 12 & 57.033 & 96.510 & 0.228 & 0.334 & 11.441 & 18.870 & 0.557 & 0.604 & 0.444 & 0.584 \tabularnewline
250 & 12 & 35.746 & 33.755 & 0.272 & 0.379 & 9.271 & 11.964 & 0.297 & 0.202 & 0.467 & 0.617 \tabularnewline
\hline 
190 & 14 & 108.210 & 206.520 & 0.411 & 0.541 & 21.159 & 33.870 & 0.533 & 0.666 & 0.514 & 0.581 \tabularnewline
200 & 14 & 410.313 & 1038.297 & 0.291 & 0.426 & 35.483 & 55.307 & 1.970 & 2.294 & 0.273 & 0.352 \tabularnewline
210 & 14 & 3.799 & 20.057 & 0.250 & 0.365 & 4.093 & 10.793 & 0.002 & 0.118 & 0.641 & 0.762 \tabularnewline
220 & 14 & 211.682 & 368.207 & 0.320 & 0.435 & 26.752 & 38.448 & 1.155 & 1.238 & 0.432 & 0.552 \tabularnewline
230 & 14 & 128.878 & 203.056 & 0.224 & 0.318 & 18.620 & 26.568 & 1.139 & 1.130 & 0.225 & 0.306 \tabularnewline
240 & 14 & 147.784 & 317.318 & 0.281 & 0.382 & 21.181 & 34.181 & 1.013 & 1.287 & 0.338 & 0.392 \tabularnewline
250 & 14 & 157.352 & 265.662 & 0.206 & 0.286 & 20.604 & 29.470 & 1.454 & 1.526 & 0.266 & 0.386  \tabularnewline
\hline 
190 & 16 & 108.200 & 193.025 & 0.378 & 0.546 & 20.262 & 32.932 & 0.581 & 0.629 & 0.476 & 0.580 \tabularnewline
200 & 16 & 463.357 & 1106.716 & 0.277 & 0.386 & 37.297 & 55.606 & 2.237 & 2.617 & 0.189 & 0.233 \tabularnewline
210 & 16 & 181.377 & 354.925 & 0.496 & 0.623 & 29.935 & 44.782 & 0.665 & 0.842 & 0.553 & 0.634 \tabularnewline
220 & 16 & 91.467 & 42.404 & 0.299 & 0.475 & 16.982 & 18.267 & 0.643 & 0.201 & 0.538 & 0.709 \tabularnewline
230 & 16 & 10.268 & 25.775 & 0.366 & 0.485 & 8.172 & 15.457 & 0.049 & 0.118 & 0.600 & 0.688 \tabularnewline
240 & 16 & 140.840 & 267.735 & 0.296 & 0.420 & 21.255 & 33.522 & 0.927 & 1.044 & 0.420 & 0.525 \tabularnewline
250 & 16 & 29.985 & 66.953 & 0.215 & 0.326 & 7.602 & 15.702 & 0.314 & 0.451 & 0.501 & 0.622  \tabularnewline
\hline 
\end{tabular}
\par\end{centering}
\centering{}\caption{Annualised statistics of sub-interval [2016-02-26, 2021-05-05] of myopic wealth and optimal wealth for unconstrained portfolio with $\gamma=-70$. Interest rate is $r=1\%$. Factor number is 6. Among all the different testing-training window combinations, the statistics for the S\&P 500 ETF, which has the largest profit are profit of $135.484\%$, volatility of $0.185$, expected return of $18.257\%$, Sharpe ratio of $1.405$, and maximum drawdown of $0.337$.}\label{tab:statistics_u_4}
%\end{table}
\end{sidewaystable}

\begin{sidewaystable}[ht]
%\begin{table}[ht]
%\small
%\footnotesize
%\scriptsize
%\tiny
\begin{centering}
\begin{tabular}{cc|cc|cc|cc|cc|cc}
\hline 
\multicolumn{2}{c|}{\diagbox[width=2.8cm]{Windows}{Statistics}} & \multicolumn{2}{c|}{Profit $\left(\%\right)$} & \multicolumn{2}{c|}{Volatility} & \multicolumn{2}{c|}{Expected Return $\left(\%\right)$} & \multicolumn{2}{c|}{Sharpe Ratio} & \multicolumn{2}{c}{Maximum Drawdown}\tabularnewline
\multicolumn{2}{c|}{} & \multicolumn{2}{c|}{$100\ensuremath{\left(\frac{W_{T}-W_{0}}{W_{0}}\right)}$} & \multicolumn{2}{c|}{$\ensuremath{\frac{1}{\sqrt{\Delta t}}\sigma\left(\frac{\Delta W_{t}}{W_{t}}\right)}$} & \multicolumn{2}{c|}{$\ensuremath{\frac{100}{\Delta t}\mathbb{E}\left(\frac{\Delta W_{t}}{W_{t}}\right)}$} & \multicolumn{2}{c|}{$\frac{0.01\left(\mathrm{Return}-r\right)}{\mathrm{Volatility}}$} & \multicolumn{2}{c}{$\underset{t}{\max}\frac{\max_{s\leq t}W_{s}-W_{t}}{\max_{s\leq t}W_{s}}$}\tabularnewline
\hline 
Train & Test & Myopic & Optimal & Myopic & Optimal & Myopic & Optimal & Myopic & Optimal & Myopic & Optimal\tabularnewline
\hline 
190 & 10 & 2746.855 & 5819.132 & 0.146 & 0.323 & 17.336 & 25.042 & 1.760 & 0.999 & 0.290 & 0.685 \tabularnewline
200 & 10 & 4065.223 & 16689.747 & 0.152 & 0.340 & 19.316 & 30.657 & 1.913 & 1.251 & 0.170 & 0.481 \tabularnewline
210 & 10 & 2732.678 & 27413.431 & 0.141 & 0.286 & 17.306 & 31.468 & 1.826 & 1.671 & 0.187 & 0.401 \tabularnewline
220 & 10 & 2331.543 & 13117.227 & 0.137 & 0.291 & 16.550 & 28.117 & 1.776 & 1.389 & 0.175 & 0.384 \tabularnewline
230 & 10 & 1620.735 & 8605.462 & 0.142 & 0.296 & 14.951 & 26.280 & 1.502 & 1.229 & 0.224 & 0.495 \tabularnewline
240 & 10 & 1408.360 & 2823.655 & 0.120 & 0.263 & 14.051 & 19.985 & 1.688 & 0.994 & 0.196 & 0.507 \tabularnewline
250 & 10 & 1308.429 & 4409.846 & 0.112 & 0.238 & 13.646 & 21.571 & 1.763 & 1.267 & 0.174 & 0.337 \tabularnewline 
\hline 
190 & 12 & 1597.937 & 8332.285 & 0.184 & 0.340 & 15.366 & 27.089 & 1.144 & 1.051 & 0.415 & 0.575 \tabularnewline
200 & 12 & 2958.950 & 5846.749 & 0.152 & 0.327 & 17.809 & 25.264 & 1.736 & 0.992 & 0.249 & 0.487 \tabularnewline
210 & 12 & 1680.458 & 3204.479 & 0.140 & 0.304 & 15.042 & 21.687 & 1.540 & 0.891 & 0.254 & 0.720 \tabularnewline
220 & 12 & 2595.199 & 19867.240 & 0.135 & 0.262 & 17.025 & 29.302 & 1.870 & 1.702 & 0.154 & 0.283 \tabularnewline
230 & 12 & 1522.264 & 3647.359 & 0.163 & 0.316 & 14.957 & 22.681 & 1.282 & 0.897 & 0.410 & 0.514 \tabularnewline
240 & 12 & 1962.002 & 4665.097 & 0.136 & 0.262 & 15.776 & 22.386 & 1.697 & 1.169 & 0.295 & 0.650 \tabularnewline
250 & 12 & 803.322 & 3501.733 & 0.119 & 0.240 & 11.531 & 20.497 & 1.348 & 1.174 & 0.194 & 0.360 \tabularnewline 
\hline 
190 & 14 & 1005.513 & 3101.713 & 0.203 & 0.380 & 13.619 & 23.735 & 0.861 & 0.701 & 0.353 & 0.592 \tabularnewline
200 & 14 & 5436.736 & 63513.307 & 0.169 & 0.353 & 20.963 & 37.623 & 1.883 & 1.606 & 0.268 & 0.546 \tabularnewline
210 & 14 & 1030.400 & 657.399 & 0.152 & 0.317 & 12.989 & 14.868 & 1.164 & 0.454 & 0.350 & 0.672 \tabularnewline
220 & 14 & 2317.590 & 13602.197 & 0.157 & 0.290 & 16.804 & 28.201 & 1.546 & 1.403 & 0.221 & 0.338 \tabularnewline
230 & 14 & 3814.400 & 48869.753 & 0.140 & 0.285 & 18.968 & 34.377 & 2.049 & 1.904 & 0.245 & 0.391 \tabularnewline
240 & 14 & 3262.232 & 18622.598 & 0.144 & 0.273 & 18.282 & 29.349 & 1.901 & 1.614 & 0.169 & 0.525 \tabularnewline
250 & 14 & 1214.144 & 5526.874 & 0.125 & 0.268 & 13.463 & 23.417 & 1.531 & 1.204 & 0.195 & 0.329 \tabularnewline 
\hline 
190 & 16 & 2494.671 & 10308.829 & 0.220 & 0.421 & 18.056 & 30.941 & 1.128 & 0.896 & 0.337 & 0.729 \tabularnewline
200 & 16 & 2288.028 & 9070.563 & 0.177 & 0.348 & 17.016 & 28.031 & 1.361 & 1.052 & 0.208 & 0.495 \tabularnewline
210 & 16 & 5381.788 & 32628.602 & 0.210 & 0.373 & 21.643 & 35.060 & 1.509 & 1.329 & 0.387 & 0.697 \tabularnewline
220 & 16 & 1837.231 & 11473.759 & 0.165 & 0.320 & 15.855 & 28.267 & 1.355 & 1.222 & 0.304 & 0.528 \tabularnewline
230 & 16 & 779.969 & 1180.806 & 0.167 & 0.354 & 12.029 & 18.761 & 0.943 & 0.533 & 0.439 & 0.650 \tabularnewline
240 & 16 & 2598.835 & 10224.917 & 0.155 & 0.301 & 17.363 & 27.222 & 1.641 & 1.265 & 0.290 & 0.502 \tabularnewline
250 & 16 & 1843.139 & 10834.226 & 0.148 & 0.287 & 15.725 & 27.300 & 1.526 & 1.348 & 0.373 & 0.527 \tabularnewline 
\hline 
\end{tabular}
\par\end{centering}
\centering{}\caption{Annualised statistics of time interval [2000-01-03, 2021-05-06] of myopic wealth and optimal wealth for constrained portfolio to be market neutral with $\gamma=-70$. Interest rate is $r=1\%$. Factor number is 6. Among all the different testing-training window combinations, the statistics for the S\&P 500 ETF, which has the largest profit are profit of $359.795\%$, volatility of $0.196$, expected return of $9.414\%$, Sharpe ratio of $0.532$, and maximum drawdown of $0.552$.}\label{tab:statistics_c}
%\end{table}
\end{sidewaystable}

\begin{sidewaystable}[ht]
%\begin{table}[ht]
%\small
%\footnotesize
%\scriptsize
%\tiny
\begin{centering}
\begin{tabular}{cc|cc|cc|cc|cc|cc}
\hline 
\multicolumn{2}{c|}{\diagbox[width=2.8cm]{Windows}{Statistics}} & \multicolumn{2}{c|}{Profit $\left(\%\right)$} & \multicolumn{2}{c|}{Volatility} & \multicolumn{2}{c|}{Expected Return $\left(\%\right)$} & \multicolumn{2}{c|}{Sharpe Ratio} & \multicolumn{2}{c}{Maximum Drawdown}\tabularnewline
\multicolumn{2}{c|}{} & \multicolumn{2}{c|}{$100\ensuremath{\left(\frac{W_{T}-W_{0}}{W_{0}}\right)}$} & \multicolumn{2}{c|}{$\ensuremath{\frac{1}{\sqrt{\Delta t}}\sigma\left(\frac{\Delta W_{t}}{W_{t}}\right)}$} & \multicolumn{2}{c|}{$\ensuremath{\frac{100}{\Delta t}\mathbb{E}\left(\frac{\Delta W_{t}}{W_{t}}\right)}$} & \multicolumn{2}{c|}{$\frac{0.01\left(\mathrm{Return}-r\right)}{\mathrm{Volatility}}$} & \multicolumn{2}{c}{$\underset{t}{\max}\frac{\max_{s\leq t}W_{s}-W_{t}}{\max_{s\leq t}W_{s}}$}\tabularnewline
\hline 
Train & Test & Myopic & Optimal & Myopic & Optimal & Myopic & Optimal & Myopic & Optimal & Myopic & Optimal\tabularnewline
\hline 
190 & 10 & 147.205 & 92.481 & 0.152 & 0.472 & 18.799 & 24.060 & 1.853 & 0.409 & 0.261 & 0.685 \tabularnewline
200 & 10 & 236.604 & 437.831 & 0.171 & 0.514 & 25.146 & 46.024 & 2.338 & 1.164 & 0.170 & 0.465 \tabularnewline
210 & 10 & 170.163 & 657.454 & 0.137 & 0.386 & 20.384 & 47.097 & 2.303 & 1.982 & 0.161 & 0.401 \tabularnewline
220 & 10 & 183.112 & 419.189 & 0.141 & 0.415 & 21.388 & 41.089 & 2.366 & 1.410 & 0.103 & 0.384 \tabularnewline
230 & 10 & 145.382 & 350.275 & 0.141 & 0.442 & 18.635 & 39.548 & 1.986 & 1.184 & 0.140 & 0.495 \tabularnewline
240 & 10 & 97.543 & 79.146 & 0.129 & 0.383 & 14.221 & 18.740 & 1.585 & 0.445 & 0.155 & 0.507 \tabularnewline
250 & 10 & 130.837 & 165.659 & 0.123 & 0.346 & 17.252 & 25.254 & 2.110 & 0.900 & 0.082 & 0.256 \tabularnewline 
\hline 
190 & 12 & 135.862 & 416.451 & 0.178 & 0.467 & 18.319 & 42.871 & 1.481 & 1.241 & 0.277 & 0.575 \tabularnewline
200 & 12 & 282.812 & 319.205 & 0.167 & 0.512 & 27.606 & 41.145 & 2.702 & 0.956 & 0.171 & 0.487 \tabularnewline
210 & 12 & 92.010 & 3.476 & 0.161 & 0.457 & 14.061 & 11.261 & 1.196 & -0.001 & 0.177 & 0.720 \tabularnewline
220 & 12 & 202.504 & 679.238 & 0.148 & 0.359 & 22.786 & 46.688 & 2.420 & 2.172 & 0.131 & 0.283 \tabularnewline
230 & 12 & 151.525 & 288.733 & 0.202 & 0.474 & 20.097 & 37.816 & 1.438 & 0.972 & 0.174 & 0.392 \tabularnewline
240 & 12 & 117.597 & 84.197 & 0.129 & 0.363 & 16.132 & 18.643 & 1.837 & 0.496 & 0.214 & 0.650 \tabularnewline
250 & 12 & 141.794 & 413.747 & 0.130 & 0.345 & 18.254 & 38.258 & 2.132 & 1.697 & 0.141 & 0.360  \tabularnewline
\hline 
190 & 14 & 90.001 & 212.098 & 0.194 & 0.497 & 14.410 & 34.790 & 0.973 & 0.741 & 0.274 & 0.504 \tabularnewline
200 & 14 & 350.254 & 2035.505 & 0.201 & 0.517 & 31.393 & 73.036 & 2.589 & 2.640 & 0.268 & 0.356 \tabularnewline
210 & 14 & 120.863 & 2.746 & 0.175 & 0.470 & 17.005 & 11.534 & 1.379 & -0.005 & 0.255 & 0.672 \tabularnewline
220 & 14 & 249.528 & 979.734 & 0.157 & 0.394 & 25.747 & 54.276 & 2.655 & 2.420 & 0.129 & 0.302 \tabularnewline
230 & 14 & 374.764 & 2090.374 & 0.159 & 0.415 & 31.863 & 69.140 & 3.438 & 3.363 & 0.088 & 0.250 \tabularnewline
240 & 14 & 183.131 & 236.359 & 0.167 & 0.394 & 21.862 & 31.605 & 2.001 & 1.021 & 0.169 & 0.525 \tabularnewline
250 & 14 & 195.180 & 652.399 & 0.149 & 0.409 & 22.454 & 48.109 & 2.367 & 1.880 & 0.195 & 0.323  \tabularnewline
\hline 
190 & 16 & 107.599 & 457.913 & 0.215 & 0.542 & 16.551 & 48.192 & 1.018 & 1.134 & 0.337 & 0.518 \tabularnewline
200 & 16 & 190.655 & 638.836 & 0.187 & 0.488 & 22.584 & 51.002 & 1.828 & 1.535 & 0.208 & 0.494 \tabularnewline
210 & 16 & 305.670 & 716.582 & 0.188 & 0.503 & 29.188 & 54.052 & 2.531 & 1.595 & 0.321 & 0.697 \tabularnewline
220 & 16 & 249.729 & 648.011 & 0.189 & 0.460 & 26.289 & 49.805 & 2.206 & 1.651 & 0.295 & 0.528 \tabularnewline
230 & 16 & 157.616 & 90.372 & 0.181 & 0.547 & 20.229 & 27.898 & 1.649 & 0.349 & 0.211 & 0.650 \tabularnewline
240 & 16 & 80.214 & 68.784 & 0.168 & 0.422 & 12.976 & 19.168 & 1.030 & 0.357 & 0.290 & 0.502 \tabularnewline
250 & 16 & 260.325 & 703.142 & 0.166 & 0.416 & 26.666 & 49.793 & 2.593 & 1.927 & 0.203 & 0.527  \tabularnewline
\hline 
\end{tabular}
\par\end{centering}
\centering{}\caption{Annualised statistics of sub-interval [2000-08-07, 2005-10-13] of myopic wealth and optimal wealth for constrained portfolio to be market neutral with $\gamma=-70$. Interest rate is $r=1\%$. Factor number is 6. Among all the different testing-training window combinations, the statistics for the S\&P 500 ETF, which has the largest profit are profit of $5.078\%$, volatility of $0.186$, expected return of $2.693\%$, Sharpe ratio of $0.023$, and maximum drawdown of $0.421$.}\label{tab:statistics_c_1}
%\end{table}
\end{sidewaystable}

\begin{sidewaystable}[ht]
%\begin{table}[ht]
%\small
%\footnotesize
%\scriptsize
%\tiny
\begin{centering}
\begin{tabular}{cc|cc|cc|cc|cc|cc}
\hline 
\multicolumn{2}{c|}{\diagbox[width=2.8cm]{Windows}{Statistics}} & \multicolumn{2}{c|}{Profit $\left(\%\right)$} & \multicolumn{2}{c|}{Volatility} & \multicolumn{2}{c|}{Expected Return $\left(\%\right)$} & \multicolumn{2}{c|}{Sharpe Ratio} & \multicolumn{2}{c}{Maximum Drawdown}\tabularnewline
\multicolumn{2}{c|}{} & \multicolumn{2}{c|}{$100\ensuremath{\left(\frac{W_{T}-W_{0}}{W_{0}}\right)}$} & \multicolumn{2}{c|}{$\ensuremath{\frac{1}{\sqrt{\Delta t}}\sigma\left(\frac{\Delta W_{t}}{W_{t}}\right)}$} & \multicolumn{2}{c|}{$\ensuremath{\frac{100}{\Delta t}\mathbb{E}\left(\frac{\Delta W_{t}}{W_{t}}\right)}$} & \multicolumn{2}{c|}{$\frac{0.01\left(\mathrm{Return}-r\right)}{\mathrm{Volatility}}$} & \multicolumn{2}{c}{$\underset{t}{\max}\frac{\max_{s\leq t}W_{s}-W_{t}}{\max_{s\leq t}W_{s}}$}\tabularnewline
\hline 
Train & Test & Myopic & Optimal & Myopic & Optimal & Myopic & Optimal & Myopic & Optimal & Myopic & Optimal\tabularnewline
\hline 
190 & 10 & 271.079 & 391.537 & 0.182 & 0.347 & 27.155 & 36.831 & 2.398 & 1.606 & 0.154 & 0.417 \tabularnewline
200 & 10 & 174.447 & 191.263 & 0.177 & 0.336 & 21.214 & 26.255 & 1.810 & 1.021 & 0.150 & 0.481 \tabularnewline
210 & 10 & 233.614 & 579.354 & 0.167 & 0.308 & 24.942 & 42.130 & 2.379 & 2.308 & 0.158 & 0.263 \tabularnewline
220 & 10 & 180.011 & 368.381 & 0.164 & 0.302 & 21.509 & 34.689 & 2.003 & 1.789 & 0.175 & 0.291 \tabularnewline
230 & 10 & 162.804 & 391.854 & 0.141 & 0.260 & 19.974 & 34.642 & 2.173 & 2.168 & 0.214 & 0.280 \tabularnewline
240 & 10 & 193.044 & 307.372 & 0.144 & 0.285 & 22.172 & 31.540 & 2.409 & 1.692 & 0.089 & 0.285 \tabularnewline
250 & 10 & 158.641 & 350.665 & 0.141 & 0.268 & 19.719 & 33.233 & 2.139 & 1.962 & 0.174 & 0.337  \tabularnewline
\hline 
190 & 12 & 172.776 & 249.840 & 0.178 & 0.326 & 21.102 & 29.625 & 1.787 & 1.269 & 0.178 & 0.293 \tabularnewline
200 & 12 & 54.100 & 43.971 & 0.170 & 0.301 & 9.902 & 11.784 & 0.705 & 0.327 & 0.249 & 0.403 \tabularnewline
210 & 12 & 256.209 & 615.959 & 0.168 & 0.321 & 26.237 & 43.442 & 2.515 & 2.296 & 0.172 & 0.259 \tabularnewline
220 & 12 & 137.000 & 279.802 & 0.165 & 0.297 & 18.243 & 30.436 & 1.621 & 1.517 & 0.154 & 0.235 \tabularnewline
230 & 12 & 172.518 & 213.405 & 0.121 & 0.257 & 20.427 & 25.758 & 2.644 & 1.452 & 0.090 & 0.366 \tabularnewline
240 & 12 & 262.011 & 531.115 & 0.146 & 0.266 & 26.358 & 39.670 & 2.953 & 2.555 & 0.118 & 0.340 \tabularnewline
250 & 12 & 98.504 & 215.748 & 0.127 & 0.242 & 14.326 & 25.582 & 1.623 & 1.565 & 0.141 & 0.230  \tabularnewline
\hline 
190 & 14 & 114.012 & 125.642 & 0.280 & 0.463 & 18.300 & 25.151 & 0.818 & 0.535 & 0.328 & 0.592 \tabularnewline
200 & 14 & 156.631 & 155.613 & 0.188 & 0.381 & 20.123 & 25.552 & 1.573 & 0.771 & 0.214 & 0.546 \tabularnewline
210 & 14 & 190.472 & 243.485 & 0.180 & 0.338 & 22.452 & 29.643 & 1.899 & 1.208 & 0.203 & 0.303 \tabularnewline
220 & 14 & 134.003 & 202.423 & 0.189 & 0.320 & 18.388 & 26.629 & 1.388 & 1.120 & 0.164 & 0.338 \tabularnewline
230 & 14 & 124.876 & 269.697 & 0.167 & 0.315 & 17.305 & 30.525 & 1.494 & 1.400 & 0.245 & 0.391 \tabularnewline
240 & 14 & 271.538 & 784.733 & 0.163 & 0.292 & 27.114 & 46.956 & 2.715 & 2.916 & 0.154 & 0.222 \tabularnewline
250 & 14 & 70.851 & 95.526 & 0.143 & 0.269 & 11.567 & 16.824 & 1.088 & 0.746 & 0.149 & 0.288  \tabularnewline
\hline 
190 & 16 & 353.246 & 327.648 & 0.264 & 0.490 & 32.633 & 39.707 & 1.978 & 1.013 & 0.201 & 0.729 \tabularnewline
200 & 16 & 143.586 & 124.098 & 0.200 & 0.367 & 19.345 & 22.452 & 1.378 & 0.670 & 0.180 & 0.495 \tabularnewline
210 & 16 & 311.656 & 609.489 & 0.187 & 0.343 & 29.403 & 43.993 & 2.578 & 2.135 & 0.140 & 0.288 \tabularnewline
220 & 16 & 101.771 & 237.584 & 0.187 & 0.323 & 15.503 & 29.019 & 1.124 & 1.243 & 0.228 & 0.384 \tabularnewline
230 & 16 & 136.717 & 270.919 & 0.169 & 0.320 & 18.334 & 30.741 & 1.585 & 1.380 & 0.283 & 0.379 \tabularnewline
240 & 16 & 376.044 & 1035.434 & 0.162 & 0.301 & 31.979 & 52.187 & 3.397 & 3.277 & 0.166 & 0.240 \tabularnewline
250 & 16 & 156.792 & 344.392 & 0.160 & 0.291 & 19.871 & 33.590 & 1.866 & 1.788 & 0.298 & 0.428  \tabularnewline
\hline 
\end{tabular}
\par\end{centering}
\centering{}\caption{Annualised statistics of sub-interval [2005-10-14, 2010-12-17] of myopic wealth and optimal wealth for constrained portfolio to be market neutral with $\gamma=-70$. Interest rate is $r=1\%$. Factor number is 6. Among all the different testing-training window combinations, the statistics for the S\&P 500 ETF, which has the largest profit are profit of $17.725\%$, volatility of $0.247$, expected return of $6.232\%$, Sharpe ratio of $0.152$, and maximum drawdown of $0.552$.}\label{tab:statistics_c_2}
%\end{table}
\end{sidewaystable}

\begin{sidewaystable}[ht]
%\begin{table}[ht]
%\small
%\footnotesize
%\scriptsize
%\tiny
\begin{centering}
\begin{tabular}{cc|cc|cc|cc|cc|cc}
\hline 
\multicolumn{2}{c|}{\diagbox[width=2.8cm]{Windows}{Statistics}} & \multicolumn{2}{c|}{Profit $\left(\%\right)$} & \multicolumn{2}{c|}{Volatility} & \multicolumn{2}{c|}{Expected Return $\left(\%\right)$} & \multicolumn{2}{c|}{Sharpe Ratio} & \multicolumn{2}{c}{Maximum Drawdown}\tabularnewline
\multicolumn{2}{c|}{} & \multicolumn{2}{c|}{$100\ensuremath{\left(\frac{W_{T}-W_{0}}{W_{0}}\right)}$} & \multicolumn{2}{c|}{$\ensuremath{\frac{1}{\sqrt{\Delta t}}\sigma\left(\frac{\Delta W_{t}}{W_{t}}\right)}$} & \multicolumn{2}{c|}{$\ensuremath{\frac{100}{\Delta t}\mathbb{E}\left(\frac{\Delta W_{t}}{W_{t}}\right)}$} & \multicolumn{2}{c|}{$\frac{0.01\left(\mathrm{Return}-r\right)}{\mathrm{Volatility}}$} & \multicolumn{2}{c}{$\underset{t}{\max}\frac{\max_{s\leq t}W_{s}-W_{t}}{\max_{s\leq t}W_{s}}$}\tabularnewline
\hline 
Train & Test & Myopic & Optimal & Myopic & Optimal & Myopic & Optimal & Myopic & Optimal & Myopic & Optimal\tabularnewline
\hline 
190 & 10 & 6.066 & 13.841 & 0.085 & 0.120 & 1.510 & 3.253 & 0.080 & 0.227 & 0.253 & 0.337 \tabularnewline
200 & 10 & 24.633 & 40.880 & 0.083 & 0.117 & 4.644 & 7.377 & 0.652 & 0.782 & 0.143 & 0.180 \tabularnewline
210 & 10 & 46.065 & 65.069 & 0.088 & 0.127 & 7.805 & 10.609 & 1.167 & 1.126 & 0.128 & 0.174 \tabularnewline
220 & 10 & 34.125 & 44.128 & 0.078 & 0.109 & 6.055 & 7.764 & 0.991 & 0.907 & 0.097 & 0.124 \tabularnewline
230 & 10 & 30.275 & 38.710 & 0.080 & 0.113 & 5.516 & 7.071 & 0.850 & 0.771 & 0.153 & 0.189 \tabularnewline
240 & 10 & 39.881 & 54.280 & 0.077 & 0.108 & 6.901 & 9.120 & 1.172 & 1.119 & 0.196 & 0.265 \tabularnewline
250 & 10 & 16.668 & 20.130 & 0.064 & 0.091 & 3.247 & 4.029 & 0.543 & 0.481 & 0.101 & 0.140  \tabularnewline
\hline 
190 & 12 & 54.444 & 70.014 & 0.093 & 0.132 & 8.901 & 11.212 & 1.299 & 1.146 & 0.182 & 0.260 \tabularnewline
200 & 12 & 65.483 & 83.743 & 0.083 & 0.114 & 10.175 & 12.529 & 1.729 & 1.555 & 0.074 & 0.102 \tabularnewline
210 & 12 & 29.749 & 38.434 & 0.087 & 0.120 & 5.474 & 7.083 & 0.766 & 0.722 & 0.213 & 0.282 \tabularnewline
220 & 12 & 55.893 & 64.379 & 0.075 & 0.108 & 8.985 & 10.326 & 1.649 & 1.306 & 0.101 & 0.152 \tabularnewline
230 & 12 & 36.468 & 32.553 & 0.082 & 0.113 & 6.443 & 6.181 & 1.007 & 0.648 & 0.140 & 0.195 \tabularnewline
240 & 12 & 50.922 & 67.277 & 0.081 & 0.112 & 8.428 & 10.749 & 1.405 & 1.318 & 0.110 & 0.144 \tabularnewline
250 & 12 & 19.989 & 27.933 & 0.062 & 0.087 & 3.783 & 5.237 & 0.705 & 0.721 & 0.127 & 0.171  \tabularnewline
\hline 
190 & 14 & 11.391 & 21.686 & 0.111 & 0.152 & 2.722 & 4.983 & 0.188 & 0.308 & 0.202 & 0.277 \tabularnewline
200 & 14 & 27.382 & 39.251 & 0.090 & 0.122 & 5.134 & 7.208 & 0.672 & 0.723 & 0.234 & 0.283 \tabularnewline
210 & 14 & 22.972 & 19.323 & 0.100 & 0.138 & 4.543 & 4.408 & 0.504 & 0.299 & 0.161 & 0.221 \tabularnewline
220 & 14 & 60.816 & 63.612 & 0.090 & 0.132 & 9.718 & 10.525 & 1.487 & 1.055 & 0.091 & 0.145 \tabularnewline
230 & 14 & 90.702 & 132.600 & 0.079 & 0.107 & 12.998 & 17.159 & 2.411 & 2.438 & 0.136 & 0.169 \tabularnewline
240 & 14 & 51.498 & 67.160 & 0.073 & 0.100 & 8.438 & 10.603 & 1.584 & 1.482 & 0.092 & 0.132 \tabularnewline
250 & 14 & 52.219 & 70.980 & 0.071 & 0.098 & 8.536 & 11.061 & 1.658 & 1.579 & 0.095 & 0.138 
 \tabularnewline
\hline
190 & 16 & 17.520 & 24.548 & 0.098 & 0.135 & 3.623 & 5.185 & 0.374 & 0.399 & 0.200 & 0.278 \tabularnewline
200 & 16 & 25.900 & 24.594 & 0.108 & 0.147 & 5.075 & 5.373 & 0.532 & 0.368 & 0.188 & 0.236 \tabularnewline
210 & 16 & 28.896 & 37.806 & 0.100 & 0.137 & 5.465 & 7.218 & 0.647 & 0.619 & 0.196 & 0.251 \tabularnewline
220 & 16 & 54.134 & 68.870 & 0.096 & 0.131 & 8.936 & 11.124 & 1.259 & 1.147 & 0.089 & 0.139 \tabularnewline
230 & 16 & 21.650 & 30.518 & 0.086 & 0.119 & 4.222 & 5.935 & 0.548 & 0.578 & 0.119 & 0.142 \tabularnewline
240 & 16 & 67.361 & 98.955 & 0.082 & 0.113 & 10.471 & 14.172 & 1.797 & 1.826 & 0.092 & 0.124 \tabularnewline
250 & 16 & 36.773 & 48.408 & 0.071 & 0.098 & 6.430 & 8.269 & 1.171 & 1.112 & 0.075 & 0.110  \tabularnewline
\hline 
\end{tabular}
\par\end{centering}
\centering{}\caption{Annualised statistics of sub-interval [2010-12-20, 2016-02-25] of myopic wealth and optimal wealth for constrained portfolio to be market neutral with $\gamma=-70$. Interest rate is $r=1\%$. Factor number is 6. Among all the different testing-training window combinations, the statistics for the S\&P 500 ETF, which has the largest profit are profit of $75.089\%$, volatility of $0.156$, expected return of $12.066\%$, Sharpe ratio of $1.028$, and maximum drawdown of $0.186$.}\label{tab:statistics_c_3}
%\end{table}
\end{sidewaystable}

\begin{sidewaystable}[ht]
%\begin{table}[ht]
%\small
%\footnotesize
%\scriptsize
%\tiny
\begin{centering}
\begin{tabular}{cc|cc|cc|cc|cc|cc}
\hline 
\multicolumn{2}{c|}{\diagbox[width=2.8cm]{Windows}{Statistics}} & \multicolumn{2}{c|}{Profit $\left(\%\right)$} & \multicolumn{2}{c|}{Volatility} & \multicolumn{2}{c|}{Expected Return $\left(\%\right)$} & \multicolumn{2}{c|}{Sharpe Ratio} & \multicolumn{2}{c}{Maximum Drawdown}\tabularnewline
\multicolumn{2}{c|}{} & \multicolumn{2}{c|}{$100\ensuremath{\left(\frac{W_{T}-W_{0}}{W_{0}}\right)}$} & \multicolumn{2}{c|}{$\ensuremath{\frac{1}{\sqrt{\Delta t}}\sigma\left(\frac{\Delta W_{t}}{W_{t}}\right)}$} & \multicolumn{2}{c|}{$\ensuremath{\frac{100}{\Delta t}\mathbb{E}\left(\frac{\Delta W_{t}}{W_{t}}\right)}$} & \multicolumn{2}{c|}{$\frac{0.01\left(\mathrm{Return}-r\right)}{\mathrm{Volatility}}$} & \multicolumn{2}{c}{$\underset{t}{\max}\frac{\max_{s\leq t}W_{s}-W_{t}}{\max_{s\leq t}W_{s}}$}\tabularnewline
\hline 
Train & Test & Myopic & Optimal & Myopic & Optimal & Myopic & Optimal & Myopic & Optimal & Myopic & Optimal\tabularnewline
\hline 
190 & 10 & 189.723 & 441.757 & 0.146 & 0.248 & 21.716 & 35.779 & 2.324 & 2.418 & 0.122 & 0.226 \tabularnewline
200 & 10 & 263.226 & 664.242 & 0.160 & 0.270 & 26.379 & 43.121 & 2.690 & 2.832 & 0.091 & 0.205 \tabularnewline
210 & 10 & 112.127 & 217.387 & 0.158 & 0.260 & 15.857 & 25.720 & 1.439 & 1.444 & 0.180 & 0.293 \tabularnewline
220 & 10 & 128.399 & 276.686 & 0.150 & 0.250 & 17.261 & 28.968 & 1.689 & 1.783 & 0.153 & 0.285 \tabularnewline
230 & 10 & 103.289 & 179.688 & 0.187 & 0.272 & 15.565 & 23.665 & 1.139 & 1.209 & 0.224 & 0.274 \tabularnewline
240 & 10 & 85.314 & 158.637 & 0.121 & 0.193 & 12.841 & 20.505 & 1.509 & 1.557 & 0.142 & 0.192 \tabularnewline
250 & 10 & 102.094 & 213.065 & 0.105 & 0.165 & 14.387 & 23.783 & 2.010 & 2.268 & 0.093 & 0.135  \tabularnewline
\hline 
190 & 12 & 68.565 & 169.894 & 0.253 & 0.346 & 12.917 & 24.384 & 0.586 & 0.902 & 0.415 & 0.490 \tabularnewline
200 & 12 & 218.773 & 448.018 & 0.168 & 0.250 & 23.922 & 36.071 & 2.242 & 2.426 & 0.119 & 0.163 \tabularnewline
210 & 12 & 100.140 & 221.550 & 0.130 & 0.209 & 14.386 & 24.965 & 1.591 & 1.824 & 0.193 & 0.284 \tabularnewline
220 & 12 & 139.880 & 307.505 & 0.136 & 0.213 & 18.022 & 29.684 & 1.998 & 2.245 & 0.118 & 0.164 \tabularnewline
230 & 12 & 73.414 & 131.901 & 0.210 & 0.309 & 12.899 & 21.012 & 0.756 & 0.839 & 0.410 & 0.514 \tabularnewline
240 & 12 & 70.614 & 140.306 & 0.170 & 0.245 & 11.902 & 20.149 & 0.907 & 1.114 & 0.295 & 0.359 \tabularnewline
250 & 12 & 56.659 & 73.136 & 0.140 & 0.211 & 9.769 & 12.929 & 0.903 & 0.754 & 0.194 & 0.279  \tabularnewline
\hline 
190 & 14 & 142.118 & 270.243 & 0.189 & 0.307 & 18.905 & 29.880 & 1.440 & 1.417 & 0.199 & 0.269 \tabularnewline
200 & 14 & 279.372 & 746.858 & 0.174 & 0.268 & 27.408 & 45.009 & 2.572 & 3.053 & 0.156 & 0.195 \tabularnewline
210 & 14 & 43.664 & 80.953 & 0.140 & 0.221 & 8.050 & 14.041 & 0.701 & 0.779 & 0.350 & 0.514 \tabularnewline
220 & 14 & 83.293 & 155.135 & 0.175 & 0.249 & 13.349 & 21.340 & 1.015 & 1.183 & 0.221 & 0.273 \tabularnewline
230 & 14 & 86.514 & 148.774 & 0.138 & 0.205 & 13.158 & 19.912 & 1.324 & 1.392 & 0.151 & 0.204 \tabularnewline
240 & 14 & 110.856 & 277.187 & 0.152 & 0.220 & 15.747 & 28.341 & 1.485 & 2.042 & 0.161 & 0.180 \tabularnewline
250 & 14 & 71.984 & 125.540 & 0.123 & 0.196 & 11.423 & 17.904 & 1.271 & 1.280 & 0.187 & 0.320  \tabularnewline
\hline 
190 & 16 & 132.600 & 247.545 & 0.261 & 0.398 & 19.283 & 30.601 & 0.990 & 1.029 & 0.222 & 0.299 \tabularnewline
200 & 16 & 169.264 & 348.527 & 0.197 & 0.299 & 21.193 & 33.532 & 1.583 & 1.733 & 0.187 & 0.256 \tabularnewline
210 & 16 & 153.224 & 305.587 & 0.311 & 0.410 & 22.454 & 34.850 & 0.934 & 1.158 & 0.387 & 0.468 \tabularnewline
220 & 16 & 79.892 & 175.559 & 0.171 & 0.274 & 12.928 & 23.485 & 0.998 & 1.177 & 0.304 & 0.425 \tabularnewline
230 & 16 & 18.357 & 38.188 & 0.206 & 0.293 & 5.332 & 10.419 & 0.190 & 0.294 & 0.439 & 0.544 \tabularnewline
240 & 16 & 87.000 & 169.719 & 0.187 & 0.285 & 13.969 & 23.350 & 0.988 & 1.108 & 0.217 & 0.321 \tabularnewline
250 & 16 & 52.808 & 105.563 & 0.170 & 0.251 & 9.889 & 17.552 & 0.695 & 0.869 & 0.373 & 0.468  \tabularnewline
\hline 
\end{tabular}
\par\end{centering}
\centering{}\caption{Annualised statistics of sub-interval [2016-02-26, 2021-05-05] of myopic wealth and optimal wealth for constrained portfolio to be market neutral with $\gamma=-70$. Interest rate is $r=1\%$. Factor number is 6. Among all the different testing-training window combinations, the statistics for the S\&P 500 ETF, which has the largest profit are profit of $135.484\%$, volatility of $0.185$, expected return of $18.257\%$, Sharpe ratio of $1.405$, and maximum drawdown of $0.337$.}\label{tab:statistics_c_4}
%\end{table}
\end{sidewaystable}

Implementation of these optimal portfolios involves the selection of five hyper-parameters: risk aversion coefficient $\gamma$, number of factors $m$, number of co-integrated stocks $d$, in-sample training window length, and out-of-sample testing window length in each testing and training window. Tables \ref{tab:statistics_u} to Table \ref{tab:statistics_c_4} show considerable variations in portfolio performances as the sliding window lengths change, which generally indicates that determining hyper-parameters appropriately is challenging. In our studies, we prearrange five sets for the hyper-parameters, which contain different values for each, and then evaluate the performances of all hyper-parameter combinations. The indication is that selection of these parameters, prior to real-time trading, is a considerable challenge for practical implementation of these portfolios. In order to achieve better portfolio performance, for every in-sample training and out-of-sample testing window, these hyper-parameters should be selected dynamically, for example utilising the $R^2$ screening method of \citet{YeoPapanicolaou2017}. 

Finally, we should mention the effects of transactions costs and liquidity, which we have not included in these sliding window backtests because our analyses have focused on the robustness of the multiple co-integrated models that we propose with respect to hyper-parameters and the comparisons of performance across different time periods. However, these backtests can rerun with a five or ten basis points penalty on trading and the results will still be reasonable from the point-of-view of a practitioner. Liquidity is a deeper issue because brokers put considerable restraint on traders' leverage and short sales during times of crisis, which means potential returns from statistical arbitrage strategies in 2008 and in 2020 will require a broker who allows trading while most other brokers are freezing their clients' trades. To summarise, backtesting of statistical arbitrage based on historical data, such as we present in this article, are an informative feasibility study, but considerations given to transactions costs and liquidity are essential for practical implementation.

%%%%%%%%%%%%%%%%%%%%%%%%%%%%%%%%%%%%%%%%%%%%%%%%%%%%%%%%%%%%%%
\section{Conclusion}
\label{sec:conclusion}

In this article, we have presented a multiple stocks model for co-integration that can be utilised to analyse statistical arbitrage opportunities. We compute optimal portfolios for both an unconstrained stochastic optimal control problem and a market-neutral-constrained stochastic optimal control problem. The optimal value functions are the solutions to Hamilton-Jacobi-Bellman equations, which for power utility function, can be solved with an exponential ansatz that leads to a system of ordinary differential equations. This system consists of a matrix Riccati equation and two first-order linear ordinary differential equations. We have presented the stability analyses for the solutions to these ordinary differential equations for both the unconstrained and constrained cases. We then apply the optimal formulae to historical data and estimate the model parameters, and then perform sliding window backtests utilising the optimal portfolios. The results of these backtests indicate that statistical arbitrage portfolios have profit-generating potential during periods of higher overall market volatility, but are sensitive to parameter estimation. In particular, we show that profits and Sharpe ratios can be quite good but will vary as we change the lengths of the sliding windows for in-sample training and out-of-sample testing. The conclusion drawn from these backtests is that a priori selection of window-length parameters will be a significant challenge for implementation of statistical arbitrage in practice. %\textcolor{red}{???do we need to add a sentence or two about potential topics for future research???}

%%%%%%%%%%%%%%%%%%%%%%%%%%%%%%%%%%%%%%%%%%%%%%%%%%%%%%%%%%%%%%%%%%%%%%%%%%%%%%
\FloatBarrier
\renewcommand\refname{References}
\phantomsection
\addcontentsline{toc}{section}{\refname}% Add references to ToC (and bookmarks)
\bibliography{bibliography.bib}

\end{document}